\documentclass[aps,prb,twocolumn,superscriptaddress,floatfix,notitlepage]{revtex4-1}

\usepackage{amsthm}
\newtheorem{theorem}{Theorem}
\newtheorem{lemma}{Lemma}

\newtheorem{corollary}{Corollary}

\usepackage{graphicx,psfrag,color}
\usepackage{mathbbol,amsmath,amsfonts,amssymb,bm,dsfont,braket}
\usepackage[mathscr]{euscript}
\usepackage{wasysym}
\usepackage[section]{placeins}
\usepackage{multirow}

\def\k{{\mathbf{k}}}

\def\mathbi#1{\textbf{\em #1}}
\def\id{{\mathds{1}}}

\usepackage{verbatim}

\usepackage{array}
\newcolumntype{x}[1]{>{\centering\hspace{0pt}}p{#1}}
\newcolumntype{L}[1]{>{\raggedright\arraybackslash}p{#1}}
\newcolumntype{C}[1]{>{\centering\arraybackslash}p{#1}}
\newcolumntype{R}[1]{>{\raggedleft\arraybackslash}p{#1}}
\usepackage[normalem]{ulem}

\begin{document}

\title{Squaring the fermion: The threefold way and the fate of zero modes}

\begin{abstract}
We investigate topological properties and classification of mean-field theories of stable bosonic systems. Of the three standard classifying symmetries, only time-reversal represents a real symmetry of the many-boson system, while the other two, particle-hole and chiral, are simply constraints that manifest as symmetries of the effective single-particle problem. For gapped systems in arbitrary space dimension we establish {\it three fundamental no-go theorems that prove the absence of: parity switches, symmetry-protected-topological quantum phases, and localized bosonic zero modes under open boundary conditions}. We then introduce a squaring, kernel-preserving map connecting non-interacting Hermitian theories of fermions and stable boson systems, which serves as a playground to reveal the role of topology in bosonic phases and their localized midgap boundary modes. Finally, we determine the symmetry classes inherited from the fermionic tenfold-way classification, unveiling an elegant threefold-way topological classification of non-interacting bosons. We illustrate our main findings in one- and two-dimensional bosonic lattice and field-theory models. 
\end{abstract}

\author{Qiao-Ru Xu}
\affiliation{\mbox{Department of Physics, Indiana University, Bloomington, Indiana 47405, USA}}

\author{Vincent P. Flynn}
\affiliation{\mbox{Department of Physics and Astronomy, Dartmouth College, 6127 Wilder Laboratory, Hanover, New Hampshire 03755, USA}} 

\author{Abhijeet Alase}
\affiliation{Institute for Quantum Science and Technology, and Department of Physics and Astronomy, \\
University of Calgary, Calgary, AB, T2N 1N4, Canada}

\author{Emilio Cobanera} 
\affiliation{\mbox{Department of Mathematics and Physics, SUNY Polytechnic Institute, 100 Seymour Rd, Utica, NY 13502, USA }}
\affiliation{\mbox{Department of Physics and Astronomy, Dartmouth College, 6127 Wilder Laboratory, Hanover, New Hampshire 03755, USA}} 

\author{Lorenza Viola}
\affiliation{\mbox{Department of Physics and Astronomy, Dartmouth College, 6127 Wilder Laboratory, Hanover, New Hampshire 03755, USA}} 

\author{Gerardo Ortiz} 
\affiliation{\mbox{Department of Physics, Indiana University, Bloomington, Indiana 47405, USA}}
\affiliation{Indiana University Quantum Science and Engineering Center, Bloomington, IN 47408, USA}

\date{\today}
\maketitle


\section{Introduction}

Weakly interacting many-body systems of fermions or bosons can be described approximately by an effectively non-interacting (mean-field) theory as long as their 
equilibrium states 
are adiabatically connected \cite{Gell-Mann51}, 
and no phase transition separates them. This is the essence of Landau's quasiparticle framework, where the symmetries of the system define the principles behind matter organization and its 
 elementary excitations. Recently, another organizing principle, linked to topology, was recognized as fundamental to characterize hidden non-local order and the resilience of  localized excitations against local perturbations. The ``tenfold way" or ``topological classification"\cite{Schnyder,Kitaev2009,Ryu} of mean-field 
(free-)fermion Hamiltonians asserts that one cannot adiabatically connect, while preserving certain classifying symmetries, topologically inequivalent gapped systems.  For \emph{stable} free-boson Hamiltonians, which (like their fermionic counterparts) are bounded from below, 
what is the equivalent result? 
We address this and related questions in full generality. 

One of the most remarkable consequences of the topological classification of free-fermion systems is the bulk-boundary correspondence\cite{Chiu,Alase20}. When a gapped free-fermion system is not adiabatically connected to a topologically trivial free-fermion system, the obstruction to the deformation is diagnosed by a bulk topological invariant assuming different values for the two phases. The bulk-boundary correspondence relates the value of this invariant to the number and properties of midgap states (in one dimension) or surface bands of the systems subject to open boundary conditions (BCs). In every space dimension $d$, there 
are precisely five classes of systems for which the bulk-boundary correspondence specifically predicts zero modes (ZMs). Those ZMs often show remarkable localization properties. The topologically mandated Majorana ZMs of 
superconductors\cite{Read00,Kitaev2001,Deng2012,Deng2014}, in particular,  are a source of endless fascination\cite{Alicea, Beenakker}.   
Are there topologically mandated bosonic ZMs? What are the algebraic, localization, and stability properties of bosonic ZMs? In this paper we investigate these issues in detail. Since much of what is known about fermionic ZMs was learned from the topological classification by way of the bulk-boundary correspondence, we 
systematically follow 
an analogous line of reasoning for bosons. The outcome of this analysis will be a series of no-go theorems. 

Similar to the fermionic case, our starting point is the identification of 
the classifying internal symmetries of free-boson systems. Of the three classifying conditions of the tenfold way (time reversal, particle-hole, and chiral), we show that \emph{only time reversal can be related to a (many-body) symmetry} of the free-boson system. In contrast to fermions\cite{Chiu}, for bosons the particle-hole and chiral classifying conditions cannot be associated to many-body symmetries.  Particle-conserving systems are effectively well-described by Hermitian single-particle matrices or operators that may belong to any of the ten symmetry classes of the tenfold way. However,
these free-boson systems Bose-condense and are generically gapless in the 
thermodynamic limit. Particle non-conserving free-boson systems, on the other hand, 
may display gapped phases. 
These systems are analyzed in terms of non-Hermitian 
effective  Bogoliubov-de Gennes (BdG) matrices or operators that satisfy a particle-hole
constraint\cite{Blaizot}. 
A symmetry of an ensemble of effective BdG matrices is a ``pseudo-unitary" matrix  (in a sense that will be made precise later) which commutes with every member of the ensemble. There is a special class of many-body linear symmetries that is in direct correspondence with pseudo-unitary symmetries. Many-body time-reversal symmetry descends 
into the product of a pseudo-unitary matrix and complex conjugation, and this product commutes with the effective BdG Hamiltonian. 

Knowing the gapped stable free-boson ensembles and the symmetries at play, we proceed to investigate symmetry-preserving adiabatic deformations. Our first result is a no-go theorem for boson parity switches. Fermionic parity can be odd or even depending on the topological nature of the state and BCs. In contrast, it is typically even for topologically trivial superfluid phases, regardless of BCs. Indeed, fermion parity switches can be used as indicators for topological transitions 
also in interacting particle-conserving fermionic systems\cite{Ortiz14, Ortiz16}. Similar to fermions, bosonic pairing terms break the symmetry of particle-conservation down to the symmetry of boson parity. However, unlike fermions, we will show that the \emph{boson parity can only be even} in the ground manifold of a gapped free-boson system. Our second no-go theorem shows explicitly that any two gapped free-boson systems are adiabatically connected, regardless of symmetry constraints. It is a ``no-go" result in the sense that it forbids non-trivial symmetry-protected topological (SPT) phases of free-boson systems. And finally, our third no-go theorem states roughly that, for open BCs, a gapped free-boson system cannot possibly host surface bands inside the gap around zero energy or midgap ZMs. For fermions, the most localized ZMs are Majorana (self-adjoint) operators each localized on opposite boundaries. 

Does our third no-go theorem mean that localized bosonic analogues of Majorana ZMs are forbidden altogether? Certainly, this complicates matters considerably because one naturally looks for examples in systems subject to open BCs and that approach is doomed to failure. Fortunately, the precept that the ``square of a fermion is a boson'' comes to our rescue. We present a kernel-preserving map between fermions and bosons that provides a systematic way to generate bosonic Majorana (self-adjoint) ZMs. The square of a fermionic BdG Hamiltonian can be naturally reinterpreted as a bosonic (non-Hermitian) effective BdG Hamiltonian. This mapping does not preserve any spectral 
properties other than its kernel. Moreover, it allows a topological classification of ``squared ensembles'', leading to the threefold way of stable free boson systems away from zero energy.

Following our squaring-the-fermion map, we can construct a wealth of examples of bosonic Majorana ZMs by taking the square of a fermionic topological superconductor hosting Majorana ZMs. For example, we find that the square of the Kitaev chain hosts two exponentially localized (self-adjoint) ZMs and these modes can be normalized so that their commutator is equal to \(i\hbar\). In this sense, one can indeed ``split a single boson'' into two widely-separated halves. The consequences for the ground manifold of the system are, however, more dramatic for bosons than for fermions, because
an exact fermionic ZM implies two-fold degeneracy only, whereas an exact bosonic ZM implies macroscopic degeneracy if the number of particles is unconstrained.
How does the squaring map bypass our no-go result on ZMs? The answer is that to obtain ZMs in gapped free-boson systems, one must enforce BCs that are not open. The squaring procedure is a way to find both the required bulk and BCs. How robust are those ZMs? 
\textit{Krein stability theory}\cite{Schulz,Peano,Vincent20} helps us to rigorously 
address this question and conclude that bosonic Majorana ZMs are as exotic as they 
are fragile, something in complete agreement with our previous no-go theorems. 
Perhaps one could have intuited this since particle-hole and chiral symmetries are not many-body features of free-boson systems. 

The paper is organized as follows. Section\,\ref{mapping} covers the background, 
including a little-known necessary and sufficient condition 
for a free-boson system with pairing to be stable and a theorem\cite{ColpaZMs} that completely characterizes general bosonic ZMs, both {\it canonical} and {\it free-particle-like}. 
Since bosonic ZMs are central to this paper, we provide a self-contained proof of this theorem in Appendix \ref{proofs}, by using modern tools from indefinite linear algebra\cite{Gohberg}. 
In Sec. \ref{internalsymmetries} we discuss 
the many-boson underpinnings of the Altland-Zirnbauer (AZ) classifying conditions
(time-reversal, particle-hole, and chiral)\cite{AZ} and conclude that only time reversal
is associated to a many-body symmetry. In addition, from a many-body perspective, 
non-Hermitian ensembles of effective BdG Hamiltonians should be symmetry-reduced 
with respect to groups of pseudo-unitary matrices (details of the symmetry-reduction analysis 
can be found in Appendix \ref{afterandbefore}). 
Section \ref{zeromodes} is devoted to our three no-go theorems for gapped
free-boson systems: no parity switches, no SPT phases, and no localized ZMs. 
In Sec.\,\ref{classification} we introduce the squaring map from fermionic BdG Hamiltonians to bosonic effective BdG Hamiltonians and investigate it from the point of view of ensembles, symmetry and topological 
classifications, and bosonic topological invariants. Appendix \ref{simpleproof} presents a simple proof of the validity of 
our threefold way classification to general stable bosonic ensembles. Finally, 
in Sec.\,\ref{examples} we address the fate of bosonic Majorana ZMs in terms of examples obtained by the squaring map and discuss their stability. 
We close 
in Sec.\,\ref{summary} with a summary and
comments on the problem of characterizing SPT phases of interacting bosonic systems.

\section{Background}
\label{mapping}

\subsection{Free particles in second quantization}
\label{mapping2}

\resizebox{0.96\hsize}{!}{Consider first a general quadratic fermionic Hamiltonian} 
\begin{eqnarray}
\widehat{H}_f=\sum_{i=1}^N\sum_{j=1}^N
\big(K_{ij}c_i^\dagger c^{\;}_j+
\frac{1}{2}\Delta_{ij}c_i^\dagger c_j^\dagger+
\frac{1}{2}\Delta_{ij}^*c_jc_i\big),
\end{eqnarray}
for a system with $N$ single-particle states. 
The creation and annihilation operators 
$c_i^\dagger$ and $c^{\;}_i$ satisfy canonical anticommutation relations. 
Since, in addition, 
$\widehat{H}_f=\widehat{H}_f^\dagger$, it follows that $K=K^\dagger$ and 
$\Delta=-\Delta^\text{T}$. In terms of the Nambu array $\hat{\Psi}=\begin{bmatrix} c \\ c^\dagger\end{bmatrix}$ 
($\hat{\Psi}^\dagger=\begin{bmatrix} c^\dagger & c \end{bmatrix}$), with $\hat{\Psi}_i=c^{\;}_i$ and $\hat{\Psi}_{N+i}=c^\dagger_i$, $i=1,\ldots,N$, 
one can rewrite $\widehat{H}_f$ as
\begin{eqnarray}
\label{fermionicBdG}
\widehat{H}_f=\frac{1}{2}\hat{\Psi}^\dagger H_f\hat{\Psi}+\frac{1}{2}\mbox{tr}(K),
\end{eqnarray}
where $H_f=
\begin{bmatrix}
K	&\Delta\\
-\Delta^*	&-K^*
\end{bmatrix}$ is the (Hermitian) BdG Hamiltonian. Let $\tau_1\equiv\sigma_1\otimes\mathds{1}_N$, with \(\sigma_1=\begin{bmatrix} 0 & 1 \\ 1 & 0 \end{bmatrix}\). Because the Nambu array satifies the particle-hole constraint \(\hat\Psi=\tau_1\hat\Psi^{\dagger\, \text{T}}\), one finds that the BdG Hamiltonian satisfies the particle-hole constraint \(\tau_1H_f^*\tau_1=-H_f\).

The diagonalization of $H_f$ implies\cite{Blaizot} that of $\widehat{H}_f$.
If the pairing contributions vanish, \(\Delta=0\), one can rewrite 
\begin{eqnarray}
\widehat{H}_f=\widehat{K}_f=\hat{\psi}^\dagger K \hat{\psi}
\end{eqnarray}
in terms of \(\hat\psi^\dagger =\begin{bmatrix} c_1^\dagger & \cdots & c_N^\dagger\end{bmatrix}\)
and the associated column array \(\hat \psi\) of annihilation operators. Because these arrays are independent, the single-particle 
Hamiltonian \(K\), as opposed to the BdG Hamiltonian $H_f$, does not satisfy any constraints other than Hermiticity. Again, the diagonalization of \(K\) implies that of \(\widehat{K}_f\).

\resizebox{0.96\hsize}{!}{Next, consider a general quadratic bosonic Hamiltonian} 
\begin{eqnarray}
\label{freeboson}
\widehat{H}_b=\sum_{i=1}^N\sum_{j=1}^N
\big(K_{ij}a_i^\dagger a^{\;}_j+
\frac{1}{2}\Delta_{ij}a_i^\dagger a_j^\dagger+
\frac{1}{2}\Delta_{ij}^*a_ja_i\big),
\end{eqnarray}
where the bosonic operators $a_i^\dagger$ and $a^{\;}_i$ satisfy canonical 
commutation relations $[a^{\;}_i,a^{\dagger}_j]=\delta_{ij}$, $[a_i,a_j]=0$. 
In addition, since $\widehat{H}_b$ is Hermitian, it follows that 
$K=K^\dagger$ and $\Delta=\Delta^\text{T}$. Rewriting $\widehat{H}_b$ in terms of the Nambu array $\hat{\Phi}=\begin{bmatrix} a \\ a^\dagger\end{bmatrix}$ 
($\hat{\Phi}^\dagger=\begin{bmatrix} a^\dagger & a \end{bmatrix}$),  with $\hat{\Phi}_i=a^{\;}_i$ 
and $\hat{\Phi}_{N+i}=a^\dagger_i$, $i=1,\ldots,N$, we have 
\begin{eqnarray}
\label{bosonicBdG}
\widehat{H}_b=\frac{1}{2}\hat{\Phi}^\dagger H_b\hat{\Phi}-\frac{1}{2}\mbox{tr}(K),
\end{eqnarray}
with $H_b=
\begin{bmatrix}
K	&\Delta\\
\Delta^*	&K^*
\end{bmatrix}$ a Hermitian matrix, and 
 \([\hat{\Phi}_i,\hat{\Phi}^\dagger_j]=(\tau_3)_{ij}\), where $\tau_3\equiv\sigma_3\otimes\mathds{1}_N$, with 
$\sigma_3=\begin{bmatrix} 1 & 0 \\ 0 & -1 \end{bmatrix}$. 
Unlike the fermionic case, the diagonalization 
of \(H_b\) does not imply that of \(\widehat{H}_b\) in general. To diagonalize the bosonic many-body Hamiltonian, one must instead diagonalize, or at least put in Jordan normal form, the following non-Hermitian \textit{effective BdG Hamiltonian}\cite{colpa1, ColpaZMs, Blaizot}:
\begin{eqnarray}
\label{eBdG}
H_\tau \equiv \tau_3 H_b ,
\end{eqnarray}
which controls the dynamics of $\hat{\Phi}$. 
The effective BdG Hamiltonian satisfies the particle-hole constraint 
$\tau_1 H_\tau^*\tau_1=-H_\tau$ because of the constraint \(\hat \Phi=\tau_1\hat \Phi^{\dagger\,\text{T}}\).

Suppose 
that \(H_b\) is positive-definite, which we indicate from now on as $H_b>0$, 
so that \(H_\tau\) is both invertible and diagonalizable\cite{Blaizot} (see also Sec.\, \ref{theiff}).  Let $\ket{\psi_n^+}$ be an eigenvector of $H_\tau$ 
corresponding to a positive eigenvalue $\epsilon_n$, with $0<\epsilon_1<\epsilon_2<\cdots<\epsilon_N$. 
Then, $\tau_1 \mathcal{K} \ket{\psi_n^+}\equiv \ket{\psi_n^-}$,  
with $\mathcal{K}$ denoting complex conjugation, is an eigenvector 
corresponding to the negative eigenvalue $-\epsilon_n$ because of the 
particle-hole constraint. As shown in Ref.\,[\onlinecite{Blaizot}], 
these eigenvectors can be normalized to satisfy the following orthonormality relations:
\begin{eqnarray}
\label{orthonormality}
\braket{\psi_m^\pm|\tau_3|\psi_n^\pm} = \pm \delta_{mn}, \quad \braket{\psi_m^\pm|\tau_3|\psi_n^\mp} = 0,
\end{eqnarray}
and to construct the completeness 
relation
\begin{eqnarray}
\label{closure}
\mathds{1}_{2N}=
\sum_{n=1}^N\big(\ket{\psi_n^+}\bra{\psi_n^+}-\ket{\psi_n^-}\bra{\psi_n^-}\big) \tau_3 ,
\end{eqnarray}
with quasiparticle (Bogoliubov) operators corresponding to $\epsilon_n$ and $-\epsilon_n$ defined as 
\begin{eqnarray}\label{quasi-particle}
\left\{
\begin{aligned}
&b^{\;}_n=\bra{\psi_n^+}\tau_3\hat{\Phi}\\
&b_n^\dagger=\hat{\Phi}^\dagger\tau_3 \ket{\psi_n^+}
\end{aligned}
\right.,\quad
\left\{
\begin{aligned}
&b^{\;}_{-n}=\bra{\psi_n^-}\tau_3\hat{\Phi}\\
&b_{-n}^\dagger=\hat{\Phi}^\dagger\tau_3 \ket{\psi_n^-}
\end{aligned}
\right..
\end{eqnarray}

Immediately, from Eq.\,(\ref{orthonormality}), one can check that $(b^{\;}_n, b_n^\dagger)$ satisfy canonical commutation relations $[b^{\;}_m,b^{\dagger}_n]=\delta_{mn}$ and $[b_m,b_n]=0$, whereas
$b^{\;}_{-n}$ and $b_{-n}^\dagger$ do not (e.g., $[b^{\;}_{-m},b^{\dagger}_{-n}]=-\delta_{mn}$).
Taking advantage of Eq.\,(\ref{closure}), we can rewrite Eq.\,(\ref{bosonicBdG}) in 
terms of quasiparticle operators,
\begin{eqnarray}
\widehat{H}_b=\frac{1}{2}\sum_{n=1}^N\epsilon_n\big(b_n^\dagger b^{\;}_n+b_{-n}^\dagger b^{\;}_{-n}\big)-\frac{1}{2}\mbox{tr}(K),
\end{eqnarray}
where only $b_n^\dagger b^{\;}_n$ corresponds to the bosonic particle number operator, and we need to rewrite $b_{-n}^\dagger b^{\;}_{-n}$ in terms of the bosonic one using the relation $b^{\;}_{-n}=-b_n^\dagger$. Finally, we arrive at the bosonic quasiparticle Hamiltonian
\begin{eqnarray}
\label{bosonic quasi-particle}
\widehat{H}_b=\sum_{n=1}^N\epsilon_n b_n^\dagger b^{\;}_n-\sum_{n=1}^N\epsilon_n\braket{\psi_n^\circ|\psi_n^\circ},
\end{eqnarray}
where $\ket{\psi_n^\circ}\equiv \frac{1}{2}(\mathds{1}_{2N}-\tau_3)\ket{\psi_n^+}$. 
Accordingly, we see that excitation energies are always positive and the vacuum (ground state) of this bosonic Hamiltonian is a state with no quasiparticles.

The problem simplifies considerably if \(\Delta=0\). Then one can rewrite 
\begin{eqnarray}
\label{bk}
\widehat{H}_b=\widehat{K}_b=\hat{\phi}^\dagger K \hat{\phi} ,
\end{eqnarray}
in terms of \(\hat \phi^\dagger =\begin{bmatrix} a_1^\dagger & \cdots & a_N^\dagger\end{bmatrix}\)
and the associated column array \(\hat \phi\) of 
annihilation operators. As before, since these arrays are independent, 
the single-particle Hamiltonian \(K\), as opposed to the effective BdG Hamiltonian $H_\tau$, does not satisfy any constraints other 
than Hermiticity. Again, the diagonalization of \(K\) 
implies that of \(\widehat{K}_b\).

Back to the general case, including pairing, under a permutation (unitary) transformation $\pi^\dagger H_b\pi$ with permutation matrix  
$\pi_{ij}=\delta_{i,(j+1)/2}+\delta_{i,j/2+N}$, the 
Hermitian matrix $H_b$ can be rewritten as a $N\times N$ block-matrix with the matrix elements $(\pi^\dagger H_b\pi)_{ij}=
\begin{bmatrix}
K_{ij}	&\Delta_{ij}\\
\Delta_{ij}^*	&K_{ij}^*
\end{bmatrix}$. Then, the diagonalization of $\pi^\dagger H_b\pi$ with respect 
to the metric $\pi^\dagger \tau_3\pi=\mathds{1}_N\otimes\sigma_3$ implies that 
of $\widehat{H}_b$. For systems in which translation symmetry may be broken only by BCs, this reordering makes it possible to leverage a block-Toeplitz 
formalism \cite{AlasePRL,PRB1,JPA}, which we will take advantage of in Sec.\,\ref{nogo} and Sec.\,\ref{1DKitaev} in order to analytically determine  
closed-form solutions in limiting cases. 
Hereinafter, we will use $\tau_3$ to denote either the metric $\sigma_3\otimes\mathds{1}_N$ or $\mathds{1}_N\otimes\sigma_3$ depending on the formalism in use (see Table\,\ref{tab:metrics}).

\begin{table}[]
    \centering
    \caption{Different metrics in the Nambu and the block-Toeplitz formalisms\cite{AlasePRL,PRB1,JPA}.}
    \begin{tabular}{cccccccc}
		\hline		\hline
		& \multicolumn{3}{c}{$\quad\,$Nambu formalism$\quad\,$} & \multicolumn{4}{c}{$\quad$Block-Toeplitz formalism$\quad$} \\
		\hline
		$\quad\tau_1$ & \multicolumn{3}{c}{$\sigma_1\otimes\mathds{1}_N\quad$} & \multicolumn{4}{c}{$\mathds{1}_N\otimes\sigma_1$} \\
		$\quad\tau_3$ & \multicolumn{3}{c}{$\sigma_3\otimes\mathds{1}_N\quad$} & \multicolumn{4}{c}{$\mathds{1}_N\otimes\sigma_3$} \\
		\hline		\hline
	\end{tabular}
    \label{tab:metrics}
\end{table}

A fully translation invariant system on a \(d\)-dimensional Bravais lattice can be described in terms of an effective Bloch-BdG Hamiltonian satisfying the particle-hole constraint \(\tau_1H_\tau^*(\mathbf{-k})\tau_1=-H_\tau(\mathbf{k})\), 
where \(\mathbf{k}=(k_1,\cdots,k_d)\) denotes a $d$-dimensional crystal momentum vector in the Brillouin zone. 
We will also consider systems on $d$-dimensional lattices that suddenly stop at a flat $(d-1)$-dimensional hypersurface. Such terminations are called ideal surfaces \cite{Bechstedt,Lee}. In this setup, the system is half-infinite in one of the \(d\) directions, and retains translation symmetry in the remaining \(d-1\) directions.  The associated effective BdG Hamiltonians will be denoted as  $H_\tau^{o}$, where the superscript ``$o$'' stands for {\it open} BCs for the termination (see Sec.\,IIA of Ref.\,[\onlinecite{Cobanera18}] for a detailed discussion).  It is advantageous to introduce the quantum number $\mathbf{k}_\parallel$, the crystal momentum in the surface Brillouin zone (SBZ)\cite{Bechstedt}. Then, the system is  described by an effective Bloch-BdG Hamiltonian of the form $H^o_{\tau,\mathbf{k}_\parallel}= \tau_3H^o_{b,\mathbf{k}_\parallel}$. For a fixed \(\mathbf{k}_\parallel\), the matrix $H^o_{\tau,\mathbf{k}_\parallel}$ can be visualized as describing a half-infinite ``virtual chain'' system; 
notice, however, that such a system satisfies 
$\tau_1(H^o_{\tau,\mathbf{k}_\parallel})^*\tau_1 = -H^o_{\tau,-\mathbf{k}_\parallel} 
\ne  -H^o_{\tau,\mathbf{k}_\parallel}$ in general, which is different from the usual particle-hole constraint for a one-dimensional system. If we change the BCs of these virtual chains back to periodic BCs, then we can describe the chains in terms of the one-dimensional crystal momentum $k \in [-\pi,\pi)$ in units of the reciprocal stacking period and matrices \(H_{\tau,\mathbf{k}_\parallel}(k)\). 
Naturally, we have \(H_{\tau,\mathbf{k}_\parallel}(k) =H_{\tau}(\mathbf{k})\). 

\subsection{Single-particle characterization of stability}
\label{theiff}

The condition of Hamiltonian stability plays essentially no role
for free-fermion {lattice} systems because of the Pauli exclusion principle. 
By contrast, quadratic bosonic Hamiltonians 
may fail to be ``stable'' in different ways, even for a finite number
of modes \(N\). We recall here a single-particle characterization
of stability 
for free-boson systems\cite{Derezinski17} upon which we 
will rely heavily in Sec.\,\ref{zeromodes}.  

For free-boson systems there are in fact 
different notions of stability of practical importance\cite{Peano,Vincent20}. One notion, which is 
the usual notion of stability in quantum mechanics, is the condition 
that \(\widehat{H}_b\) should be bounded below, and applies to particle-conserving and non-conserving systems alike. The following theorem identifies the necessary and sufficient condition for this form of stability:

\smallskip

\noindent
\textbf{Theorem ([\onlinecite{Derezinski17}]).}
\textit{The quadratic bosonic Hamiltonian 
\(\widehat{H}_b\) is stable if and only if the Hermitian
matrix \(H_b\) is positive semi-definite. }

\smallskip

A weaker notion of stability is meaningful only 
for particle-conserving systems: in such a case, it may happen that \(\widehat{H}_b=\widehat{K}_b\) is bounded from below in any subspace with a fixed number of particles but not over the full Fock space. Stability in the usual sense is then achieved if and only if $K$ is positive semi-definite. This condition will be indicated as \(K\geq 0\) from now on. 

It is interesting to note the following related result:
\smallskip

\noindent
\textit{
If \(K\) is not positive semi-definite, then neither is \(H_b\), regardless of the properties of the pairing matrix \(\Delta\).}
\smallskip

\noindent 
This has implications for interacting particle-conserving systems. For suppose that one is interested in a weakly-interacting, particle-conserving boson system. Then, one may try a mean-field approximation that breaks particle conservation. But, if \(K\) is not positive semi-definite, this mean-field approximation will necessarily yield an unstable quadratic bosonic Hamiltonian. (See also Ref.\,[\onlinecite{Vincent20}] for a more general discussion, partly motivated by topology, of unstable free-boson systems.) 

One can also approach the notion of stability from the point of view of the response of the system to a classical forcing term.
For free boson systems, but not for free fermions, 
the simplest model one may consider is described by the linear-quadratic
Hamiltonian 
\begin{eqnarray}
\label{linear_quadratic}
\widehat{H}_{b,F}=\widehat{H}_b+\hat{\Phi}^\dagger F ,
\end{eqnarray}
where \(F=\begin{bmatrix}f_1&\cdots & f_N&f_1^*&\cdots & f_N^*\end{bmatrix}^\text{T}\) 
is a vector of complex parameters. 
If there is a solution \(Z=\begin{bmatrix}z_1&\dots & z_N&z_1^*&\dots & z_N^*\end{bmatrix}^\text{T}\) of the equation \(F=H_b Z\), then the 
Hamiltonian of Eq.\,\eqref{linear_quadratic} satisfies the relationship 
\begin{align}
\widehat{H}_b+\hat{\Phi}^\dagger F=U^{\;}_Z\widehat{H}_bU_Z^\dagger-
\frac{1}{2}Z^\dagger H_b Z , 
\end{align}
in terms of the unitary map \(U_Z=e^{\sum_{i=1}^N(z_i^*a_i-z_ia_i^\dagger)}\). 
For stable systems without ZMs (\(H_b>0\)), \(Z=H_b^{-1}F\). For stable systems with ZMs (\(H_b\geq0\)), \(Z\) may fail to exist.

\subsection{(Not so well) Known results on general bosonic zero modes}
\label{statistic}

To the best of our knowledge, the first complete characterization of ZMs of 
quadratic bosonic Hamiltonians with \(H_b\geq 0\) appeared in Ref.\,[\onlinecite{ColpaZMs}]. We recall this somewhat hidden result here for later use in the form of a Theorem and, as mentioned, include a new 
proof 
in Appendix \,\ref{proofs}.  
The problem is difficult for particle non-conserving systems because 
the normal modes and frequencies of the system are calculated from the effective 
BdG Hamiltonian $H_\tau=\tau_3 H_b$ defined in Eq. \eqref{eBdG} 
and this matrix is \textit{not} Hermitian in general.

The starting point of the analysis are several spectral properties
of the non-Hermitian matrix \(H_\tau\). First, because
$\tau_3 H_\tau^\dag \tau_3 = H_\tau$ (pseudo-Hermiticity) and 
$\tau_1 H_\tau^* \tau_1 = -H_\tau$ (particle-hole constraint), it 
follows that the eigenvalues of $H_\tau$ come in quartets 
$\{\epsilon, \epsilon^*, -\epsilon, -\epsilon^*\}$. 
Since in this paper we always assume that $H_b\geq 0$, it  
also follows that the eigenvalues of $H_\tau$ are purely real\cite{Blaizot}.
Finally, since $H_\tau$ is not a Hermitian matrix in general, it could 
fail to be diagonalizable. Again, this possibility is highly constrained
by the condition $H_b\geq 0$. Theorem 5.7.2 in Ref.\,[\onlinecite{Gohberg}] 
tells us that \(H_\tau\) is diagonalizable except possibly on the subspace of 
ZMs. The Jordan normal form of \(H_\tau\) restricted to this subspace
can contain Jordan blocks that are \emph{at most of size two} (\(2\times 2\) blocks).
Since there is an even number of eigenvectors associated to non-zero eigenvalues,
due to the particle-hole constraint, it follows that the algebraic multiplicity 
of the zero eigenvalue must be even. And, since the Jordan blocks
are at most of size two, the Jordan chains of length one come in pairs. 
\smallskip

\noindent\textbf{Theorem (ZMs, [\onlinecite{ColpaZMs}]).}
\textit{For the effective BdG Hamiltonian $H_\tau=\tau_3H_b$, let $2n$ and $m$ be the number of linearly independent zero eigenvectors associated to Jordan chains of length one and two, respectively. Then there are $n$ pairs of canonical boson $b^{\;}_{0,j},b_{0,j}^\dag$ 
that commute with the many-body Hamiltonian \(\widehat{H}_b\) 
and all other normal modes of the system. 
In addition, there exist $m$ pairs of Hermitian
operators $P_{0,j},Q_{0,j}$ that also commute with all other normal modes of the system and obey $[Q_{0,j},P_{0,\ell}]=i\delta_{j\ell}$, $[\widehat{H}_b,P_{0,j}]=0$, 
and $[\widehat{H}_b,Q_{0,j}]= (i/\mu_j)P_{0,j}$, with \(\mu_j>0\).}

\medskip

\noindent
It is crucial to point out that Goldstone modes of free-boson systems\cite{Blaizot} are a particular instance of the above theorem. Generically, Goldstone modes are bulk, {\it delocalized} modes. In contrast, our focus in this paper
is on {\it localized} bosonic ZMs. 

How robust are bosonic ZMs? There are some results in the literature 
regarding small perturbations\cite{Schulz,Peano,Yaku}. In practice,
it can be hard to decide whether a small perturbation preserves the 
stability condition \(H_b\geq 0\), thus it is possible for a ZM to split 
away from zero into the complex plane. In this case one says that the 
mode has become \textit{dynamically unstable}\cite{Peano,Vincent20}. From a 
dynamical perspective, the loss of diagonalizability is also regarded as a 
dynamical instability even if the system satisfies \(H_b\geq 0\). 
Fortunately, these additional complications (as compared to the Hermitian case)
come paired with an additional theoretical tool: the Krein stability 
theory of dynamical systems in an indefinite inner-product 
space\cite{Schulz}. 

We recall some essential definitions. Given a vector $\ket{v}\in\mathbb{C}^{2N}$, 
the sign of the $\tau_3$-norm $\braket{v|\tau_3|v}$ is its \emph{Krein signature}. If $\braket{v|\tau_3|v}=0$, we say the Krein signature is $0$. Let $\lambda$ be an eigenvalue of a matrix that is Hermitian with respect to the indefinite $\tau_3$ inner product (e.g., $H_\tau$). If all eigenvectors associated with 
$\lambda$ have either a $+1$ or $-1$ Krein signature, then we say that $\lambda$ is $\pm$-definite. Otherwise we say $\lambda$ is indefinite. Note that $\lambda$ being $\pm$-definite requires that $\lambda\in\mathbb{R}$. A key result in the theory of Krein stability is the Krein-Gel'fand-Lidskii theorem\cite{Schulz}: 
\smallskip

\noindent
If $\lambda$ is a $\pm$-definite eigenvalue of a $\tau_3$-pseudo-Hermitian matrix $M$, then there exists an open neighborhood of $M$, such that all matrices in this neighborhood have eigenvalues close to $\lambda$ that remain on the real axis and have diagonal Jordan blocks.
\smallskip

\noindent
Physically, this means that the corresponding modes remain \emph{dynamically stable} (that is, they result in bounded evolution in time) 
under all sufficiently small perturbations. Furthermore, in any open neighborhood 
of a matrix $M$ with an eigenvalue that is either indefinite or has a non-diagonal 
Jordan block, there are matrices with eigenvalues off the real axis\cite{Yaku}.

Returning to the bosonic problem, recall that the kernel vectors of $H_\tau$ are 
either associated with canonical bosonic ZMs or free-particle-like ZMs (Hermitian quadratures). The canonical bosonic ZMs arise from pairs of eigenvectors 
with different Krein signatures, while the free-particle-like ZMs arise from 
eigenvectors with vanishing Krein signatures\cite{Blaizot}. 
Hence, according to Krein stability theory, there exist arbitrarily small 
perturbations that will cause these ZMs to become dynamically unstable. 
The stability properties of 
bosonic ZMs we have summarized here will illuminate some features of our prototype squared Kitaev chain model in Sec.\,\ref{1DKitaev}. 

\subsection{The tenfold way}
\label{fermary}

In this section, we summarize the basics of the classification of Hermitian ensembles 
(or equivalently, the classification of free-fermion SPT phases) and the bulk-boundary correspondence because these subjects will guide and inspire our investigation of free-boson systems. For this summary we especially benefited from the review 
articles in Ref.\,[\onlinecite{Zirnbauer10}] and Ref.\,[\onlinecite{Chiu}].
The classification of non-Hermitian ensembles\cite{BLC} will also play a role in Sec.\,\ref{classification}. Since this subject is still evolving fast, we refer the interested reader to Ref.\,[\onlinecite{Kawabata}] and references therein.

A Bloch Hamiltonian is a Hermitian-matrix-valued function \(H({\bf k})\). For fixed but
arbitrary \({\bf k}\), \(H({\bf k})\) is an operator acting on the Hilbert space 
\({\cal H}_{\sf int}\) of internal degrees of freedom (the same for all \({\bf k}\)). 
The symmetry classification  of Bloch Hamiltonians is a classification scheme 
for sets (``ensembles") of  \(H({\bf k})\)'s acting on a common Hilbert space 
\({\cal H}_{\sf int}\) and with a common argument \({\bf k}\). The dimensionality 
\(d\) of the ensemble is that of \(\k\). 
The group of symmetries of an ensemble is the group of isometries of
\({\cal H}_{\sf int}\) that commute with every \(H({\bf k})\) in the ensemble. 
An ensemble is called \textit{irreducible} if its group of symmetries consists only of the identity up to a phase. 

The simultaneous block-diagonalization of an ensemble and its unitary symmetries 
decomposes the ensemble into a sum of irreducible ensembles. An irreducible element \(H({\bf k})\) satisfies some subset of the following conditions,
\begin{align}
\nonumber
{\cal T}H({\bf k}){\cal T}^{-1}&=H(-{\bf k})\,\,\,\,\,\,\,({\cal T}=U^\dagger_T{\cal K},\,U^{\;}_TU_T^*=\pm \mathds{1}),\\
\label{DAZconditions}
{\cal C}H({\bf k}){\cal C}^{-1}&=-H(-{\bf k})\,\,\,({\cal C}=U^\dagger_C{\cal K},\,U^{\;}_CU_C^*=\pm \mathds{1}),\\
\nonumber 
{\cal S}H({\bf k}){\cal S}^{-1}&=-H({\bf k})\,\,\,\,\,\,\,\,({\cal S}=U^\dagger_S,\,U_S^2=\mathds{1}).
\end{align}
Here, \({\cal K}\) denotes complex conjugation in some preferred basis of 
\({\cal H}_{\sf int}\) and the linear isometries \(U^{\;}_T, U_C, U^{\;}_S\) of \({\cal H}_{\sf int}\)
do not depend on \(H({\bf k})\) but only on the specific irreducible ensemble. 
For \(d=0\) ensembles, that is, dropping the \({\bf k}\)-dependence, these conditions 
become the usual conditions of the AZ classification\cite{Dyson,AZ,Verbaarschot00}. There are precisely ten distinct combinations of these conditions, 
which motivates the name ``tenfold way"\cite{Ryu}. The standard names of the ten classes of irreducible Hermitian ensembles are shown in Table \ref{table:names}, with \{A, AIII\} being the complex classes while the other eight the real classes. 

The condition associated to \({\cal T}\) can be understood physically in terms of the time-reversal symmetry. Ignoring the other two conditions that do not have in general an obvious physical interpretation, one arrives at the ``threefold way" of Dyson\cite{Dyson}.  The other two conditions, dubbed particle-hole (or charge conjugation) and chiral ``symmetries" arise naturally for free fermions as ``descendants'' of special many-body symmetries\cite{AZ,Verbaarschot00}. We will come back to this point in more detail for bosons. The reminder of this subsection is focused on (single-particle or BdG) Bloch Hamiltonians associated to free-fermion systems.

Consider the following question: Given a choice of Fermi energy (zero energy for
superconductors) and two members, $H_1({\bf k})$ and $H_2({\bf k})$, of
an irreducible ensemble of Bloch Hamiltonians fully gapped at that Fermi energy, is it 
possible to find a continuous deformation \(H_s({\bf k})\) of \(H_1({\bf k})\) into \(H_2({\bf k})\) such that (1) \(H_s({\bf k})\) is fully gapped at the Fermi energy for all \(s\), and (2) \(H_s({\bf k})\) is a member of the same irreducible ensemble of Bloch Hamiltonians for all \(s\)? The answer can be ``Yes" or ``No" depending on the classifying parameters: the dimension \(d\) and the symmetry class.
If the answer is ``No", the obstruction to the deformation is characterized by a topological invariant that can be calculated directly for individual Hamiltonians. 
Bloch Hamiltonians that cannot be deformed into one another are distinguished by the value of some topological invariant. The pattern of ``Yes/No" and the topological invariants at play, both as a function of the symmetry class, are periodic in \(d\), with period two for complex classes and period eight for real classes. This remarkable result goes by a catchy name, the ``periodic table" (of the tenfold way), 
 which was established through Clifford algebras and $K$-theory in Ref.\,[\onlinecite{Kitaev2009}]. For example, when $d=0$, the two complex classes \{A, AIII\} are associated with two types of complex Clifford algebras and classifying spaces $\mathcal{C}_q, q=0,1$, as shown in Table
 \ref{table:names}, while the other eight real classes are associated with eight types of real Clifford algebras and classifying spaces $\mathcal{R}_q, \, q=0,1,\cdots,7$. Then, the number of disconnected components of $\mathcal{C}_q$ and $\mathcal{R}_q$ can be used as topological invariants to classify topological phases. Hereinafter, we will always refer to the classifying space of a symmetry class as the one associated with $d=0$. 

\begin{table}[]
\centering
\caption{Standard names for the ten symmetry classes of irreducible ensembles of Bloch Hamiltonians for all dimensions. 
The three symmetries $\mathcal{T}$, $\mathcal{C}$ and $\mathcal{S}$ are denoted by $1\,(-1)$ if they square to $\mathds{1}\,(-\mathds{1})$ and denoted by $0$ if they are absent. Classifying spaces (CS) associated with complex classes and real classes in $d=0$ are denoted by $\mathcal{C}_q, \,q=0,1$ and $\mathcal{R}_q,\, q=0,1,\cdots,7$, respectively.}
\smallskip
    \begin{tabular}{x{7mm}  x{4mm} x{7mm} x{3mm} x{6mm} x{7mm} x{5mm} x{7mm} x{6mm} x{6mm} x{5mm} x{6mm}}
    \hline\hline \\ [-2ex]
             \      &A &AIII & &AI &BDI &D &DIII &AII &CII &C &CI
    \tabularnewline [0.5ex] 
    \hline \\[-2ex]
    \({\cal T}\) & 0 & 0    & & 1  &  1  & 0 &  -1   & -1   & -1   & 0 & 1
    \tabularnewline [1ex]
    \({\cal C}\) & 0 & 0    & & 0  &  1  & 1 &  1   & 0   & -1   & -1 & -1
    \tabularnewline [1ex]
    \({\cal S}\) & 0 & 1    & & 0  &  1  & 0 &  1   & 0   & 1   & 0 & 1
    \tabularnewline [1ex]
      CS    & $\mathcal{C}_0$ & $\mathcal{C}_1$    & & $\mathcal{R}_0$  &  $\mathcal{R}_1$  & $\mathcal{R}_2$ &  $\mathcal{R}_3$   & $\mathcal{R}_4$   & $\mathcal{R}_5$   & $\mathcal{R}_6$ & $\mathcal{R}_7$
    \tabularnewline [0.5ex] \hline\hline
    \end{tabular}
\label{table:names}
\end{table}

As it turns out, it is possible in general to continuously deform two Bloch Hamiltonians 
into each other without closing the gap at the Fermi energy, provided that one is allowed to break the classifying symmetries at intermediate steps. In this sense, the topological distinction between Bloch Hamiltonians is ``symmetry-protected". As hinted by the role of the Fermi energy in this discussion, fermionic statistics lead to a strong connection between low-energy many-body physics and the predictions of the topological classification. 

First, for bulk systems, the topological classification of gapped Bloch Hamiltonians translates into a classification of SPT phases of free-fermion systems. Recall that two many-body ground states are regarded as describing distinct SPT phases if it is not possible to deform one into the other adiabatically 
(as in the Gell-Mann-Low theorem\cite{Gell-Mann51}) without closing the many-body 
gap -- while maintaining a preferred, ``protecting" set of symmetries at all steps of the 
deformation. For fermions, topologically distinct ensembles of gapped Bloch Hamiltonians
are in one-to-one correspondence with distinct SPT phases. This result is consistent with 
the fact that the classifying conditions of Eq.\,\eqref{DAZconditions} are in 
correspondence with many-body symmetries for fermions. For superconductors, 
SPT order is often signaled by the fermion parity observable. 

Second, the celebrated bulk-boundary correspondence relates non-zero 
values of the bulk topological invariants to the presence of boundary states of individual Hamiltonians subject to open BCs. For \(d>1\), these boundary or edge states cross the 
Fermi energy (zero energy for superconductors) and establish another link between
low-energy many-body physics and the topological classification, i.e., the
topologically dictated boundary metals/gapless superconductors that emerge at the 
termination of a fully gapped topologically non-trivial bulk. The Bloch Hamiltonians
for the integer quantum Hall effect are important examples; they form
a \(d=2\), class A ensemble. For \(d=1\), \{AIII, BDI, D, DIII, CII\} are the five classes where the bulk-boundary correspondence dictates midgap ZMs for topologically
non-trivial bulks. From a many-body perspective, the most remarkable example
is the Kitaev chain because, loosely speaking, it features a single fermion split into halves in terms of a pair of Majorana ZMs localized on opposite ends of the chain.

\section{Internal symmetries of many-boson systems}
\label{internalsymmetries}

In this section we investigate a special class of symmetries 
of free-boson systems that we call \emph{Gaussian}\cite{Cobanera}.
These symmetry transformations 
on Fock space are special because they map (by similarity)
creation and annihilation operators to linear combinations of themselves.
For fermions, the classifying conditions of the symmetry classification 
known as the tenfold way 
are in correspondence with many-body Gaussian symmetries. As we will see,
this is only partially true for bosons. The symmetry analysis of this section will 
be important in Sec.\,\ref{classification} when we consider non-Hermitian classification schemes as they apply to certain effective BdG Hamiltonians.   

\subsection{Particle-conserving systems}
\label{pcsec}

An ensemble of particle-conserving quadratic bosonic Hamiltonians is in one-to-one correspondence with an auxiliary ensemble of single-particle Hamiltonians by Eq.\,\eqref{bk}, 
regarded as a mapping. We will focus on 
symmetry transformations on Fock space that map particle-conserving quadratic bosonic Hamiltonians to Hamiltonians of the same type. There are two possibilities in principle: canonical mappings of the form
\begin{eqnarray}
\label{gaussianb}
\mathbi{V}\, \hat{\phi} \, \mathbi{V}^{\, -1}
=\begin{bmatrix} \mathbi{V}\, a_1 \mathbi{V}^{\, -1} \\ \vdots  \\ \mathbi{V}\, a_N \mathbi{V}^{\, -1}\end{bmatrix}
=U_{\! \mathbi{V}} \, \hat{\phi} ,
\end{eqnarray}
and of the form 
\begin{eqnarray}
\label{wouldbeb}
\mathbi{C} \, \hat{\phi} \, \mathbi{C}^{\, -1}=U_{\! \mathbi{C}}^* \, \hat{\phi}^{\dagger \text{T}} ,
\end{eqnarray}
where, because the commutation relations are necessarily preserved, $U_{\! \mathbi{V}}$ and $U_{\! \mathbi{C}}$ are $N\times N$ unitary matrices. The same is true for fermions with \(\hat \phi\) replaced by \(\hat \psi\).
Notice that if \(\mathbi{O}\) denotes an operator in Fock space, the equation \((\mathbi{V}\, \mathbi{O} \, \mathbi{V}^{\, -1})^\dagger=\mathbi{V}\, \mathbi{O}^\dagger \, \mathbi{V}^{\, -1}\) also holds for \(\mathbi{V}\) antilinear provided it is antiunitary. 

The two distinct possibilities above arise because of the requirement of particle conservation. For particle non-conserving systems, there is no obstruction
to mix creation and annihilation operators and so there is only one
kind of Gaussian map (see Eq.\,\eqref{gaussianb} below). This fact is closely related to the particle-hole constraint satisfied by Nambu arrays. 
For particle-conserving free fermions, both possibilities are realized and transformations 
of the second kind, conventionally assumed linear without loss of generality,
so that \(\mathbi{C}^{\, -1}=\mathbi{C}^{\, \dagger}\),
are called particle-hole (or charge conjugation) symmetries because 
the symmetry condition \(\mathbi{C}\, \widehat{K}_f \, \mathbi{C}^{\, \dagger}=\widehat{K}_f\) translates into a particle-hole-like condition \(U_{\! \mathbi{C}}^\dagger \, K^* \, U^{\;}_{\! \mathbi{C}}=-K\) for the single-particle Hamiltonian \(K\). As we will see, particle-hole and so-called chiral symmetries are forbidden for bosons by the canonical commutation relations. 

Going back to ensembles of particle-conserving free-boson systems, 
the unitary symmetries of the auxiliary ensemble of single-particle
Hamiltonians are in one-to-one correspondence with the Gaussian symmetries of the 
ensemble of many-body Hamiltonians. Suppose
\begin{eqnarray}
\label{pcg}
\mathbi{V} \, \hat\phi \, \mathbi{V}^{\, \dagger}= U_{\! \mathbi{V}} \, \hat \phi
\end{eqnarray}
is a unitary Gaussian symmetry transformation.
Then, the condition \(\mathbi{V}\, \widehat{K}_b \, \mathbi{V}^{\, \dagger}=\widehat{K}_b\) for all \(\widehat{K}_b\) implies that \(U_{\! \mathbi{V}}^\dagger \, K \, U^{\;}_{\! \mathbi{V}}=K\) for all \(K\). That is,
a Gaussian symmetry of the many-body ensemble descends into a symmetry of
the auxiliary ensemble of single-particle Hamiltonians. For the reverse
process of lifting a single-particle symmetry to a Gaussian many-body
symmetry, one needs to invoke the Stone-von Neuman-Mackey theorem\cite{rosenberg_j:qc2004a,Derezinski2006}, which guarantees that the unitary transformation \(\mathbi{V}\) in Eq.\,\eqref{pcg}
exists given \(U_{\! \mathbi{V}}\) as the input, provided that the number of modes \(N<\infty\). 

As a result of this analysis, one concludes that the decomposition of the auxiliary ensemble of single-particle Hamiltonians into irreducible AZ ensembles is equivalent to the decomposition of the ensemble of bosonic many-body Hamiltonians as a sum of commuting bosonic many-body Hamiltonians
labeled by Gaussian-symmetry quantum numbers. It is important to appreciate how particle 
conservation fits in this discussion.  
Let \(\mathbi{N}\equiv\sum_{j=1}^Na_j^\dagger a^{\,}_j\)
denote the number operator. Then, \(e^{i\theta \mathbi{N}}\, \hat \phi \,  e^{-i\theta \mathbi{N}}=e^{-i\theta} \hat\phi\) and thus particle conservation induces a trivial symmetry of the auxiliary ensemble of single-particle Hamiltonians. 

This aspect of the problem is identical for bosons and fermions with the
complication, in the case of bosons, that one must rely on the highly-non-trivial Stone-von Neumann-Mackey theorem to establish it (see also Theorem 11 in Ref.\,[\onlinecite{Derezinski2006}] for fermions, with $N<\infty)$. The key difference between particle-conserving fermions and bosons is that for bosons not all of the classifying AZ conditions are associated to Gaussian many-body symmetries. Time reversal works fine. 
Suppose \(\mathbi{T}\) is an antilinear isometry of the Fock space such that
\begin{align}
\mathbi{T} \, \hat \phi \, \mathbi{T}^{-1}=U^{\;}_T \, \hat\phi, \quad\quad \mathbi{T} \, i=-i\mathbi{T}.
\end{align}
Because commutation relations are preserved and by that \(U^{\;}_T\)
is a unitary matrix,
\(\mathbi{T} \, \widehat{K}_b \, \mathbi{T}^{-1}=\widehat{K}_b\)
is equivalent to \({\cal T} K {\cal T}^{-1}=K\) in terms of \({\cal T}=U_T^\dagger{\cal K}\). 
Conversely, given a unitary transformation that intertwines \(K\) and \(K^*\), one 
can lift it into an antilinear Gaussian symmetry by a variation of the argument
of the previous paragraph. 

How about particle-hole symmetries, Eq.\,\eqref{wouldbeb}? 
If such a \(\mathbi{C}\) were to exist, it would have to preserve the
commutation relations, which would imply that 
\(U^{\;}_\mathbi{\!C} \, U_\mathbi{\!C}^\dagger=-\mathds{1}_N\). There is no solution of this equation. Finally, a chiral symmetry is a Gaussian symmetry of the form \(\mathbi{S}=\mathbi{T}\mathbi{C}\), which is of interest for fermions in situations where \(\mathbi{T}\) and \(\mathbi{C}\) are not separately symmetries of the many-body ensemble but \(\mathbi{S}\) is. Since particle-hole symmetries do not exist for bosons, neither do chiral symmetries. The interested reader can play around with the idea of defining a chiral symmetry for bosons, say, as a sub-lattice symmetry, to gain physical insight into this no-go result.

In summary, for particle-conserving free-boson systems, the time-reversal classifying condition \({\cal T}\) is in correspondence with an antilinear Gaussian symmetry  \(\mathbi{T}\). By contrast, and contrary to fermions, the particle-hole and chiral classifying conditions can well emerge at the single-particle level but have \emph{no many-body counterpart}. But then, what exactly is accomplished by feeding a topologically non-trivial single-particle Hamiltonian into Eq.\,\eqref{bk}? 

\subsection{Particle non-conserving systems}
\label{pncsec}

For particle non-conserving free-boson systems, Eq.\,\eqref{bosonicBdG}, 
regarded as a map, puts ensembles of quadratic bosonic Hamiltonians in correspondence with ensembles of Hermitian matrices \(H_b\) that satisfy the constraint \(\tau_1H_b^*\tau_1=H_b\). However, the many-body system is governed by the spectral properties of the effective BdG Hamiltonian \(H_\tau=\tau_3H_b\). Moreover, as we will see next, Gaussian symmetries are in correspondence with symmetries of \(H_\tau\), not \(H_b\). Hence, one is drawn to the conclusion that ensembles of particle non-conserving free-boson systems should be regarded as being in correspondence with those of pseudo-Hermitian effective BdG Hamiltonians.   

The focus is again on symmetry transformations on Fock space that \textit{generically} map particle non-conserving quadratic bosonic Hamiltonians to Hamiltonians of the same type. Since there is no need to preserve particle conservation, nothing prevents mixing creation and annihilation operators and, thus, the notion of particle-hole symmetry becomes redundant. Gaussian symmetry transformations are of the form 
\begin{align}
\label{gaussianb}
\mathbi{V} \, \hat{\Phi} \, \mathbi{V}^{-1}=U_{\! \mathbi{V}} \, \hat{\Phi}
=U_{\! \mathbi{V}}' \, \hat \Phi^{\dagger\, \text{T}} , 
\end{align}
with \(\mathbi{V}\) linear or antilinear. Because canonical commutation relations are necessarily preserved, the above implies
\begin{align}
U^{\;}_{\! \mathbi{V}} \, \tau_3 \, U_{\! \mathbi{V}} {\!\,}^\dagger=\tau_3,\quad\quad
U'_{\! \mathbi{V}} \, \tau_3 \, U'_{\! \mathbi{V}} {\!\,}^\dagger
=-\tau_3,
\end{align}
with $U_{\!\mathbi{V}}$ a {\it pseudo-unitary} matrix, and \(U'_{\! \mathbi{V}}\,(=U_{\! \mathbi{V}} \, \tau_1)\) a {\it skew-pseudo-unitary} matrix. Hence, one can focus on the pseudo-unitary matrices \(U_{\! \mathbi{V}}\) without loss of generality. In addition, the Nambu constraint \(\hat{\Phi}=\tau_1\hat \Phi^{\dagger\, \text{T}}\) is also preserved by a Gaussian isometry. As a result, the pseudo-unitary transformation \(U_{\! \mathbi{V}}\) is real symplectic, satisfying conditions $\tau_1  U_{\! \mathbi{V}}^*\tau_1=U_{\! \mathbi{V}}$ and $U_{\! \mathbi{V}}^\text{T} \tau_2 U_{\! \mathbi{V}}=\tau_2$.

Consider first linear Gaussian symmetries. One finds that \(\mathbi{V} \, \widehat{H}_b \, \mathbi{V}^{\, \dagger}=\widehat{H}_b\) if and only if \(U_{\! \mathbi{V}}^{-1} H_\tau U_{\! \mathbi{V}}=H_\tau\). Hence, linear Gaussian symmetries of the many-body ensemble descend into symmetries of the auxiliary pseudo-Hermitian ensemble of \(H_\tau\). Conversely, one can again invoke the Stone-von Neumann-Mackey theorem to lift a symmetry of the auxiliary ensemble into a linear Gaussian symmetry of the many-body ensemble. If those symmetries are continuous, then $\mathbi{V}=e^{i\widehat{Q}_b}$ for some Hermitian quadratic bosonic generator $\widehat{Q}_b$. An infinitesimal symmetry transformation leads to a pseudo-unitary matrix $U_\mathbi{V}= e^{u_\mathbi{V}}\approx \mathds{1}+ u_\mathbi{V}$, such that $u_\mathbi{V} \, \tau_3+\tau_3 \, u_\mathbi{V}^\dagger=0$ and \(\tau_1 \, u_\mathbi{V}^*\, \tau_1=u_\mathbi{V}\). Combining these observations with the identity \(i[\widehat{Q}_b,\hat \Phi]=-i\tau_3Q_b \hat\Phi=-iQ_\tau\hat \Phi\) and Eq.\,\eqref{gaussianb}, one concludes that \(u_\mathbi{V}=-iQ_\tau\) and \(U_\mathbi{V}^{-1}H_\tau U^{\;}_\mathbi{V}=H_\tau\) if and only if \([Q_\tau,H_\tau]=0\). Physically, the matrix \(Q_\tau=\tau_3Q_b\) may represent a conserved charge. One can reach the same conclusion directly by calculating 
\begin{eqnarray}
\label{bosym}
[\widehat{Q}_b,\widehat{H}_b]=\frac{i}{2}\hat{\Phi}^\dagger\tau_3[Q_\tau,H_\tau]\hat{\Phi}.
\end{eqnarray}

It is reasonable to have an ensemble of \(H_\tau\) reducible, under the same conditions as before, and assume that any such ensemble can be 
 decomposed into a sum of irreducible ensembles by, roughly speaking, block-diagonalizing the ensemble together with its Gaussian symmetries. 
This decomposition of the auxiliary ensemble corresponds to decomposing Hamiltonians of the many-body ensemble as sums of commuting many-body Hamiltonians labelled by the quantum numbers of the Gaussian symmetries.    

Next, consider antilinear Gaussian symmetries, that is, time-reversal-like symmetries of the form
\begin{align}
\mathbi{T}\, \hat{\Phi} \, \mathbi{T}^{-1}
=U^{\;}_T \, \hat \Phi, \quad \mathbi{T} \, i=-i\mathbi{T} , 
\end{align}
where $U^{\;}_T$ is a pseudo-unitary matrix. 
From the symmetry condition $\mathbi{T}\, \widehat{H}_b \, \mathbi{T}^{-1}=\widehat{H}_b$ we 
get a constraint on the effective BdG hamiltonians $H_\tau$ of the form
\begin{align}
\label{TRS}
U^{-1}_T H_\tau^* U^{\;}_T={\cal T}H_\tau{\cal T}^{-1}=H_\tau
\end{align}
in terms of \({\cal T}=U_T^{-1}{\cal K}\). 
Applying Eq.\,(\ref{TRS}) to $H_\tau$ twice and assuming that the ensemble is irreducible
one finds that $U^{\;}_TU_T^*=\pm\mathds{1}_{2N}$. If we consider the simplest case where there are no internal degrees of freedom, $U^{\;}_T$ could be $\mathds{1}_{2N}$, which is typical for phonons. Then one obtains the reality condition $H_\tau^*=H_\tau$.

\section{Topology of gapped free bosons: No-go theorems}
\label{zeromodes}

Particle-conserving free-boson systems in equilibrium are generically gapless (in the thermodynamic limit) regardless of the spectral gaps of the auxiliary single-particle Hamiltonian. By contrast, particle-non-conserving systems can be fully gapped at the many-body level. This observation opens up the possibility for SPT phases. In general, non-trivial SPT phases of particle non-conserving free-fermion systems display at least one of the following signatures:
\begin{itemize}
\item ZMs for open BCs; and
\item Odd or even fermion parity in the ground state, 
depending on BCs being periodic or antiperiodic  \cite{Ortiz14, Ortiz16}.
\end{itemize}
In this section we prove that fully gapped free-boson systems cannot possibly show these signatures. Moreover, we show that any pair of fully gapped free-boson Hamiltonians can always be connected adiabatically without closing the many-body gap or breaking any protecting symmetries. Hence, non-trivial SPT phases do not exist for free-boson systems. For the special case $d=2$,  periodic BCs, and no additional symmetries, a related result was derived based on the triviality of the Chern number in Ref. [\onlinecite{Shindou}].

We organize these results as three no-go theorems for fully gapped free-boson systems. Our proofs are independent of each other and, in all cases, \textit{the central obstruction is the stability constraint} discussed in Sec.\,\ref{theiff}, rather than profound differences of internal symmetries between fermions and bosons. The theorems are:
\begin{itemize}
\item Theorem 1 (no parity switches): The boson parity of the ground state is always even.
\item Theorem 2 (no non-trivial SPT phases): All Hamiltonians are adiabatically connected
regardless of the choice of protecting symmetries.
\item Theorem 3 (no localized ZMs): The system subject to open BCs cannot develop localized ZMs.
\end{itemize}
All these results are pleasingly consistent. Of course, our no-go Theorem 3 does not forbid localized ZMs for gapped systems subject to BCs other than open. The squaring-the-fermion map can precisely generate such ZMs, and we will see concrete examples in Sec. \ref{examples}. With minimal modifications, our Theorems 2 and 3 hold also for linear-quadratic Hamiltonians of Eq.\,\eqref{linear_quadratic} in Sec.\,\ref{theiff}.

\subsection{Theorem 1: No parity switches}
\label{nops}

\begin{theorem}
Let $\widehat{H}_{b}$ be a quadratic bosonic 
Hamiltonian with finitely many modes \(N\) and \(H_b>0\). 
Then, its ground state is non-degenerate 
with even boson parity. 
\end{theorem}

\noindent

\begin{proof}
Since \(H_b>0\), it follows that there are no ZMs and $\widehat{H}_{b}$ can be written in the form of Eq.\,(\ref{bosonic quasi-particle}). The ground state
is the unique vacuum for the quasi-particles. Let the Bogoliubov transformation from the $a$ to the $b$ bosonic modes be furnished by a $2N \times 2N$ matrix of the following block structure: 
\begin{eqnarray}
\label{xy}
\begin{bmatrix} b \\ b^\dagger \end{bmatrix} = 
\begin{bmatrix} X & -Y \\ -Y^* & X^* \end{bmatrix}
\begin{bmatrix} a \\ a^\dagger \end{bmatrix} ,
\end{eqnarray}
then, the quasi-particle vacuum state is \cite{Blaizot}:
\begin{equation}
\label{bgs}
|\Omega\rangle = \det(X^\dagger X)^{-1/4}\exp
\Big[ \frac{1}{2}\hat \phi^\dagger  (X^{-1}Y) \hat \phi^{\dagger  \text{T}} \Big]|0\rangle.
\end{equation}
Since the exponent is quadratic, the even parity of the ground state $|\Omega\rangle$ follows. We still need to show that the formula always holds, that is, we need to show that \(X\) is invertible. Because the matrix in Eq.\,\eqref{xy} is a canonical transformation, it follows that $XX^\dag-YY^\dag=\mathds{1}_{N}$, 
another form of the orthonormality relations of 
Eq.\,(\ref{orthonormality}). Thus,   $\det(XX^\dag)=\det(\mathds{1}_{N}+YY^\dag)\geq 1$, whereby 
$X$ is necessarily invertible.  
\end{proof}

On the one hand, odd fermion parity in the ground state and fermion parity switches are usual (albeit not mandatory\cite{AlasePRL,PRB1}) signatures of non-trivial SPT phases in fermionic systems, even interacting ones \cite{Ortiz14, Ortiz16}. On the other hand, for free fermions a ground-state expression equivalent to Eq. \eqref{bgs} exists\cite{Blaizot}, seemingly implying that fermions always have even parity ground states. Crucially, the fermionic analogue of the matrix \(X\) in Eq.\,\eqref{xy} can fail to be invertible, leading in those cases to an odd-parity ground state. For bosons, this possibility is excluded because $H_b>0$ implies a many-body gap (as we discussed in Sec.\,\ref{theiff}), a condition that has no analogue for fermions. We conclude that parity switches are non-existent for  stable free bosons. 

\subsection{Theorem 2: No SPT phases}

\begin{theorem}
Let $\widehat{H}_{b,1}$ and $\widehat{H}_{b,2}$ denote gapped, 
particle-non-conserving quadratic bosonic Hamiltonians sharing
a group of symmetries. Then $\widehat{H}_{b,1}$ can be adiabatically deformed into $\widehat{H}_{b,2}$ without breaking any of the symmetries, or closing the many-body gap. 
\end{theorem}

\begin{proof}
Consider the continuous path
\begin{align}
\widehat{H}_{b}(s) = (1-s)\widehat{H}_{b,1}+s\widehat{H}_{b,2},\quad s \in [0,1],
\label{mapL}
\end{align}
which implies the Hermitian matrix
$H_b(s)=(1-s)H_{b,1}+sH_{b,2}$
satisfies the constraint $\tau_1 H_b^*(s)\tau_1=H_b(s)$ for all \(s\). 
Moreover, if $\mathbi{V}$ is a 
linear or antilinear symmetry shared by the initial and final free-boson systems, then \( \mathbi{V}\widehat{H}_{b}(s)\mathbi{V}^{-1} = \widehat{H}_{b}(s)\) for all \(s\). Finally, since
$\widehat{H}_{b,1}$ and $\widehat{H}_{b,2}$ are particle non-conserving and gapped, it follows that \(H_{b,1}, H_{b,2} >0\), hence 
\(H_b(s)=(1-s)H_{b,1}+sH_{b,2}>0\). In other words, \(\widehat{H}_{b}(s)\)
is fully gapped for all \(s\).  
\end{proof}

Notice that the map in Eq. \eqref{mapL} preserves the locality properties of \(H_{b,1}\) and $H_{b,2}$. A corollary of this theorem is that there are no non-trivial SPT phases of particle non-conserving free-boson systems. Any such system can be adiabatically deformed into a topologically trivial system
without closing the many-body gap or breaking any symmetry. Hence, one expects that stable, gapped systems of free bosons will not display any signatures of non-trivial topology.

How do free-fermion systems escape this triviality result? Again, the answer is that there is no counterpart for fermions of the bosonic gap condition \(H_b>0\). Referring back to the proof of Theorem 2,  
there is nothing  that can prevent, in general, the closing of the gap along the path \(H_{f}(s)=(1-s)H_{f,1}+sH_{f,2}\) of free-fermion Hamiltonians. In fact, a key insight of the tenfold way is that the gap \emph{must close} when \(H_f(s)\) interpolates between two topologically distinct Hamiltonians in the same symmetry class. By contrast, it is easy to find examples of paths of gapped Hamiltonians that connect different symmetry classes, regardless of topological invariants.   

\subsection{Theorem 3: No localized ZMs}
\label{nogo}

\begin{theorem}
\label{nogothm}
Consider a $d$-dimensional, translation-invariant free-boson system 
\(H_b(\mathbf{k})>0\), and the same system subject
to open BCs and described by $H^o_{\tau,\mathbf{k}_\parallel}$. 
Then, zero is not an eigenvalue of $H^o_{\tau,\mathbf{k}_\parallel}$.
\end{theorem}
\smallskip

\noindent
\textit{Remark.} This theorem is true even if the pairing contributions vanish, \(\Delta=0\), in which case the condition \(H_b({\bf k})>0\) is equivalent to \(K({\bf k})>0\) and does not imply a many-body gap in general.

\begin{proof}
The proof requires results from the theory of {\it matrix Wiener-Hopf factorization}\cite{gohberg,Alase20}. While the following argument is not self-contained, we provide references for all the necessary auxiliary  theorems. With reference to Sec.\,\ref{mapping2} for
our notation, let $\mathbf{k}_\parallel \in \text{SBZ}$ and  
\begin{align}
G(e^{ik}) \equiv H_{b,\mathbf{k}_\parallel}(k), \quad \forall 
k \in [-\pi,\pi),
\end{align}
that is, \(G\) is explicitly defined as a function on the unit circle in ${\mathbb C}$. Since $H_b(\mathbf{k})>0$, it follows that $G(e^{ik}) = H_{b,\mathbf{k}_\parallel}(k) > 0$ for all $k$. The matrix-valued function \(G\) is the \textit{symbol} of the block-Toeplitz operator $G^{o} = H^{o}_{b,\mathbf{k}_\parallel}.$
Regarding \(G^o\) as an infinite matrix, one can state the block-Toeplitz property  as \([G^o]_{i,j}=[G^o]_{i+1,j+1}\), with \(i,j=1,2,\cdots,\infty\) being the coordinate of a lattice point in the direction perpendicular to the termination of the lattice. The blocks \([G^o]_{i,j}\) act on internal, not lattice degrees of freedom. Back to the symbol, because \(G(e^{ik})>0\) for all $k \in [-\pi,\pi)$, it admits a {\it canonical Wiener-Hopf factorization}, that is, a factorization of the form 
\begin{eqnarray}
G(e^{ik}) = G_+(e^{ik})G_-(e^{ik}),
\end{eqnarray}
where the entries of $G_+(e^{ik})$ $\big(G_-(e^{ik})\big)$ and their inverses are analytic inside (outside) the unit circle, see Theorem 1.13 in Ref.\,[\onlinecite{gohberg}] (and Ref.\,[\onlinecite{youla}] for a system-theoretic perspective). Further, according to Theorem 2.13 in Ref.\,[\onlinecite{gohberg}], a block-Toeplitz operator is invertible if and only if its symbol admits a canonical Wiener-Hopf factorization. In our case, this implies that the block-Toeplitz operator $G^{o}$ is invertible. Consequently, $H^{o}_{\tau,\mathbf{k}_\parallel} = 
\tau_3 H^{o}_{b,\mathbf{k}_\parallel} = \tau_3G^o$ is also invertible. In conclusion, zero does not belong to the spectrum of $H^{o}_{\tau,\mathbf{k}_\parallel}$ for any value of $\mathbf{k}_\parallel \in \text{SBZ}$, and therefore it does not belong to the spectrum of $H^{o}_{\tau}$. 
\end{proof}

\begin{corollary}
\label{nogocorollary}
Localized midgap states cannot exist in the spectral gap separating positive from negative 
eigenvalues of $H^{o}_{\tau,\mathbf{k}_\parallel}$.
\end{corollary}

\begin{proof}
Suppose that a state with eigenvalue $\epsilon$ lies in that spectral gap. 
Then, similarly to the proof in Theorem 3, $G(e^{ik})-\epsilon \tau_3$ is positive definite and admits a canonical Wiener-Hopf factorization. As a result, $H^o_{b,\mathbf{k}_\parallel} - \epsilon \tau_3$ 
and $H^o_{\tau,\mathbf{k}_\parallel} - \epsilon\mathds{1}_\infty$ are invertible, whereby it follows that $\epsilon$ does not belong to the spectrum of $H^o_{\tau,\mathbf{k}_\parallel}$.
\end{proof}

\noindent
Physically, Theorem \ref{nogothm} excludes the possibility of having  a ZM. Corollary \ref{nogocorollary} excludes surface bands altogether in the spectral gap around zero energy, regardless of whether they cross zero energy. 

Since our no-go Theorem \ref{nogothm} and its Corollary \ref{nogocorollary} address a spectral connection between periodic and open BCs, they may be taken to provide a no-go result for a bulk-boundary correspondence, at least of a standard form.  
Their reach can be further extended by considering generic single-particle perturbations bounded in the operator norm. Suppose that a perturbation $W_b$ that satisfies the constraint $\tau_1 W_b^*\tau_1=W_b$ is added to $H_\tau^o$. Then the effective bosonic BdG Hamiltonian is 
\begin{align}
H_\tau^o+W_\tau = H_\tau^o[\mathds{1} + (H_\tau^o)^{-1}W_\tau], 
\end{align}
with $W_\tau= \tau_3 W_b$. As long as $||W_\tau|| < 1/||(H_\tau^o)^{-1}||$, the above Hamiltonian 
is invertible and therefore does not have a zero eigenvalue. Notice that  $1/||(H_\tau^o)^{-1}||$ equals the smallest energy eigenvalue of 
$H^o_\tau$, which is the first excitation energy above the bosonic vacuum.
These perturbations do not include bulk disorder but can model a variety of BCs and boundary disorder that decays sufficiently fast into the bulk.

\section{Squaring the fermion}
\label{classification}

\subsection{The square of a fermion is a boson}
\label{diquare}

An even-dimensional Hermitian matrix $H$ can arise as the BdG Hamiltonian 
of a free-fermion system if and only if it satisfies the particle-hole constraint $\tau_1H^*\tau_1=-H$ of Sec. \ref{mapping2}. One can recast this constraint in terms of a projector superoperator,  
\begin{eqnarray}
{\cal F}(H)\equiv \frac{1}{2}\left( H-\tau_1H^*\tau_1 \right), \quad H=H^\dagger.
\end{eqnarray}
That is, an even-dimensional Hermitian matrix $H_f$ can be associated to a fermionic BdG Hamiltonian if and only if ${\cal F}(H_f)=H_f$. Similarly, an even-dimensional Hermitian matrix $H$ can be associated to a free-boson system if and only if it satisfies the constraint $\tau_1H^*\tau_1=H$, which again can be recast in terms of a projector superoperator, 
\begin{eqnarray}
{\cal B}(H)\equiv \frac{1}{2}\left(H+\tau_1H^*\tau_1\right).
\end{eqnarray}
We call a Hermitian matrix $H_b$ with ${\cal B}(H_b)=H_b$ \textit{bosonic}. 

Now we are in a position to state two interesting relationships between fermionic BdG Hamiltonians and bosonic matrices:
\begin{enumerate}
    \item Every even-dimensional Hermitian matrix is the sum of a {\em unique} fermionic BdG Hamiltonian and a {\em unique} bosonic matrix; and
    \item The square of a fermionic BdG Hamiltonian is a bosonic matrix.
\end{enumerate}
The first result follows because the projectors \({\cal F}\) and \({\cal B}\) are complementary, $H={\cal F}(H)+{\cal B}(H)$, and have disjoint ranges, \({\cal F}\circ{\cal B}={\cal B}\circ{\cal F}=\emptyset\). The second result follows because ${\cal B}(H_f^2)=H_f^2$. More generally, bosonic matrices can be obtained from fermionic ones via application of a broader class of even functions, e.g., if \(\mathds{P}\) is an even polynomial, then \(\mathds{P}(H_f)\) is a bosonic matrix. However, not only does the square function provides the simplest mathematical option but, as we discuss next, the resulting fermion-to-boson mapping has a number of remarkable properties from a physical standpoint.   

By construction, the bosonic matrix 
\(H_f^2\) is positive semi-definite, $H_f^2\geq 0$. Hence, 
the free-boson system described by the quadratic bosonic Hamiltonian 
\begin{eqnarray}
\label{BdG}
\epsilon_0 \widehat{H}_b=\frac{1}{2}\hat{\Phi}^\dagger H_f^2\hat{\Phi}-\frac{1}{2}\mbox{tr}(K'), \quad \epsilon_0>0,
\end{eqnarray}
is stable. Here \(\epsilon_0\) is some suitable constant with units of energy and \(K' \equiv K^2-\Delta\Delta^*\) in terms of the single-particle Hamiltonian and pairing for the free-fermion system. Accordingly, we have identified a \textit{squaring map}, 
\begin{align}
{\mathscr S}(H_f)\equiv \tau_3H_f^2 \equiv H_{\tau,{\mathscr S}},
\end{align}
from BdG Hamiltonians of free-fermion systems to effective BdG Hamiltonians of stable free-boson systems. 

The above squaring map is interesting because the kernel, that is, the ZMs, of \(H_f\) and \(\tau_3H_f^2\) coincide, even though \(\tau_3H_f^2\) may also display additional zero-energy \emph{generalized} eigenvectors (see Sec.\,\ref{1DKitaev} for an example). In this sense, the square-of-a-fermion procedure offers a systematic way of constructing bosonic models with ``Majorana bosons'' in a gap or as part of a surface band. Notice that locality is not a concern. If $H_f$ features a non-zero hopping amplitude between sites that are \(r\) units apart, then \(H_f^2\) features a hopping amplitude between sites that are at most \(2r\) units apart (see Sec.\,\ref{examples} for explicit examples). Importantly, not all effective BdG Hamiltonians are in the range of the squaring map. For example, for periodic BCs, the effective BdG Hamiltonian of the gapless harmonic chain, \(\widehat{H}_{\sf hc}=\sum_{k} \sqrt{2(1-\cos k)} (a^\dagger_k a^{\;}_k + 1/2) + P_0^2/2\), displays a \emph{single} zero-energy eigenvector associated to the conserved total momentum operator $P_0$. Since zero eigenvectors of \(H_f\) come in pairs, the same is true of \(\tau_3H_f^2\) and so the harmonic chain \emph{cannot} possibly be the square of a fermion. 

Besides lattice models as in Sec. \ref{mapping}, the squaring map we have introduced extends naturally to continuum models. We consider the simplest case of the Dirac Hamiltonian in three spatial dimensions for illustration\cite{NoteHerbut}. In the absence of gauge fields, the latter reads
\begin{eqnarray}
H_D=c(\gamma_1 p_1+\gamma_2 p_2+\gamma_3 p_3)+mc^2\gamma_4,
\end{eqnarray}
where \(p_\nu=-i\hbar{\partial}/{\partial x^\nu}\), \(\gamma_\nu\, \nu=1,2,3,4,\) are Hermitian matrices satisfying the Clifford algebra $\{\gamma_\nu,\gamma_{\nu'}\}=2\delta_{\nu\nu'}$,
which force them to be (at least) \(4\times 4\) matrices. The choice $\gamma_1=
\mathds{1}_2\otimes\sigma_1,\,\gamma_2=\sigma_1\otimes\sigma_3,\,\gamma_3=
\sigma_2\otimes\sigma_3$ and $\gamma_4=\sigma_3\otimes\sigma_3$ highlights the fact 
that the Dirac Hamiltonian can be regarded as a (continuous-coordinate) fermionic BdG Hamiltonian of the form 
$H_D=\begin{bmatrix}
K& \Delta\\
-\Delta^*& -K^*
\end{bmatrix}$, with $K\equiv\sigma_1 cp_1+\sigma_3 mc^2$ and $\Delta\equiv\sigma_3 c(p_2-ip_3)$. Thus, one can second-quantize $H_D$ as  
\begin{eqnarray}
\widehat{H}_D=\frac{1}{2}\int \hat{\Psi}^\dagger(\vec{x})H_D\hat{\Psi}(\vec{x})\, d^3x,
\end{eqnarray}
a field theory in terms of the Nambu array 
\begin{align}
\hat{\Psi}^\dagger(\vec{x})=
\begin{bmatrix}c_1^\dagger(\vec{x})&c_2^\dagger(\vec{x})&c_1(\vec{x})&c_2(\vec{x})\end{bmatrix} ,
\end{align}
with
$\big\{c_i(\vec{x}),c_j(\vec{y})\big\}=0$, 
$\big\{c_i(\vec{x}),c^\dagger_j(\vec{y})\big\}=\delta_{ij}\delta(\vec{x}-\vec{y})$.
Applying the map ${\mathscr S}$ yields the associated free-boson system with the bosonic matrix 
\begin{eqnarray}
H_D^2=(p^2c^2+m^2c^4)\otimes \mathds{1}_4.
\end{eqnarray}
Choosing \(\epsilon_0=2mc^2\), the gap of \(H_D\), we obtain the 
free-boson second-quantized form
\begin{align}
\widehat{H}_b= \frac{1}{2}\int \!\hat{\Phi}^\dagger
(\vec{x})\Big(\frac{p^2}{2m}\otimes \mathds{1}_4+\frac{mc^2}{2}\otimes \mathds{1}_4\Big)
\hat{\Phi}(\vec{x})\, d^3x ,
\end{align}
in terms of the Nambu array of canonical bosons 
\(\hat\Phi^\dagger(\vec{x})=\begin{bmatrix}a_1^\dagger(\vec{x})&a_2^\dagger(\vec{x})&a_1(\vec{x})&a_2(\vec{x})\end{bmatrix}\). From now on we will drop any explicit reference to \(\epsilon_0\).

\subsection{Symmetry analysis of the squaring map}
\label{SS}

While the squaring map can break some fermionic symmetries, it certainly
preserves many as well. Here, we investigate those fermionic continuous symmetries that are inherited by the squaring map. In preparation for the 
classification of squared ensembles we will consider in Sec.\,\ref{symmetryclasses}, we will also need a detailed understanding of how the resulting symmetry reduction 
intertwines with the squaring map and the indefinite metric $\tau_3$. 

Consider the 
generators of symmetries of $H_f$
\begin{align}
\mathfrak{g}_f\equiv \{Q_f=Q_f^\dagger \, | \, \tau_1Q_f^*\tau_1=-Q_f, \ [Q_f,H_f]=0\} .
\end{align}
\(Q_f\in \mathfrak{g}_f\) if and only if the conserved charge $\widehat{Q}_f$ obeys \([\widehat{Q}_f,\widehat{H}_f]=0\). The fermion number and the total spin are good examples. The Lie group $e^{i\mathfrak{g_f}}$ of unitary matrices $U_f$ that commute with \(H_f\) and satisfy 
\(\tau_1U_f^*\tau_1=U^{\;}_f\) is precisely the group 
associated to fermionic Gaussian symmetries of particle-non-conserving systems. The block structure of these matrices is
\[
U_f=
\begin{bmatrix}
{\sf A} & {\sf B}\\ {\sf B}^* & {\sf A}^*
\end{bmatrix}, 
\quad {\sf A}{\sf A}^\dagger+{\sf B}{\sf B}^\dagger=\mathds{1}_N,\quad {\sf A}{\sf B}^\text{T}+{\sf B}{\sf A}^\text{T}=0,
\]
with \({\sf A}\) and \({\sf B}\) \(N\times N\) matrices. 
Similarly, let 
\begin{align}
\mathfrak{g}_b\equiv \{Q_b=Q_b^\dagger \, |\, \tau_1Q_b^*\tau_1=Q_b,\ [Q_\tau,H_\tau]=0\}
\end{align}
be the generators of symmetries of \(H_\tau=\tau_3H_f^2\), with $Q_\tau=\tau_3Q_b$
(see Sec.\,\ref{pncsec}). Again, 
\(Q_b\in \mathfrak{g}_b\) if and only if \([\widehat{Q}_b,\widehat{H}_b]=0\).
Since the squaring map involves the $\tau_3$ matrix (Nambu formalism in Table I), one needs to analyze fermionic symmetries that either commute or anticommute with $\tau_3$. 

First, let us focus on $U_f$ that commutes with \(H_f\) and \(\tau_3\). Then, 
\(U_f\) is also a pseudo-unitary matrix (with 
\(U^{\;}_f\tau_3U_f^\dagger=U_fU_f^\dagger \tau_3
=\tau_3\)), that commutes with \(\tau_3H_f^2\).
Consider the class of symmetry generators of the form
\begin{eqnarray}
\label{diagq}
Q_f=Q_\tau=Q\equiv \begin{bmatrix}q& 0\\ 0 & -q^*\end{bmatrix},\quad q=q^\dagger ,
\end{eqnarray}
resulting in
\begin{align}
[Q,H_f]=\begin{bmatrix} [q,K] & q\Delta+\Delta q^*\\
(q\Delta+\Delta q^*)^* & [q,K]^* \end{bmatrix} .
\end{align}
Then, symmetries of the particle-conserving part $K$, which are preserved by pairing $\Delta$, will always be symmetries of the free-boson system $\tau_3H_f^2$. 

The symmetry reduction of the squared ensemble can be achieved by first reducing the fermionic ensemble and then applying the squaring map ${\mathscr S}$
to each block, \textit{but with respect
to a suitably defined indefinite metric}. A key aspect of the problem is precisely the determination of the appropriate reduced metric. Technical details of this symmetry-reduction analysis can be found in Appendix \ref{afterandbefore}, whereas a summary is shown
in Table\,\ref{squared_subensembles}. This table describes the key structural feature of the subensembles of \(\{\tau_3H_f^2\}\) as determined by the shared conserved quantum number labelling the blocks and the associated subensembles of \(\{H_f\}\). 

\begin{table}
    \centering 
    \caption{Subensembles of \(\{\tau_3 H_f^2\}\) parametrized by the subensembles of \(\{H_f\}\).
    The blocks are labelled by 
    eigenvalues (Eig) of \(q\) [see Eq.\,\eqref{diagq}] in the first column, with the corresponding degeneracy (Deg) listed in the second column. 
    The symbol (f)/(b) in the third/fourth column means fermions/bosons (squared fermions in particular). SPH stands for the single-particle Hamiltonian, 
    with the block describing a particle-conserving many-body block, while the block of BdG describes a particle non-conserving many-body block and displays the particle-hole constraint. The block of \(U(1)\) describes a 
    particle-non-conserving many-body block with a conserved  
    \(U(1)\) charge, for example, the total spin; thanks to this charge, 
    the need for the Nambu formalism is bypassed 
    and there is no particle-hole constraint. }
        \smallskip
    \begin{tabular}{x{8mm} x{9mm} x{19mm} x{14mm} x{29mm}}
    \hline\hline \\ [-2ex]
     Eig &  Deg   & Block (f) & Block (b) & Metric
    \tabularnewline [1ex] \hline \\[-1ex]
    \(\kappa\)         &  \(m\)    &    \(K_\kappa\)\ \, (SPH)      & \(K_\kappa^2\) & \(\mathds{1}_m\) 
    \tabularnewline [1ex] \\
    0               &   \(m\)   & \(H_{f,0}\)\  (BdG)   &  \(\tau_3H_{f,0}^2\)  & \(\tau_3=\sigma_3\otimes \mathds{1}_m\)
    \tabularnewline [1ex] \\
    \(\kappa, -\kappa\) &  \( m, n\)    & \(H_{\pm\kappa}\)\ (\(U(1)\)) & \(\tau_{m,n}H_{\pm\kappa}^2\)
     & \(\tau_{m,n}=\begin{bmatrix}\mathds{1}_m& 0\\0&-\mathds{1}_n\end{bmatrix}\)      
    \tabularnewline\\ [-1ex] \hline\hline
    \end{tabular}
    \label{squared_subensembles}
\end{table}

\smallskip

{\em Example.} For illustration, we square the BdG Hamiltonian of a conventional BCS superconductor. We start with the BCS mean-field Hamiltonian 
\begin{align*}
\widehat{H}_f=\!\sum_{\mathbf{k}, \sigma\in \{ \uparrow,\downarrow \} }(\epsilon_\mathbf{k}-\mu)c_{\mathbf{k}\sigma}^\dagger c^{\;}_{\mathbf{k}\sigma}
+\sum_{\mathbf{k}}(\Delta c_{\mathbf{k}\uparrow}^\dagger c_{\mathbf{-k}\downarrow}^\dagger + \text{h.c.}).
\end{align*}
In the BdG formalism, the above equation becomes $\widehat{H}_f=\sum_{\mathbf{k}}\hat{\Psi}_\mathbf{k}^\dagger 
H_f(\mathbf{k})\hat{\Psi}_\mathbf{k}/2+\sum_{\mathbf{k}}(\epsilon_\mathbf{k}-\mu)$, where $\hat{\Psi}_\mathbf{k}^\dagger
=[c_{\mathbf{k}\uparrow}^\dagger\,c_{\mathbf{k}\downarrow}^\dagger\,c^{\;}_{\mathbf{-k}\uparrow}\,c^{\;}_{\mathbf{-k}\downarrow}]$ 
and $H_f(\mathbf{k})=
\begin{bmatrix}
K(\mathbf{k})	&\Delta(\mathbf{k})\\
-\Delta^*(\mathbf{-k})	&-K^*(\mathbf{-k})
\end{bmatrix}$, with $K(\mathbf{k})=
\begin{bmatrix}
\epsilon_\mathbf{k}-\mu	&0\\
0	                    &\epsilon_\mathbf{k}-\mu
\end{bmatrix}$ and $\Delta(\mathbf{k})=
\begin{bmatrix}
0	    &\Delta\\
-\Delta	&0
\end{bmatrix}$.
Then, besides the built-in particle-hole constraint 
$\tau_1\mathcal{K}=(\sigma_1\otimes\mathds{1}_2)\mathcal{K}$, $H_f(\mathbf{k})$ has another particle-hole symmetry $(\sigma_2\otimes\sigma_3)\mathcal{K}$ and, therefore, a unitary commuting symmetry $U=\sigma_3\otimes\sigma_3$, with $[\tau_3, U]=0$. Furthermore, $U$ is also a unitary commuting symmetry of $H_\tau(\mathbf{k})=\tau_3 H_f^2(\mathbf{k})$, that can be block-diagonalized. After defining a $4\times 4$ permutation matrix $P$ with non-vanishing elements $P_{11}=P_{23}=P_{34}=P_{42}=1$, we have 
 \[
 P^\dagger H_\tau(\mathbf{k})P=P^\dagger\tau_3 H_f^2(\mathbf{k})P=\big(P^\dagger\tau_3 P\big)\big(P^\dagger H_f(\mathbf{k})P\big)^2,
 \] 
 where 
 $P^\dagger\tau_3 P
 =\begin{bmatrix}
\sigma_3 &0\\
0        &\sigma_3
\end{bmatrix}$ 
and 
\begin{align}
P^\dagger H_f(\mathbf{k})P=
\begin{bmatrix}
\widetilde{H}_f(\mathbf{k},\mu,\Delta) &0\\
0                                      &\widetilde{H}_f(\mathbf{k},\mu,-\Delta)
\end{bmatrix},
\end{align}
with 
$\widetilde{H}_f(\mathbf{k},\mu,\Delta)=\begin{bmatrix}
\epsilon_\mathbf{k}-\mu &\Delta\\
\Delta^*                  &-(\epsilon_\mathbf{-k}-\mu)
\end{bmatrix}$, that is, the irreducible block of $H_f(\mathbf{k})$. Now we have
\begin{align}
P^\dagger H_\tau(\mathbf{k})P=
\begin{bmatrix}
\sigma_3\widetilde{H}_f^2(\mathbf{k},\mu,\Delta) &0\\
0                                      &\sigma_3\widetilde{H}_f^2(\mathbf{k},\mu,-\Delta)
\end{bmatrix},
\end{align}
and we only need to focus on the irreducible block 
$\sigma_3\widetilde{H}_f^2(\mathbf{k},\mu,\Delta)\equiv\sigma_3\widetilde{H}_f^2(\mathbf{k})\equiv 
\widetilde{H}_\tau(\mathbf{k})$ of $H_\tau(\mathbf{k})$ and the irreducible block 
$\widetilde{H}_f(\mathbf{k},\mu,\Delta)\equiv\widetilde{H}_f(\mathbf{k})$ of $H_f(\mathbf{k})$. 
As we can see, $\widetilde{H}_\tau(\mathbf{k})=\sigma_3\widetilde{H}_f^2(\mathbf{k})$ resembles $H_\tau(\mathbf{k})=\tau_3 H_f^2(\mathbf{k})$ a lot. However, $\widetilde{H}_f(\mathbf{k})$ no longer has the built-in particle-hole constraint. Instead, $\widetilde{H}_f(\mathbf{k})$ has another particle-hole symmetry $\sigma_2\mathcal{K}\equiv U^\dagger_C\mathcal{K}$, and usually a time-reversal symmetry $\mathcal{K}$ as well (with $U^{\;}_T=\mathds{1}_2$), if $\epsilon_\mathbf{k}=\epsilon_\mathbf{-k}$ and $\Delta=\Delta^*$ -- which means that the BdG Hamiltonian of conventional BCS superconductors 
usually belongs to class CI. Nonetheless, we will work with the case $\epsilon_\mathbf{k}\neq\epsilon_\mathbf{-k}$, so that the pairing potential will not vanish in $\widetilde{H}_f^2(\mathbf{k})$, with
\begin{align}
\widetilde{H}_f^2(\mathbf{k})=\begin{bmatrix}
(\epsilon_\mathbf{k}-\mu)^2+|\Delta|^2              &\Delta(\epsilon_\mathbf{k}-\epsilon_\mathbf{-k})\\
\Delta^*(\epsilon_\mathbf{k}-\epsilon_\mathbf{-k})  &(\epsilon_\mathbf{-k}-\mu)^2+|\Delta|^2
\end{bmatrix},
\end{align}
and with the quadratic bosonic Hamiltonian associated to $H_\tau(\mathbf{k})$ being given by
\begin{align}
\label{BCSsquared}
\widehat{H}_b=&\sum_{\mathbf{k}\sigma}\big[(\epsilon_\mathbf{k}-\mu)^2+|\Delta|^2\big]a_{\mathbf{k}\sigma}^\dagger a^{\;}_{\mathbf{k}\sigma}\nonumber\\
+&\sum_{\mathbf{k}}\big[\Delta(\epsilon_\mathbf{k}-\epsilon_\mathbf{-k}) a_{\mathbf{k}\uparrow}^\dagger a_{\mathbf{-k}\downarrow}^\dagger + \text{h.c.}\big].
\end{align}

Back to the general discussion, consider next a unitary symmetry \(U_f\) that anti-commutes with \(\tau_3\). It follows that \(U_f^\dagger\tau_1\) is a pseudo-unitary matrix and 
\(U_f^\dagger\tau_1{\cal K}\equiv \mathcal{T}\) commutes with \(\tau_3H_f^2\). 
An interesting physical observation emerges in this case. 
If \(\{U_f,\tau_3\}=0\), one can write \begin{align}
U_f=\begin{bmatrix} 0 & U_C^\dagger\\ U_C^\text{T} & 0\end{bmatrix}, \quad U^{\;}_CU_C^\dagger=\mathds{1}_N .
\end{align} 
Then, the symmetry condition $[U_f,H_f]=0$ reads 
\begin{align}
\begin{bmatrix} -U_C^\dagger K^*U_C & -U_C^\dagger\Delta^*U_C^*\\
U_C^\text{T}\Delta U_C & U_C^\text{T} K U_C^*\end{bmatrix}=
\begin{bmatrix} K& \Delta\\ -\Delta^* & -K^*\end{bmatrix}.
\end{align}
Hence, these symmetries of the fermionic BdG Hamiltonian are inherited
from a particle-hole condition satisfied by the single-particle Hamiltonian.
The associated bosonic system after squaring inherits instead a time-reversal symmetry \({\cal T}=U_f^\dagger\tau_1{\cal K}\). This phenomenon is akin to the notion of symmetry transmutation first discussed in the context of dualities\cite{Cobanera}. 

Finally, fermionic symmetries \(U_f\) that neither commute nor anticommute
with \(\tau_3\) are broken by the squaring procedure. These symmetries mix \(K\) and \(\Delta\), and  emerge because of the specific interplay between \(K\) and \(\Delta\). 
An ensemble of BdG Hamiltonians necessarily satisfies a particle-hole constraint and could satisfy other classifying conditions either before or after the symmetry reduction. We will address the interplay of these classifying conditions 
with the squaring map next.

\subsection{Topological classification}
\label{symmetryclasses}

\begin{table*}[]
    \centering
    \caption{Particle-conserving free fermions $H_f(\mathbf{k})$ under the squaring map. The left major column corresponds to free fermions, 
    whereas the right major column  
    corresponds to free bosons, obtained by squaring. 
    The three symmetries $\mathcal{T}$, $\mathcal{C}$ and $\mathcal{S}$ are denoted by $1\,(-1)$ if they square to $\mathds{1}\,(-\mathds{1})$, and by $0$ if they are absent. 
    For $H_f(\mathbf{k})$ 
    with chiral symmetry, a unitary commuting symmetry of $H_f^2(\mathbf{k})$ exists and, 
    after block-diagonalization, 
    $H_f^2(\mathbf{k})$ will fall into classes \{A, AI, AII\}.} 
\begin{tabular}{C{1.45cm} C{1.45cm} C{1.45cm} C{1.45cm} c | C{1.45cm} C{1.45cm} C{1.45cm} C{1.45cm} c}
\hline \hline
 $H_f(\mathbf{k})$ &  $\mathcal{T}$  & $\mathcal{C}$ &  $\mathcal{S}$  & Classifying space & $H^2_f(\mathbf{k})$ &  $\mathcal{T}$  & $\mathcal{C}$ &  $\mathcal{S}$  & Classifying space \\ 
 \hline
 A    & 0 & 0 & 0 & $\mathcal{C}_0$ &  A & 0 & 0 & 0 & $\mathcal{C}_0$ \\ 
 AIII    & 0 & 0 & 1 & $\mathcal{C}_1$ & \multicolumn{5}{c}{ \footnotesize   A unitary commuting symmetry of $H_f^2$ exists}   \\ 
 AI    & 1 & 0 & 0 & $\mathcal{R}_0$ &  AI    & 1 & 0 & 0 & $\mathcal{R}_0$ \\ 
 BDI   & 1 & 1 & 1 & $\mathcal{R}_1$ &   \multicolumn{5}{c}{ \footnotesize   A unitary commuting symmetry of $H_f^2$ exists} \\ 
 D    & 0 & 1 & 0 & $\mathcal{R}_2$ &  AI    & 1 & 0 & 0 & $\mathcal{R}_0$   \\ 
 DIII    & -1 & 1 & 1 & $\mathcal{R}_3$ &   \multicolumn{5}{c}{ \footnotesize   A unitary commuting symmetry of $H_f^2$ exists} \\ 
 AII    & -1 & 0 & 0 & $\mathcal{R}_4$ &  AII    & -1 & 0 & 0 & $\mathcal{R}_4$  \\ 
 CII    & -1 & -1 & 1 & $\mathcal{R}_5$ &   \multicolumn{5}{c}{ \footnotesize   A unitary commuting symmetry of $H_f^2$ exists}  \\ 
 C    & 0 & -1 & 0 & $\mathcal{R}_6$ &  AII    & -1 & 0 & 0 & $\mathcal{R}_4$   \\ 
 CI   & 1 & -1 & 1 & $\mathcal{R}_7$ &   \multicolumn{5}{c}{ \footnotesize   A unitary commuting symmetry of $H_f^2$ exists}  \\ \hline \hline
\end{tabular}
    \label{tab:Particle_conserving_squared}
\end{table*}

Let \(\{H_f(\mathbf{k})\}\) denote an ensemble of Bloch-BdG Hamiltonians. We briefly reviewed the topological classification of these ensembles in Sec.\,\ref{fermary}. Ensembles of systems without translation symmetry are included as the special case \(\mathbf{k}=0\). The squaring map yields an associated ensemble of effective Bloch-BdG Hamiltonians \({\mathscr S}(\{H_f(\mathbf{k})\})=\{\tau_3H_f^2(\mathbf{k})\}\), with the notation \(H_f^2(\mathbf{k})\equiv [H_f(\mathbf{k})]^2\). We call these bosonic -- in general pseudo-Hermitian -- ensembles the \textit{squared ensembles}. 

In this subsection, we will perform a symmetry-class analysis and establish a topological classification of these squared ensembles. As we saw in Sec.\,\ref{pncsec}, unitary transformations (as opposed to pseudo-unitary ones) do not have, in general, a many-body interpretation for particle-non-conserving free-boson systems. However, since symmetries of our squared ensembles are inherited from those of free fermions, although they could be pseudo-unitarily implemented, there exists necessarily a unitary implementation. As shown in Ref.\,[\onlinecite{Kawabata}], a pseudo-Hermitian matrix implemented with unitary symmetries (see also Refs.\,[\onlinecite{BLC}] and [\onlinecite{Zhou}]), together with a real energy gap (i.e., $\text{L}_\text{r}$ of Ref.\,[\onlinecite{Kawabata}]), can be continuously deformed into a Hermitian matrix while keeping its symmetries and gap. Therefore, some of the classification spaces of Hermitian matrices (see Table\,\ref{table:names}) can be used to label our squared ensembles once their symmetry analysis is realized.

For stable particle-non-conserving free-boson systems, we have excluded the existence of SPT phases in Sec.\,\ref{zeromodes}. Therefore, a topological classification of free bosons is only meaningful when we talk about single-particle states (rather than many-body ground states), i.e., SPT boundary states at finite energy, which is the main topic of this section. An earlier suggestion of classifying general effective BdG Hamiltonians using non-Hermitian symmetry classes can be found in Ref.\,[\onlinecite{Lieu}]. 

\subsubsection{Squared ensembles with vanishing pairing}

When the pairing potential in $H_f(\mathbf{k})$ vanishes, $H_\tau(\mathbf{k})$ becomes block-diagonal and we only need to focus on one block of it, e.g. $K_f^2(\mathbf{k})$, which corresponds to squaring a particle-conserving free fermion. The irreducible blocks of Hermitian ensembles of the form \(\{K_f^2(\mathbf{k})\}\) cannot possibly satisfy a particle-hole or chiral classifying symetry and so they must belong to one of the three classes \{A, AI, AII\}.  It is instructive to track in more detail the fate of the classifying conditions. For $H_f(\mathbf{k})$ with chiral symmetry $U_S^\dagger H_f(\mathbf{k})U^{\;}_S=-H_f(\mathbf{k})$, after squaring we get 
\begin{align}\label{chiralsquared}
U_S^\dagger H_f^2(\mathbf{k})U^{\;}_S=H_f^2(\mathbf{k}),
\end{align}
which means that $U^{\;}_S$ is a unitary commuting symmetry of $H_f^2(\mathbf{k})$.
For $H_f(\mathbf{k})$ with a time-reversal or particle-hole symmetry $U_{T/C}^\dagger H_f^*(\mathbf{-k})U^{\;}_{T/C}=\pm H_f(\mathbf{k})$, we have 
\begin{align}
U_{T/C}^\dagger H_f^{*2}(\mathbf{-k})U^{\;}_{T/C}=H_f^2(\mathbf{k}),
\end{align}
which means that $U^{\dagger}_{T/C}\mathcal{K}$ is a time-reversal symmetry of $H_f^2(\mathbf{k})$. These results are listed in Table \ref{tab:Particle_conserving_squared}. 

Even when the pairing potential in $H_f(\mathbf{k})$ does not vanish, 
it is still possible that, upon squaring, the pairing potential in $\tau_3H_f^2(\mathbf{k})$ vanishes, i.e., $K_f(\mathbf{k})\Delta_f(\mathbf{k})-\Delta_f(\mathbf{k})K_f^*(-\mathbf{k})=0$. This outcome is expected, for example, for Dirac BdG Hamiltonians because of the defining relations of the Clifford algebra. More concrete examples are spinless $2\times 2$ BdG Hamiltonians $H_f(\mathbf{k})$ of class BDI and spinless $2\times 2$ BdG Hamiltonians $H_f(\mathbf{k})$ of class D subject to the time-reversal symmetry of $K_f(\mathbf{k})=K_f^*(-\mathbf{k})$. Because of $K_f(\mathbf{k})=K_f^*(-\mathbf{k})$ in these two examples, 
$K_f(\mathbf{k})\Delta_f(\mathbf{k})-\Delta_f(\mathbf{k})K_f^*(-\mathbf{k})=0$ is obviously satisfied because $K_f(\mathbf{k})$ and $\Delta_f(\mathbf{k})$ are 
 scalars and so they commute with each other. In any case, if pairing vanishes in the squared ensemble because of the squaring map, we need to focus on one block of $H_\tau(\mathbf{k})$, e.g., $K_f^2(\mathbf{k})-\Delta_f(\mathbf{k})\Delta_f^*(-\mathbf{k})$.
The analysis of these blocks 
is included in Table \ref{tab:Particle_conserving_squared} as well.

\subsubsection{Squared ensembles with non-vanishing pairing}
\label{Squaredensembleswithnon-vanishingpairing}

Now we 
focus on the other cases where the pairing potential does not vanish in the squared ensemble. Because of the pseudo-Hermiticity of $H_\tau(\mathbf{k})\equiv \tau_3H_f^2(\mathbf{k})$ mentioned in Sec.\,\ref{mapping2}, i.e., $\tau_3 H_\tau^\dagger(\mathbf{k})\tau_3=H_\tau(\mathbf{k})$, the (anti)commutation relations between the metric $\tau_3$ and the three internal classifying
symmetries are crucial to the classification of free bosons.

\begin{center}
\textit{Case 1: Irreducible ensemble} $\{H_f(\mathbf{k})\}$
\end{center}

Let us first focus on irreducible ensembles $\{H_f(\mathbf{k})\}$. Because of the 
built-in particle-hole constraint with $U_C=\tau_1$ and $U_CU_C^*=\mathds{1}$, the symmetry class of $\{H_f(\mathbf{k})\}$ can only be \{BDI, D, DIII\}.

\medskip

\noindent
\underline{Class D}:
For irreducible $H_f(\mathbf{k})$ in class D, the only classifying condition is the build-in particle-hole constraint $U_C^\dagger H_f^*(\mathbf{-k})U_C=- H_f(\mathbf{k})$,  with $U_CU_C^*=\mathds{1}$ and $\{\tau_3, U_C\}=0$. After squaring, we have
\begin{align}\label{PHS_anticommuting}
U_C^\dagger H_\tau^*(\mathbf{-k})U_C=U_C^\dagger \tau_3 H_f^{*2}(\mathbf{-k})U_C=-H_\tau(\mathbf{k}),
\end{align}
with $U_C^\dagger\mathcal{K}$ the usual build-in particle-hole constraint of $H_\tau(\mathbf{k})$, satisfying the skew-pseudo-unitary condition $U_C\tau_3U_C^\dagger=-\tau_3$. If the squared ensemble has no emergent symmetries, we can conclude that $H_\tau(\mathbf{k})$ also belongs to class D. Here, we have adopted the same nomenclature for pseudo-Hermitian symmetry classes as that for Hermitian ones [see Table \ref{table:names} and Eq.\,(\ref{DAZconditions})]\cite{nomenclature}. 
However, we should always keep in mind that $H_\tau(\mathbf{k})$ is subject to the pseudo-Hermitian condition. After continuously deforming $H_\tau(\mathbf{k})$ to a Hermitian matrix\cite{Kawabata}, the classifying space of $H_\tau(\mathbf{k})$ is revealed as $\mathcal{C}_0$. 

\medskip

\begin{table*}[]
	\centering 
	\caption{Particle non-conserving free fermions $H_f(\mathbf{k})$ under the squaring map, where $H_f(\mathbf{k})$ is irreducible. The left major column corresponds to free fermions $H_f(\mathbf{k})$,
	whereas the right major column corresponds to free bosons $H_\tau(\mathbf{k})\equiv\tau_3H_f^2(\mathbf{k})$, 
	obtained by squaring.
	The three minor columns of fermionic symmetries $\mathcal{T}\equiv U^\dagger_T\mathcal{K}$, $\mathcal{C}\equiv U^\dagger_C\mathcal{K}=\tau_1\mathcal{K}$ and $\mathcal{S}=U_C^\dagger U^{\;}_T=\tau_1 U^{\;}_T$ specify their (anti)commutation relations to the metric $\tau_3$, while the three minor columns of bosonic symmetries $\mathcal{T}$, $\mathcal{C}$ and $\mathcal{S}$ are inherited from free fermions, specifying their (anti)commutation relations to the metric $\tau_3$ as well. For $H_f(\mathbf{k})$ with chiral symmetry $\mathcal{S}$ and $[\tau_3, \mathcal{S}]=0$, a unitary commuting symmetry of $H_\tau(\mathbf{k})$ exists.}
	\begin{tabular*}{\textwidth}{cccccccccc}
		\hline		\hline
		\multirow{2}[0]{*}{$H_f(\mathbf{k})$} & \multirow{2}[0]{*}{$\mathcal{T}$}   & \multirow{2}[0]{*}{$\mathcal{C}$}   & \multirow{2}[0]{*}{$\mathcal{S}$} & \multicolumn{1}{c|}{Classifying} & \multirow{2}[0]{*}{$H_\tau(\mathbf{k})$} & \multirow{2}[0]{*}{$\mathcal{T}$}   & \multirow{2}[0]{*}{$\mathcal{C}$}   & \multirow{2}[0]{*}{$\mathcal{S}$} & Classifying \\
		&    &    &  & \multicolumn{1}{c|}{space} &  &    &    &  & space \\
		\hline
		\multirow{2}[1]{*}{BDI} & $\{\tau_3, U^{\;}_T\}=0$ & \multirow{2}[1]{*}{$\{\tau_3, U_C\}=0$} & \multicolumn{1}{c}{$[\tau_3, U_C^\dagger U^{\;}_T]=0$} & \multicolumn{1}{c|}{\multirow{2}[1]{*}{$\mathcal{R}_1$}} & \multicolumn{5}{c}{\footnotesize A unitary commuting symmetry of $H_\tau(\mathbf{k})$ exists.}\\
		& $[\tau_3, U^{\;}_T]=0$ &       & $\{\tau_3, U_C^\dagger U^{\;}_T\}=0$ & \multicolumn{1}{c|}{} & BDI & $[\tau_3, U^{\;}_T]=0$ & $\{\tau_3, U_C\}=0$ & $\{\tau_3, U_C^\dagger U^{\;}_T\}=0$ & $\mathcal{R}_0$ \\
		& & & & \multicolumn{1}{c|}{} & & & & \multicolumn{1}{c}{} & \\
		D & - & $\{\tau_3, U_C\}=0$ & - & \multicolumn{1}{c|}{$\mathcal{R}_2$} & D & - & $\{\tau_3, U_C\}=0$ & - & $\mathcal{C}_0$ \\
		& & & & \multicolumn{1}{c|}{} & & & & \multicolumn{1}{c}{} & \\
		\multirow{2}[1]{*}{DIII} & $\{\tau_3, U^{\;}_T\}=0$ & \multirow{2}[1]{*}{$\{\tau_3, U_C\}=0$} & \multicolumn{1}{c}{$[\tau_3, U_C^\dagger U^{\;}_T]=0$} & \multicolumn{1}{c|}{\multirow{2}[1]{*}{$\mathcal{R}_3$}} & \multicolumn{5}{c}{\footnotesize A unitary commuting symmetry of $H_\tau(\mathbf{k})$ exists.}\\
		& $[\tau_3, U^{\;}_T]=0$ &       & $\{\tau_3, U_C^\dagger U^{\;}_T\}=0$ & \multicolumn{1}{c|}{} & DIII & $[\tau_3, U^{\;}_T]=0$ & $\{\tau_3, U_C\}=0$ & $\{\tau_3, U_C^\dagger U^{\;}_T\}=0$ & $\mathcal{R}_4$ \\
		\hline		\hline
	\end{tabular*}
	\label{tab:Particle_non-conserving_squared_irreducible}
\end{table*}

\noindent
\underline{Class BDI}:
In addition to the particle-hole constraint, we 
have $U_T^\dagger H_f^*(\mathbf{-k})U^{\;}_T= H_f(\mathbf{k})$ and $U^{\;}_TU_T^*=\mathds{1}$. Following the discussion of Sec.\,\ref{SS}, we need to consider two cases, either $\{\tau_3, U^{\;}_T\}=0$ or $[\tau_3, U^{\;}_T]=0$. 

\smallskip

(i) $\{\tau_3, U^{\;}_T\}=0\ $---
In this case,  \(U_C^\dagger U^{\;}_T\) is both unitary and pseudo-unitary. After squaring, we have 
\begin{align}\label{TRS_anticommuting}
U_T^\dagger H_\tau^*(\mathbf{-k})U^{\;}_T=U_T^\dagger\tau_3 H_f^{*2}(\mathbf{-k})U^{\;}_T=-H_\tau(\mathbf{k}),
\end{align}
which means $U^\dagger_T\mathcal{K}$ is a particle-hole symmetry of $H_\tau(\mathbf{k})$. 
Together with Eq.\,(\ref{PHS_anticommuting}), we have $U_T^\dagger U_C H_\tau(\mathbf{k})U_C^\dagger U^{\;}_T=H_\tau(\mathbf{k})$, which means $U_C^\dagger U^{\;}_T\equiv U^{\;}_S$ is a unitary transformation that commutes with $H_\tau(\mathbf{k})$,
\begin{align}\label{Chiral_commuting}
U_S^\dagger H_\tau(\mathbf{k})U^{\;}_S=U_S^\dagger \tau_3 H_f^2(\mathbf{k})U^{\;}_S=H_\tau(\mathbf{k}).
\end{align}
One cannot draw further conclusions about this squared ensemble without reducing away this symmetry first.

\smallskip

(ii) $[\tau_3, U^{\;}_T]=0\ $---
In this case, \(U^{\;}_T\) is both unitary and pseudo-unitary. After squaring, we have
\begin{align}\label{TRS_commuting}
U_T^\dagger H_\tau^*(\mathbf{-k})U^{\;}_T=U_T^\dagger\tau_3 H_f^{*2}(\mathbf{-k})U^{\;}_T=H_\tau(\mathbf{k}),
\end{align}
which means $U^\dagger_T\mathcal{K}$ is a time-reversal symmetry of $H_\tau(\mathbf{k})$. If we start from the
the chiral symmetry of $H_f(\mathbf{k})$, i.e., $U_S^\dagger H_f(\mathbf{k})U^{\;}_S=- H_f(\mathbf{k})$, because
$\{\tau_3, U^{\;}_S\}=0$ we find 
\begin{align}\label{Chiral_anticommuting}
U_S^\dagger H_\tau(\mathbf{k})U^{\;}_S=U_S^\dagger \tau_3 H_f^2(\mathbf{k})U^{\;}_S=-H_\tau(\mathbf{k}),
\end{align}
which means $U^{\;}_S$ is also a chiral symmetry of $H_\tau(\mathbf{k})$. 
Taking together Eq.\,\eqref{PHS_anticommuting} and \eqref{TRS_commuting}, we have $H_\tau(\mathbf{k})$ belonging to class BDI, with the classifying space being $\mathcal{R}_0$. 

\medskip

\noindent
\underline{Class DIII}: The only difference between this class and BDI is that \(U^{\;}_TU_T^*=-\mathds{1}\). 
Hence, the analysis of how the classifying conditions descend to the squared ensemble is the same as above. If $\{\tau_3, U^{\;}_T\}=0$, we again obtain Eq.\,(\ref{TRS_anticommuting}) and Eq.\,(\ref{Chiral_commuting})
and one must block-diagonalize away the unitary symmetry \(U^{\;}_S\) before one can proceed. If $[\tau_3, U^{\;}_T]=0$, we again obtain Eq.\,(\ref{TRS_commuting}) and Eq.\,(\ref{Chiral_anticommuting}) and, because of $U^{\;}_TU_T^*=-\mathds{1}$, we now have $H_\tau(\mathbf{k})$ belonging to class DIII, with the classifying space being $\mathcal{R}_4$. 

These analyses of squared ensembles associated to irreducible ensembles of BdG Hamiltonians are summarized in Table \ref{tab:Particle_non-conserving_squared_irreducible}, with the right major column associated to symmetry classes \{BDI, D, DIII\} of $H_\tau(\mathbf{k})$. Notice that the topological classification of classes \{BDI, D, DIII\} of $H_\tau(\mathbf{k})$ is known in the literature (see for example Sec.\,VIII\,A of Ref.\,[\onlinecite{Kawabata}])\cite{nomenclature}, however, as it turns out, Table \ref{tab:Particle_non-conserving_squared_irreducible} is just a sub-table of a more general Table \ref{tab:Particle_non-conserving_squared_reducible} and more symmetry classes, besides \{BDI, D, DIII\}, of $H_\tau(\mathbf{k})$ could emerge when we analyze reducible ensembles $\{H_f(\mathbf{k})\}$, as we do next. 

\begin{table*}[]
    \centering 
    \caption{Particle non-conserving free fermions $H_f(\mathbf{k})$ under the squaring map, where $H_f(\mathbf{k})$ is reducible but with irreducible blocks $\widetilde{H}_f(\mathbf{k})$. The left major column 
    corresponds to free fermions $\widetilde{H}_f(\mathbf{k})$,
    whereas the right major column 
    corresponds to free bosons $\tau_3\widetilde{H}_f^2(\mathbf{k})\equiv\widetilde{H}_\tau(\mathbf{k})$, obtained by squaring, 
    with the balanced metric $\tau_3$ considered. The three minor columns of fermionic symmetries $\mathcal{T}=U^\dagger_T\mathcal{K}$, $\mathcal{C}=U^\dagger_C\mathcal{K}$ and $\mathcal{S}=U^\dagger_S$ specify their (anti)commutation relations to the metric $\tau_3$, while the three minor columns of bosonic symmetries $\mathcal{T}$, $\mathcal{C}$ and $\mathcal{S}$ are inherited from free fermions, specifying their (anti)commutation relations to the metric $\tau_3$ as well. For example, for class AI of $\widetilde{H}_f(\mathbf{k})$ with $\mathcal{T}=U^\dagger_T\mathcal{K}$ and $\{\tau_3, U^{\;}_T\}=0$, after squaring $U^\dagger_T\mathcal{K}$ is a particle-hole symmetry of $\widetilde{H}_\tau(\mathbf{k})$ rather than a time-reversal symmetry of $\widetilde{H}_\tau(\mathbf{k})$, so $\{\tau_3, U^{\;}_T\}=0$ is specified under the minor column $\mathcal{C}$ of the right major column. 
    For $\widetilde{H}_f(\mathbf{k})$ with chiral symmetry $\mathcal{S} (=U^\dagger_S$ or $U_C^\dagger U^{\;}_T)$ and $[\tau_3, \mathcal{S}]=0$, a unitary commuting symmetry of $\widetilde{H}_\tau(\mathbf{k})$ exists. Note that the special case corresponding to Eq.\,(\ref{Unitary_anticommuting}) and Eq.\,(\ref{PHS_commuting}) is not listed in this table.}
	\begin{tabular*}{\textwidth}{cccccccccc}
		\hline		\hline
		\multirow{2}[0]{*}{$\widetilde{H}_f(\mathbf{k})$} & \multirow{2}[0]{*}{$\mathcal{T}$}   & \multirow{2}[0]{*}{$\mathcal{C}$}   & \multirow{2}[0]{*}{$\mathcal{S}$} & \multicolumn{1}{c|}{Classifying} & \multirow{2}[0]{*}{$\widetilde{H}_\tau(\mathbf{k})$} & \multirow{2}[0]{*}{$\mathcal{T}$}   & \multirow{2}[0]{*}{$\mathcal{C}$}   & \multirow{2}[0]{*}{$\mathcal{S}$} & Classifying \\
		 &    &    &  & \multicolumn{1}{c|}{space} &  &    &    &  & space \\
		\hline
& & & & \multicolumn{1}{c|}{} & & & & \multicolumn{1}{c}{} & \\
		A     & -     & -     & - & \multicolumn{1}{c|}{$\mathcal{C}_0$} & A     & -     & -     & \multicolumn{1}{c}{-} & $\mathcal{C}_0$ \\
& & & & \multicolumn{1}{c|}{} & & & & \multicolumn{1}{c}{} & \\
& & & & \multicolumn{1}{c|}{} & & & & \multicolumn{1}{c}{} & \\
		\multirow{2}[2]{*}{AIII} & \multirow{2}[2]{*}{-} & \multirow{2}[2]{*}{-} & \multicolumn{1}{c}{$\{\tau_3, U^{\;}_S\}=0$} &\multicolumn{1}{c|}{\multirow{2}[2]{*}{$\mathcal{C}_1$}} & AIII & - & - & \multicolumn{1}{c}{$\{\tau_3, U^{\;}_S\}=0$} & $\mathcal{C}_0$ \\
		&       &       & \multicolumn{1}{c}{$[\tau_3, U^{\;}_S]=0$} &\multicolumn{1}{c|}{} & \multicolumn{5}{c}{\footnotesize A unitary commuting symmetry of $\widetilde{H}_\tau(\mathbf{k})$ exists.}\\
& & & & \multicolumn{1}{c|}{} & & & & \multicolumn{1}{c}{} & \\
& & & & \multicolumn{1}{c|}{} & & & & \multicolumn{1}{c}{} & \\
		\multirow{2}[2]{*}{AI} & $\{\tau_3, U^{\;}_T\}=0$ & \multirow{2}[2]{*}{-} & \multirow{2}[2]{*}{-} & \multicolumn{1}{c|}{\multirow{2}[2]{*}{$\mathcal{R}_0$}} & D & - & $\{\tau_3, U^{\;}_T\}=0$ & \multicolumn{1}{c}{-} & $\mathcal{C}_0$\\
		& $[\tau_3, U^{\;}_T]=0$ &       & & \multicolumn{1}{c|}{} & AI    & $[\tau_3, U^{\;}_T]=0$ & -     & \multicolumn{1}{c}{-} & $\mathcal{R}_0$ \\
& & & & \multicolumn{1}{c|}{} & & & & \multicolumn{1}{c}{} & \\
& & & & \multicolumn{1}{c|}{} & & & & \multicolumn{1}{c}{} & \\
		\multirow{4}[2]{*}{BDI} & $\{\tau_3, U^{\;}_T\}=0$ & \multirow{2}[1]{*}{$\{\tau_3, U_C\}=0$} & \multicolumn{1}{c}{$[\tau_3, U_C^\dagger U^{\;}_T]=0$} & \multicolumn{1}{c|}{} & \multicolumn{5}{c}{\footnotesize A unitary commuting symmetry of $\widetilde{H}_\tau(\mathbf{k})$ exists.}\\
		& $[\tau_3, U^{\;}_T]=0$ &       & \multicolumn{1}{c}{\multirow{2}[0]{*}{$\{\tau_3, U_C^\dagger U^{\;}_T\}=0$}} & \multicolumn{1}{c|}{\multirow{2}[0]{*}{$\mathcal{R}_1$}} & \multicolumn{1}{c}{\multirow{2}[0]{*}{BDI}} & $[\tau_3, U^{\;}_T]=0$ & $\{\tau_3, U_C\}=0$ & \multicolumn{1}{c}{\multirow{2}[0]{*}{$\{\tau_3, U_C^\dagger U^{\;}_T\}=0$}} & \multirow{2}[0]{*}{$\mathcal{R}_0$} \\
		& $\{\tau_3, U^{\;}_T\}=0$ & \multirow{2}[1]{*}{$[\tau_3, U_C]=0$} & & \multicolumn{1}{c|}{} &       & $[\tau_3, U_C]=0$ & $\{\tau_3, U^{\;}_T\}=0$ & \multicolumn{1}{c}{} &\\
		& $[\tau_3, U^{\;}_T]=0$ &     & $[\tau_3, U_C^\dagger U^{\;}_T]=0$  & \multicolumn{1}{c|}{} & \multicolumn{5}{c}{\footnotesize A unitary commuting symmetry of $\widetilde{H}_\tau(\mathbf{k})$ exists.} \\
& & & & \multicolumn{1}{c|}{} & & & & \multicolumn{1}{c}{} & \\
& & & & \multicolumn{1}{c|}{} & & & & \multicolumn{1}{c}{} & \\
		\multirow{2}[2]{*}{D} & \multirow{2}[2]{*}{-} & $\{\tau_3, U_C\}=0$ & \multicolumn{1}{c}{\multirow{2}[2]{*}{-}} & \multicolumn{1}{c|}{\multirow{2}[2]{*}{$\mathcal{R}_2$}} & D     & -     & $\{\tau_3, U_C\}=0$ & \multicolumn{1}{c}{-} & $\mathcal{C}_0$ \\
		&       & $[\tau_3, U_C]=0$ & & \multicolumn{1}{c|}{} & AI    & $[\tau_3, U_C]=0$ & -     & \multicolumn{1}{c}{-} & $\mathcal{R}_0$ \\
& & & & \multicolumn{1}{c|}{} & & & & \multicolumn{1}{c}{} & \\
& & & & \multicolumn{1}{c|}{} & & & & \multicolumn{1}{c}{} & \\
		\multirow{4}[2]{*}{DIII} & $\{\tau_3, U^{\;}_T\}=0$ & \multirow{2}[1]{*}{$\{\tau_3, U_C\}=0$} & \multicolumn{1}{c}{$[\tau_3, U_C^\dagger U^{\;}_T]=0$} & \multicolumn{1}{c|}{} & \multicolumn{5}{c}{\footnotesize A unitary commuting symmetry of $\widetilde{H}_\tau(\mathbf{k})$ exists.} \\
		& $[\tau_3, U^{\;}_T]=0$ &       & \multicolumn{1}{c}{\multirow{2}[0]{*}{$\{\tau_3, U_C^\dagger U^{\;}_T\}=0$}} & \multicolumn{1}{c|}{\multirow{2}[0]{*}{$\mathcal{R}_3$}} & DIII & $[\tau_3, U^{\;}_T]=0$ & $\{\tau_3, U_C\}=0$ & \multicolumn{1}{c}{\multirow{2}[0]{*}{$\{\tau_3, U_C^\dagger U^{\;}_T\}=0$}} & $\mathcal{R}_4$ \\
		& $\{\tau_3, U^{\;}_T\}=0$ & \multirow{2}[1]{*}{$[\tau_3, U_C]=0$} & & \multicolumn{1}{c|}{} & CI    & $[\tau_3, U_C]=0$ & $\{\tau_3, U^{\;}_T\}=0$ & \multicolumn{1}{c}{} & $\mathcal{R}_0$ \\
		& $[\tau_3, U^{\;}_T]=0$ &       & \multicolumn{1}{c}{$[\tau_3, U_C^\dagger U^{\;}_T]=0$} & \multicolumn{1}{c|}{} & \multicolumn{5}{c}{\footnotesize A unitary commuting symmetry of $\widetilde{H}_\tau(\mathbf{k})$ exists.} \\
& & & & \multicolumn{1}{c|}{} & & & & \multicolumn{1}{c}{} & \\
& & & & \multicolumn{1}{c|}{} & & & & \multicolumn{1}{c}{} & \\
		\multirow{2}[2]{*}{AII} & $\{\tau_3, U^{\;}_T\}=0$ & \multirow{2}[2]{*}{-} & \multicolumn{1}{c}{\multirow{2}[2]{*}{-}} & \multicolumn{1}{c|}{\multirow{2}[2]{*}{$\mathcal{R}_4$}} & C & -     & $\{\tau_3, U^{\;}_T\}=0$ & \multicolumn{1}{c}{-} & $\mathcal{C}_0$ \\
		& $[\tau_3, U^{\;}_T]=0$ &       & & \multicolumn{1}{c|}{} & AII   & $[\tau_3, U^{\;}_T]=0$ & -     & \multicolumn{1}{c}{-} & $\mathcal{R}_4$ \\
& & & & \multicolumn{1}{c|}{} & & & & \multicolumn{1}{c}{} & \\
& & & & \multicolumn{1}{c|}{} & & & & \multicolumn{1}{c}{} & \\
		\multirow{4}[2]{*}{CII} & $\{\tau_3, U^{\;}_T\}=0$ & \multirow{2}[1]{*}{$\{\tau_3, U_C\}=0$} & \multicolumn{1}{c}{$[\tau_3, U_C^\dagger U^{\;}_T]=0$} & \multicolumn{1}{c|}{} & \multicolumn{5}{c}{\footnotesize A unitary commuting symmetry of $\widetilde{H}_\tau(\mathbf{k})$ exists.} \\
		& $[\tau_3, U^{\;}_T]=0$ &       & \multicolumn{1}{c}{\multirow{2}[0]{*}{$\{\tau_3, U_C^\dagger U^{\;}_T\}=0$}} & \multicolumn{1}{c|}{\multirow{2}[0]{*}{$\mathcal{R}_5$}} & \multirow{2}[0]{*}{CII} & $[\tau_3, U^{\;}_T]=0$ & $\{\tau_3, U_C\}=0$ & \multicolumn{1}{c}{\multirow{2}[0]{*}{$\{\tau_3, U_C^\dagger U^{\;}_T\}=0$}} & \multirow{2}[0]{*}{$\mathcal{R}_4$} \\
		& $\{\tau_3, U^{\;}_T\}=0$ & \multirow{2}[1]{*}{$[\tau_3, U_C]=0$} & & \multicolumn{1}{c|}{} &       & $[\tau_3, U_C]=0$ & $\{\tau_3, U^{\;}_T\}=0$ & \multicolumn{1}{c}{} & \\
		& $[\tau_3, U^{\;}_T]=0$ &       & \multicolumn{1}{c}{$[\tau_3, U_C^\dagger U^{\;}_T]=0$} & \multicolumn{1}{c|}{} & \multicolumn{5}{c}{\footnotesize A unitary commuting symmetry of $\widetilde{H}_\tau(\mathbf{k})$ exists.} \\
& & & & \multicolumn{1}{c|}{} & & & & \multicolumn{1}{c}{} & \\
& & & & \multicolumn{1}{c|}{} & & & & \multicolumn{1}{c}{} & \\
		\multirow{2}[2]{*}{C} & \multirow{2}[2]{*}{-} & $\{\tau_3, U_C\}=0$ & \multicolumn{1}{c}{\multirow{2}[2]{*}{-}} & \multicolumn{1}{c|}{\multirow{2}[2]{*}{$\mathcal{R}_6$}} & C     & -     & $\{\tau_3, U_C\}=0$ & \multicolumn{1}{c}{-} & $\mathcal{C}_0$ \\
		&       & $[\tau_3, U_C]=0$ & & \multicolumn{1}{c|}{} & AII   & $[\tau_3, U_C]=0$ & -     & \multicolumn{1}{c}{-} & $\mathcal{R}_4$ \\
& & & & \multicolumn{1}{c|}{} & & & & \multicolumn{1}{c}{} & \\
& & & & \multicolumn{1}{c|}{} & & & & \multicolumn{1}{c}{} & \\
		\multirow{4}[2]{*}{CI} & $\{\tau_3, U^{\;}_T\}=0$ & \multirow{2}[1]{*}{$\{\tau_3, U_C\}=0$} & \multicolumn{1}{c}{$[\tau_3, U_C^\dagger U^{\;}_T]=0$} & \multicolumn{1}{c|}{} & \multicolumn{5}{c}{\footnotesize A unitary commuting symmetry of $\widetilde{H}_\tau(\mathbf{k})$ exists.} \\
		& $[\tau_3, U^{\;}_T]=0$ &       & \multicolumn{1}{c}{\multirow{2}[0]{*}{$\{\tau_3, U_C^\dagger U^{\;}_T\}=0$}} & \multicolumn{1}{c|}{\multirow{2}[0]{*}{$\mathcal{R}_7$}} & CI    & $[\tau_3, U^{\;}_T]=0$ & $\{\tau_3, U_C\}=0$ & \multicolumn{1}{c}{\multirow{2}[0]{*}{$\{\tau_3, U_C^\dagger U^{\;}_T\}=0$}} & $\mathcal{R}_0$ \\
		& $\{\tau_3, U^{\;}_T\}=0$ & \multirow{2}[1]{*}{$[\tau_3, U_C]=0$} & & \multicolumn{1}{c|}{} & DIII  & $[\tau_3, U_C]=0$ & $\{\tau_3, U^{\;}_T\}=0$ & \multicolumn{1}{c}{} & $\mathcal{R}_4$\\
		& $[\tau_3, U^{\;}_T]=0$ &       & \multicolumn{1}{c}{$[\tau_3, U_C^\dagger U^{\;}_T]=0$} & \multicolumn{1}{c|}{} & \multicolumn{5}{c}{\footnotesize A unitary commuting symmetry of $\widetilde{H}_\tau(\mathbf{k})$ exists.} \\
& & & & \multicolumn{1}{c|}{} & & & & \multicolumn{1}{c}{} & \\
		\hline		\hline
	\end{tabular*}
    \label{tab:Particle_non-conserving_squared_reducible}
\end{table*}

\begin{center}
\textit{Case 2: Reducible ensemble} $\{H_f(\mathbf{k})\}$
\end{center}

Next we tackle the reducible ensembles $\{H_f(\mathbf{k})\}$.
That is, $H_f(\mathbf{k})$ with unitary commuting symmetry $U^\dagger H_f(\mathbf{k})U=H_f(\mathbf{k})$. From the discussion of Sec.\,\ref{SS}, we have either $\{\tau_3, U\}=0$ or $[\tau_3, U]=0$.

\smallskip

(i) \(\{\tau_3,U\}=0\)\ --- In this case, \(U\) is both unitary and skew-pseudo-unitary. After squaring, we have 
\begin{align}\label{Unitary_anticommuting}
U^\dagger H_\tau(\mathbf{k})U=U^\dagger \tau_3 H_f^2(\mathbf{k})U=-H_\tau(\mathbf{k}),
\end{align}
which means $U$ is a chiral symmetry of $H_\tau(\mathbf{k})$. 
Together with Eq.\,(\ref{PHS_anticommuting}), we have $U^\dagger\tau_1 H_\tau^*(\mathbf{-k})\tau_1 U=H_\tau(\mathbf{k})$ or, equivalently, we have $U_C^\dagger H_\tau^*(\mathbf{-k})U_C=H_\tau(\mathbf{k})$ with $U_C\equiv \tau_1 U$. As a self-consistency check, we can start from the non-built-in particle-hole symmetry of $H_f(\mathbf{k})$, i.e., $U_C^\dagger H_f^*(\mathbf{-k})U_C=- H_f(\mathbf{k})$, with $[\tau_3, U_C]=0$; we have
\begin{align}
\label{PHS_commuting}
U_C^\dagger H_\tau^*(\mathbf{-k})U_C=U_C^\dagger \tau_3 H_f^{*2}(\mathbf{-k})U_C=H_\tau(\mathbf{k}),
\end{align}
which means $U^\dagger_C\mathcal{K}$ is a time-reversal symmetry of $H_\tau(\mathbf{k})$, satisfying the pseudo-unitary condition $U_C\tau_3U_C^\dagger=\tau_3$. Together with Eq.\,(\ref{PHS_anticommuting}), we have $H_\tau(\mathbf{k})$ belonging to either class BDI (if $U_CU_C^*=\mathds{1}$), as obtained during the analysis of \underline{Class BDI} in \textit{Case 1}, or DIII (if $U_CU_C^*=-\mathds{1}$), as obtained during the analysis of \underline{Class DIII} in \textit{Case 1}.

\smallskip
(ii) \([\tau_3,U]=0\)\ --- In this case, \(U\) is both unitary and pseudo-unitary. After squaring, we have 
\begin{align}\label{Unitary_commuting}
U^\dagger H_\tau(\mathbf{k})U=U^\dagger \tau_3 H_f^2(\mathbf{k})U=H_\tau(\mathbf{k}),
\end{align}
with $U$ a unitary commuting symmetry of $H_\tau(\mathbf{k})$. 
Then, the block-diagonalization of $H_f(\mathbf{k})$ implies that of $H_\tau(\mathbf{k})$, and we only need to focus on irreducible blocks.
When dealing with such irreducible blocks,
there will also be a set of symmetry constraints like Eqs.\,(\ref{PHS_anticommuting})-(\ref{Chiral_anticommuting}) and Eq.\,(\ref{PHS_commuting}) imposed on them. 
In the analysis of irreducible blocks below, we refer to such equations  
as {\em symmetry constraints imposed directly on the blocks}. Furthermore, since the topological classification of non-Hermitian Bloch Hamiltonians includes pseudo-Hermiticity with respect to an indefinite metric as a classifying condition, as we saw in Sec.\,\ref{SS}, the metric appropriate for defining the block effective BdG Hamiltonian need not be \(\tau_3=\sigma_3\otimes\mathds{1}_m\), 
but can have instead the more complicated structure $\tau_{m,n}$ in Table\,\ref{squared_subensembles}. In addition, blocks could emerge with
vanishing pairing. Whether any of these possibilities actually occur is controlled by spectral features of the symmetry that is being block-diagonalized together with the ensemble. The symmetries of spin rotations and lattice translations, in particular, do not induce these exotic blocks: all blocks consist of effective BdG Hamiltonians that are pseudo-Hermitian with respect to a ``balanced" metric \(\tau_3=\sigma_3\otimes\mathds{1}_m\). While in the following we assume that the metric is always of this ``balanced'' form,  it is interesting to notice that more exotic scenarios can also be realized in free-boson systems.

\smallskip
\noindent
\underline{Class A}: 
For an irreducible block of $H_f(\mathbf{k})$ in class A, there are no symmetry constraints and therefore, except for pseudo-Hermiticity, no symmetry constraints imposed on the corresponding block of $H_\tau(\mathbf{k})$ either. Hence, this block of $H_\tau(\mathbf{k})$ belongs to class A, with the classifying space mathematically\cite{Kawabata} being $\mathcal{C}_0\times\mathcal{C}_0$, where each $\mathcal{C}_0$ independently describes either the positive or the negative energy bands. However, since as seen in Sec.\,\ref{mapping2} the negative energy bands cannot be occupied by bosons, one only needs to consider the positive-energy bands. Thus, physically, the classifying space for this block of $H_\tau(\mathbf{k})$ in class A should be $\mathcal{C}_0$. Here we adopt this physical point of view. Accordingly, hereinafter, the classifying space for class AI is $\mathcal{R}_0$ (rather than $\mathcal{R}_0\times\mathcal{R}_0$), and that for class AII is $\mathcal{R}_4$ (rather than $\mathcal{R}_4\times\mathcal{R}_4$).

\smallskip
\noindent
\underline{Class AIII}:
Now there is only a chiral symmetry imposed on the block of $H_f(\mathbf{k})$. Therefore, besides pseudo-Hermiticity, analysis like Eq.\,(\ref{Chiral_commuting}) will lead to a unitary commuting symmetry imposed on the block of $H_\tau(\mathbf{k})$, and, similarly, analysis like Eq.\,(\ref{Chiral_anticommuting}) will lead to a chiral symmetry imposed on the block of $H_\tau(\mathbf{k})$, with the block belonging to class AIII and the classifying space being $\mathcal{C}_0$.

\smallskip
\noindent
\underline{Class AI}: There is only a time-reversal symmetry imposed on the block of $H_f(\mathbf{k})$. So, apart from pseudo-Hermiticity, analysis like Eq.\,(\ref{TRS_anticommuting}) will lead to a particle-hole symmetry imposed on the block of $H_\tau(\mathbf{k})$, and the block belongs to class D, with the classifying space being $\mathcal{C}_0$. Similarly, besides pseudo-Hermiticity, analysis like Eq.\,(\ref{TRS_commuting}) will lead to a time-reversal symmetry imposed on the block of $H_\tau(\mathbf{k})$, and the block belongs to class AI, with the classifying space being $\mathcal{R}_0$. 

\smallskip
\noindent
\underline{Class BDI}:
Now in addition to the two possibilities with $\{\tau_3, U_C\}=0$ analyzed in \textit{Case 1} of \underline{Class BDI}, two more possibilities with $[\tau_3, U_C]=0$ also arise. If $[\tau_3, U_C]=0$ and $\{\tau_3, U^{\;}_T\}=0$, analysis based on Eq.\,(\ref{PHS_commuting}) and Eq.\,(\ref{TRS_anticommuting}) will lead to a time-reversal symmetry and a particle-hole symmetry imposed on the block of $H_\tau(\mathbf{k})$, and the block belongs to class BDI, with the classifying space being $\mathcal{R}_0$. If $[\tau_3, U_C]=0$ and $[\tau_3, U^{\;}_T]=0$, analysis based on Eq.\,(\ref{PHS_commuting}) and Eq.\,(\ref{TRS_commuting}) will lead to two time-reversal symmetries imposed on the block of $H_\tau(\mathbf{k})$, and therefore a unitary commuting symmetry of the block.

\smallskip
\noindent
\underline{Class D}:
Besides the possibility with $\{\tau_3, U_C\}=0$ analyzed in \textit{Case 1} of \underline{Class D}, one more possibility with $[\tau_3, U_C]=0$ also arises. If $[\tau_3, U_C]=0$, analysis like Eq.\,(\ref{PHS_commuting}) will lead to a time-reversal symmetry imposed on the block of $H_\tau(\mathbf{k})$, and the block belongs to class AI, with the classifying space being $\mathcal{R}_0$. 

\smallskip
\noindent
\underline{Class DIII}:
 Besides the two possibilities with $\{\tau_3, U_C\}=0$ analyzed in \textit{Case 1} of \underline{Class DIII}, two more possibilities with $[\tau_3, U_C]=0$ also arise. If $[\tau_3, U_C]=0$ and $\{\tau_3, U^{\;}_T\}=0$, analysis like Eq.\,(\ref{PHS_commuting}) and Eq.\,(\ref{TRS_anticommuting}) will lead to a time-reversal symmetry and a particle-hole symmetry imposed on the block of $H_\tau(\mathbf{k})$, and the block belongs to class CI, with the classifying space being $\mathcal{R}_0$. If $[\tau_3, U_C]=0$ and $[\tau_3, U^{\;}_T]=0$, analysis like Eq.\,(\ref{PHS_commuting}) and Eq.\,(\ref{TRS_commuting}) will lead to two time-reversal symmetries imposed on the block of $H_\tau(\mathbf{k})$, and therefore a unitary commuting symmetry of the block.

\smallskip
\noindent
\underline{Class AII}:
The only difference between this class and AI is that \(U^{\;}_TU_T^*=-\mathds{1}\). If $\{\tau_3, U^{\;}_T\}=0$, we again obtain Eq.\,(\ref{TRS_anticommuting}) and, since $U^{\;}_TU_T^*=-\mathds{1}$, we now have the block of $H_\tau(\mathbf{k})$ belonging to class C, with the classifying space being $\mathcal{C}_0$. If $[\tau_3, U^{\;}_T]=0$, we again obtain Eq.\,(\ref{TRS_commuting}) and, since $U^{\;}_TU_T^*=-\mathds{1}$, we now have the block belonging to class AII, with the classifying space being $\mathcal{R}_4$. 

\smallskip
\noindent
\underline{Class CII}:
There are four possibilities. If $\{\tau_3, U_C\}=0$ and $\{\tau_3, U^{\;}_T\}=0$, because of Eqs.\,(\ref{PHS_anticommuting})-(\ref{TRS_anticommuting}), 
the block of $H_\tau(\mathbf{k})$ has a unitary commuting symmetry. If $\{\tau_3, U_C\}=0$ and $[\tau_3, U^{\;}_T]=0$, because of Eq.\,(\ref{PHS_anticommuting}) and Eq.\,(\ref{TRS_commuting}), the block of $H_\tau(\mathbf{k})$ belongs to CII, with the classifying space being $\mathcal{R}_4$. If $[\tau_3, U_C]=0$ and $\{\tau_3, U^{\;}_T\}=0$, because of Eq.\,(\ref{PHS_commuting}) and Eq.\,(\ref{TRS_anticommuting}), 
the block of $H_\tau(\mathbf{k})$ belongs to CII, with the classifying space being $\mathcal{R}_4$. If $[\tau_3, U_C]=0$ and $[\tau_3, U^{\;}_T]=0$, Eq.\,(\ref{PHS_commuting}) and Eq.\,(\ref{TRS_commuting}) lead to a unitary commuting symmetry of the block.

\smallskip
\noindent
\underline{Class C}:
There are two possibilities. If $\{\tau_3, U_C\}=0$, because of Eq.\,(\ref{PHS_anticommuting}), the block of $H_\tau(\mathbf{k})$ belongs to class C, with the classifying space being $\mathcal{C}_0$. If $[\tau_3, U_C]=0$, because of Eq.\,(\ref{PHS_commuting}), the block of $H_\tau(\mathbf{k})$ belongs to class AII, with the classifying space being $\mathcal{R}_4$. 

\smallskip
\noindent
\underline{Class CI}:
The only difference between this class and CII is that \(U^{\;}_TU_T^*=\mathds{1}\). If $\{\tau_3, U_C\}=0$ and $\{\tau_3, U^{\;}_T\}=0$, because of Eqs.\,(\ref{PHS_anticommuting})-\,(\ref{TRS_anticommuting}), the block of $H_\tau(\mathbf{k})$ has a unitary commuting symmetry. If $\{\tau_3, U_C\}=0$ and $[\tau_3, U^{\;}_T]=0$, because of Eq.\,(\ref{PHS_anticommuting}) and Eq.\,(\ref{TRS_commuting}), the block of $H_\tau(\mathbf{k})$ belongs to CI, with the classifying space being $\mathcal{R}_0$. If $[\tau_3, U_C]=0$ and $\{\tau_3, U^{\;}_T\}=0$, because of Eq.\,(\ref{PHS_commuting}) and Eq.\,(\ref{TRS_anticommuting}), the block of $H_\tau(\mathbf{k})$ belongs to DIII, with the classifying space being $\mathcal{R}_4$. If $[\tau_3, U_C]=0$ and $[\tau_3, U^{\;}_T]=0$, Eq.\,(\ref{PHS_commuting}) and Eq.\,(\ref{TRS_commuting}) lead to a unitary commuting symmetry of the block.

\smallskip

The results of the above analysis are summarized in Table\,\ref{tab:Particle_non-conserving_squared_reducible}, based on the  (anti)commutation relations between $\tau_3$ and three internal symmetries of the irreducible blocks of $H_f(\mathbf{k})$.
As we can see, the table reproduces the results of Table \ref{tab:Particle_non-conserving_squared_irreducible}, but with additional possibilities: 
besides the three symmetry classes \{A, AI, AII\} arising when squaring a particle-conserving free fermion (Table\,\ref{tab:Particle_conserving_squared}), all the other seven symmetry classes also appear.

As an example, if we return to the squared BCS model examined in Sec.\,\ref{SS} and start with Eq.\,\eqref{BCSsquared}, 
we see that $\widetilde{H}_f(\mathbf{k})$ has only the particle-hole symmetry $\sigma_2\mathcal{K}$ and belongs to class C. According to Table\,\ref{tab:Particle_non-conserving_squared_reducible}, since $\{\sigma_3, U_C\}=0$, $\widetilde{H}_\tau(\mathbf{k})$ will have the same particle-hole symmetry and belongs to class C as well.

\subsubsection{Assessment: The threefold way}
\label{Assessment}

\begin{table*}[t]
\centering
\caption{Periodic table for stable free-boson systems. 
For $d=0$ and $d=4$ all the classes are topologically non-trivial, while for $d=1$ and $d=5$ all of them are topologically trivial. }
\begin{tabular}{C{1cm} C{1cm} C{1.5cm} C{1.5cm} c  C{1.1cm} C{1.1cm} C{1.1cm} C{1.1cm} C{1.1cm} C{1.1cm} C{1.1cm} C{1cm}} 
\hline \hline
 &  $\mathcal{T}$  & $\mathcal{C}$ &  $\mathcal{S}$  & Classifying space & $d=0$ &  $d=1$  &$d=2$ &  $d=3$  & $d=4$ &$d=5$ &  $d=6$  & $d=7$\\ 
 \hline
 A    & & 0 & 0 &  &    & &   &  &   &  &   &  \\ 
 AIII    &  & 0 & 1 &  \multirow{1}{*}{\vspace*{-0.5cm} $\mathcal{C}_0$} &    &  &   &  &   &  &   &     \\ 
 D    & \vspace*{-0.4cm} 0 & 1 & 0  &  & \vspace*{-0.4cm}  $\mathbb{Z}$  &\vspace*{-0.4cm}  - & \vspace*{-0.4cm} $\mathbb{Z}$  & \vspace*{-0.4cm} - &\vspace*{-0.4cm} $\mathbb{Z}$  & \vspace*{-0.4cm} - & \vspace*{-0.4cm} $\mathbb{Z}$  & \vspace*{-0.4cm} -   \\ 
 C    &  & -1 & 0 &  &    &  & &  &   & &  &    \\ 
 \hline
 AI    &  & 0 & 0 &  &    &  &  &  &   &   &  &  \\ 
 BDI   & 1 & 1 & 1 & $\mathcal{R}_0$ &  $\mathbb{Z}$  & - & -  & - & $2\mathbb{Z}$  & - & $\mathbb{Z}_2$  & $\mathbb{Z}_2$     \\
 CI   &  & -1 & 1 &  &   &   &   &   &   &   &   &     \\ 
 \hline
 AII    &  & 0 & 0 &  &    &   &   &   &   &   &   &    \\ 
 DIII    & -1 & 1 & 1 & $\mathcal{R}_4$ &   $2\mathbb{Z}$  & - & $\mathbb{Z}_2$  & $\mathbb{Z}_2$ & $\mathbb{Z}$  & - & -  & -  \\ 
 CII    &  & -1 & 1 &  &   &  &   &  &   &  &   &    \\ 
\hline \hline
\end{tabular}
    \label{tab:Bosonic_periodic_table}
\end{table*} 

Although for particle-non-conserving free bosons under squaring 
we obtained all ten AZ symmetry classes, due to the pseudo-Hermiticity of (the irreducible blocks of) $H_\tau(\mathbf{k})$, \emph{only three classifying spaces \{$\mathcal{C}_0, \mathcal{R}_0, \mathcal{R}_4$\} appear}. More specifically, we find that AZ symmetry classes with no time-reversal symmetry correspond to the classifying space $\mathcal{C}_0$, classes with $\mathcal{T}^2=\mathds{1}$ correspond to $\mathcal{R}_0$, and classes with $\mathcal{T}^2=-\mathds{1}$ correspond to $\mathcal{R}_4$. Therefore, based on Table\,\ref{tab:Particle_conserving_squared} and Table\,\ref{tab:Particle_non-conserving_squared_reducible}, what emerges is a unified periodic table, Table\,\ref{tab:Bosonic_periodic_table},  for either particle-conserving or non-conserving free-boson systems obtained by application of the squaring map ${\mathscr S}$. Furthermore, with simple 
modifications, Table\,\ref{tab:Bosonic_periodic_table} holds for the case of ``unbalanced" metric $\tau_{m,n}$ (see Table\,\ref{squared_subensembles}) as well, with the associated symmetry classes being \{A, AI, AII\}. Thus, we conclude that the topological classification of free bosons under squaring  {\em depends only on the existence of time-reversal symmetry} and reduces to the ``threefold way" of Dyson\cite{Dyson}. 
 As a byproduct, we claim that Table\,\ref{tab:Bosonic_periodic_table} also holds for stable free-boson systems not arising from the squaring procedure (see Appendix\,\ref{simpleproof} for a simple proof). 

A careful examination of Table\,\ref{tab:Bosonic_periodic_table} suggests that the main difference between free fermions before squaring and free bosons after squaring is that all the symmetry classes of free bosons are topologically non-trivial for $d=0$ and $d=4$, while all of them are topologically trivial for $d=1$ and $d=5$. However, just like the fermionic tenfold classification, our threefold classification may fail when additional symmetries (other than $\mathcal{T}$, $\mathcal{C}$ or $\mathcal{S}$) are present. For example, the spinless Su-Schrieffer-Heeger model\cite{SSH} of a dimerized chain with non-zero chemical potential or with next-nearest-neighbor hopping has only time-reversal symmetry with $\mathcal{T}^2=\mathds{1}$, which is topologically trivial in one dimension according to the fermionic tenfold classification table. However, this is not correct because the Berry phase of the lowest band is precisely quantized to $0/\pi$ (mod $2\pi$). The same is true if we square this model. The reason for the non-triviality of this model either before squaring or after squaring is the presence of an additional (inversion) symmetry.

\subsection{Bosonic topological bulk invariants} 
\label{invariants}

We have established in Sec.\,\ref{zeromodes} the absence of SPT phases in gapped free-boson systems. However, this does not imply that these systems cannot display topologically non-trivial excitations in their spectrum. Indeed, 
Table \ref{tab:Bosonic_periodic_table} predicts that topologically non-trivial bosonic excitations may exist. Here, we show how to characterize those excitations in terms of topological bulk invariants and provide a numerically gauge-invariant way of computing them. In Sec.\,\ref{2DHofstadter} we will apply these ideas and illustrate a bulk-boundary correspondence between a non-zero value of the bulk invariant and the existence of localized SPT boundary (non-zero) modes. 

Let us first provide a simple derivation of the bosonic Berry phase for $d\geq 1$, following similar steps to those used to derive the traditional Berry phase \cite{Berry}. For an effective BdG Hamiltonian $H_\tau$ (with $H_b>0$), we have effective Schr\"{o}dinger equations (either time-dependent or time-independent) written as 
\begin{eqnarray}\label{Schrodinger}
\left\{
\begin{aligned}
    &H_\tau\big(\mathbf{k}(t)\big)|\Psi_n^\pm(t)\rangle=i\hbar|\dot{\Psi}_n^\pm(t)\rangle,\\
    &H_\tau(\mathbf{k})|\psi_n^\pm(\mathbf{k})\rangle=\pm E_n(\mathbf{k})|\psi_n^\pm(\mathbf{k})\rangle,
\end{aligned}
\right.
\end{eqnarray}
where $E_n(\mathbf{k})>0$ and $\langle\psi_n^\pm(\mathbf{k})|\tau_3|\psi_n^\pm(\mathbf{k})\rangle=\pm 1$. 
Suppose
\begin{align}
|\Psi_n^\pm(t)\rangle=e^{\pm\frac{1}{i\hbar}{\int}_0^tdt'E_n\big(\mathbf{k}(t')\big)}
e^{i\gamma_n^\pm(t)}|\psi_n^\pm\big(\mathbf{k}(t)\big)\rangle,
\end{align}
after plugging into Eq.\,(\ref{Schrodinger}), one obtains
\begin{eqnarray}
\dot{\gamma}_n^\pm(t)=\pm i\langle\psi_n^\pm\big(\mathbf{k}(t)\big)|\tau_3
\nabla_\mathbf{k}|\psi_n^\pm\big(\mathbf{k}(t)\big)\rangle\cdot\dot{\mathbf{k}}(t).
\end{eqnarray}
Defining $\gamma_n^\pm\equiv \gamma_n^\pm(T)-\gamma_n^\pm(0)$ with 
$\mathbf{k}(T)=\mathbf{k}(0)$, we arrive at the bosonic Berry phase
\begin{eqnarray}
\gamma_n^\pm=\oint d\mathbf{k}\cdot\mathcal{A}_n^\pm(\mathbf{k}),
\end{eqnarray}
where $\mathcal{A}_n^\pm(\mathbf{k})\equiv \pm i\langle\psi_n^\pm(\mathbf{k})|\tau_3
\nabla_\mathbf{k}|\psi_n^\pm(\mathbf{k})\rangle$ is the bosonic Berry connection, which is purely real. 

As explained in Sec.\,\ref{mapping2}, while $\gamma_n^+$ is associated with bosonic states with positive energies, $\gamma_n^-$ is associated with non-bosonic states with negative energies. 
Furthermore, since $|\psi_n^-(\mathbf{k})\rangle=\tau_1\mathcal{K}|\psi_n^+(\mathbf{k})\rangle$, 
it leads to
$\gamma_n^-=-\gamma_n^+.$
Therefore, $\gamma_n^+$ and $\gamma_n^-$ are not independent.
Following a procedure similar to that in Ref.\,[\onlinecite{Deng}] Appendix\,D, for a one-dimensional 
lattice system, the bosonic Berry phase $\gamma_n^+$ can be evaluated in a numerically gauge-invariant way as
\begin{eqnarray}
\gamma_n^+= 
\lim_{N\rightarrow \infty}
\text{Im} \ln\prod_{j=1}^N \langle\psi_n^+(k_{j+1})|\tau_3|\psi_n^+(k_{j})\rangle,
\end{eqnarray}
where $|\psi_n^+(k_{N+1})\rangle\equiv|\psi_n^+(k_{1})\rangle$ and $k_{j+1}\equiv2\pi j/N-\pi$.

Another topological invariant associated specifically to free bosons (when $d=2$) is the bosonic Chern number\cite{Shindou}
\begin{align}\label{bosonicChern}
\mathcal{C}_n^\pm=\frac{1}{2\pi} \oiint d\mathbf{k}\cdot\mathbf{\Omega}_n^\pm(\mathbf{k}),
\end{align}
with $\mathbf{\Omega}_n^\pm(\mathbf{k})\equiv \nabla_\mathbf{k}\times\mathcal{A}_n^\pm(\mathbf{k})$ 
the bosonic Berry curvature. Using an argument similar 
to the one used above for relating $\gamma_n^+, \gamma_n^-$,
one finds 
$\mathcal{C}_n^-=-\mathcal{C}_n^+.$
So, $\mathcal{C}_n^+$ and $\mathcal{C}_n^-$ are also not independent. 
The bosonic Chern number $\mathcal{C}_n^+$ can be evaluated in a numerically gauge-invariant way as
\begin{widetext}
\begin{align}
\mathcal{C}_n^+=\lim_{\substack{N_x\rightarrow\infty \\ N_y\rightarrow\infty}}
\frac{\text{Im}}{\pi}\sum_{i=1}^{N_x}\sum_{j=1}^{N_y}\ln
\resizebox{0.71\hsize}{!}{$\langle\psi_n^+(k_i^x, k_j^y)|\tau_3|\psi_n^+(k_i^x, k_{j+1}^y)\rangle 
\langle\psi_n^+(k_i^x, k_{j+1}^y)|\tau_3|\psi_n^+(k_{i+1}^x, k_j^y)\rangle 
\langle\psi_n^+(k_{i+1}^x, k_j^y)|\tau_3|\psi_n^+(k_i^x, k_j^y)\rangle$},
\label{ChernB}
\end{align}
\end{widetext}
where $k_{i+1}^x\equiv 2\pi i/N_x-\pi$ and $k_{j+1}^y\equiv 2\pi j/N_y-\pi$.

\section{The fate of Zero Modes}
\label{examples}

We are finally in a position to investigate localized ZMs of stable, gapped free-boson systems. As we already noted, this does not contradict our no-go Theorem \ref{nogothm} and its Corollary, since these conclusions pertain to free-boson systems subject to open BCs. We will illustrate the main ideas with several examples of our (kernel-preserving) squaring map and address many of the questions we posed, together with consequences and ramifications of our work up to now. Specifically, the following three models will be used for investigation: 

(i) The squared Kitaev chain. This example exemplifies how \emph{bosonic Majorana ZMs do exist}, by virtue of the fact that the squaring map introduces special BCs. These ZMs mimic fermionic Majorana ZMs in localization and Hermiticity but, naturally, the canonical anticommutation relation is replaced by the Heisenberg commutation relation. Since free-boson systems do not host SPT phases (no-go Theorem 2), one does not expect them being protected. We provide numerical results to quantify this expectation and frame our results within the Krein stability theory 
summarized in Sec.\,\ref{statistic}. 

(ii) The square of the Jackiw-Rebbi model, adapted for charge-neutral 
fermions (essentially, the field-theory version of the Kitaev chain). We find  
that half of a bosonic degree of freedom (one ``quadrature"), is trapped at the 
location of the soliton, while the conjugate quadrature is pushed to infinity
and out of the physical spectrum. 

(iii) The square of the Harper-Hofstadter model with flux \(\phi=1/3\) 
per plaquette and nearest-neighbor non-chiral pairing. Here, we see our no-go 
Theorem 3 (no surface bands around zero energy for open BCs), at work in full force.
In addition, we calculate and compare the Chern numbers of the fermionic model and the 
bosonic Chern numbers of its square.

\subsection{The squared Kitaev chain}
\label{1DKitaev}


We consider the following dimensionless Kitaev Hamiltonian\cite{Kitaev2001} at zero chemical potential and with an odd number of lattice sites:
\begin{eqnarray}
\widehat{H}_f = -\sum_{j=1}^{2(N-1)} \big(c_j^\dag c^{\;}_{j+1} + \frac{\Delta}{t} c_j^\dag c_{j+1}^\dag + \text{h.c.}\big), 
\end{eqnarray}
where $c_j^\dagger$ ($c^{\;}_j$) is the fermionic creation (annihilation) operator, $2N-1$ is the number of lattice sites, $t\in \mathbb{R}$ is the hopping amplitude, and $\Delta\in \mathbb{R}$ is the pairing potential. 
This system is known to host two exact ZMs exponentially localized at the two ends\cite{Cobanera}, namely, two unpaired Majorana fermions. Thus, after squaring the fermionic model subject to open BCs, the bosonic model will also host two exact ZMs. However, unexpectedly, besides these two exact ZMs, we find also two asymptotic ZMs, which we can derive analytically using the method for diagonalizing corner-modified banded block-Toeplitz operators developed in Refs.\,[\onlinecite{PRB1}]-[\onlinecite{JPA}].

Let us first square $H_f$ under periodic BCs. As mentioned in Sec.\,\ref{symmetryclasses}, for a spinless $2\times 2$ BdG Hamiltonian $H_f(\mathbf{k})$ of class BDI, the pairing potential vanishes in $\tau_3 H_f^2(\mathbf{k})$. Therefore, self-adjoint ZMs of $\tau_3 H_f^2$ must arise from special BCs that result from the squaring map of $H_f$ subject to open BCs. After some calculations, we find that $\tau_3 H_f^2$ is a Hamiltonian with next-nearest-neighbor hopping and, indeed, with an impurity potential at each end (by 
removing these impurities, i.e., imposing open BCs on $\tau_3 H_f^2$, one can check that ZMs no longer exist, consistently with no-go Theorem 3 in Sec.\,\ref{nogo}). Furthermore, 
due to the absence of nearest-neighbor hopping, lattice sites labeled by odd numbers decouple from those labeled by even numbers. Because we are focused on ZMs, we only need to consider the bosonic many-body Hamiltonian associated with odd lattice sites, which is a two-impurity Hamiltonian of the form $\widehat{H}_b+\widehat{W}$, with $\widehat{H}_b$ being the translation-invariant bulk and $\widehat{W}$ the boundary impurities. In units of $(\Delta^2+t^2)/t^2$, we have
\begin{eqnarray}
\label{two-impurity}
\left\{\begin{aligned}
&\widehat{H}_b=2\sum_{j=1}^N a_j^\dagger a^{\;}_j-\cos\theta\sum_{j=1}^{N-1}(a_j^\dagger a^{\;}_{j+1}+\text{h.c.}),\\
&\widehat{W}=-a_1^\dagger a^{\;}_1-a_N^\dagger a^{\;}_N-\frac{\sin\theta}{2}
(a_1^\dagger a_1^\dagger-a_N^\dagger a_N^\dagger+\text{h.c.}),
\end{aligned}
\right.
\end{eqnarray}
where 
$a_j^\dagger$ ($a^{\;}_j$) is the bosonic creation (annihilation) operator, 
$\sin\theta=2\Delta t/(\Delta^2+t^2)$ and $\cos\theta=(\Delta^2-t^2)/(\Delta^2+t^2)$. 
Without loss of generality, we take $\theta\in(0,\pi)$. In the rest of this subsection, we will work within the representation $\pi^\dagger H_b\pi$ we mentioned in Sec.\,\ref{mapping2} (recall Table \ref{tab:metrics}).

\subsubsection{Bosonic Majorana ZMs}

As mentioned, the effective BdG Hamiltonian $H_\tau=\tau_3 (H_b+W)$ has two exact ZMs. From Ref.\,[\onlinecite{Cobanera}] and assuming $\theta\neq \pi/2$, we write these two exact ZMs as quadratures, in the form  
\begin{eqnarray}
\label{quadraturesv2}
\left\{
\begin{aligned}
&\ket{\hat{x}_1}=
\frac{1}{\mathcal{N}}\sum_{j=1}^N(\sec\theta-\tan\theta)^{N+1-j}\ket{j}\otimes
\begin{bmatrix}
1\\
-1
\end{bmatrix},\\
&\ket{\hat{p}_1}=
\frac{i}{\mathcal{N}}\sum_{j=1}^N(\sec\theta-\tan\theta)^j\ket{j}\otimes
\begin{bmatrix}
1\\
1
\end{bmatrix},
\end{aligned}
\right.
\end{eqnarray}
where the normalization constant $\mathcal{N}$ is chosen so that $\braket{\hat{x}_1|\tau_3|\hat{p}_1}=i$ and, therefore, $[\hat{x}_1,\hat{p}_1]=i$, with $\hat{x}_1=\bra{\hat{x}_1}\tau_3\hat{\Phi}$ and $\hat{p}_1=\bra{\hat{p}_1}\tau_3\hat{\Phi}$ being self-adjoint (``Majorana bosons"). That is, we have 
\begin{eqnarray}
\left\{
\begin{aligned}
&\hat{x}_1=
\frac{1}{\mathcal{N}}\sum_{j=1}^N(\sec\theta-\tan\theta)^{N+1-j}(a^{\;}_j+a_j^\dagger),\\
&\hat{p}_1=
\frac{i}{\mathcal{N}}\sum_{j=1}^N(\sec\theta-\tan\theta)^j(a^{\;}_j-a_j^\dagger).
\end{aligned}
\right.
\end{eqnarray}
Obviously, $\hat{x}_1$ and $\hat{p}_1$ are exponentially localized at sites $j=N$ and $j=1$, respectively, which means that any bosonic ZM constructed by their linear combination has weights at both ends and therefore is non-local.

We next consider computation of the two asymptotic ZMs  mentioned previously. The 
strategy is to split the eigenvalue equation $H_\tau \ket{\epsilon}=\epsilon \ket{\epsilon}$ into two equations, the bulk equation and the boundary equation,
\begin{eqnarray}
\label{bulkboundaryeq}
P_B H_\tau \ket{\epsilon} = \epsilon P_B\ket{\epsilon},\quad
P_\partial H_\tau \ket{\epsilon} = \epsilon P_\partial \ket{\epsilon},
\end{eqnarray}
with $P_B\equiv \sum_{j=2}^{N-1} \ket{j}\bra{j}\otimes \mathds{1}_2$ the bulk projector and $P_\partial\equiv \mathds{1}_{2N}-P_B$ the boundary projector. After that, we need to solve the bulk equation, whose solutions will then be used to parameterize the solutions of the boundary equation (see Ref.\,[\onlinecite{PRB1}] for details). From Eq.\,(\ref{two-impurity}), we have
\begin{eqnarray}
\left\{
\begin{aligned}
    &H_b=\left[2\mathds{1}_N -\cos\theta(T+ T^\dagger)\right]\otimes \mathds{1}_2,\\
	&W=|1\rangle\langle 1|\otimes w_l + |N\rangle\langle N|\otimes w_r,
\end{aligned}
\right.
\end{eqnarray}
where $T\equiv \sum_{j=1}^{N-1} |j\rangle\langle j+1|$ is the left-shift operator acting on the lattice space, $w_l\equiv -\mathds{1}_2-\sin\theta\sigma_1$,
and $w_r\equiv -\mathds{1}_2+\sin\theta\sigma_1$, respectively.
The {\it reduced} bulk Hamiltonian $H(z)$ of $\tau_3 H_b$ and the associated polynomial $P(\epsilon,z)$ then read 
\begin{eqnarray}
\left\{
\begin{aligned}
    &H(z)=\left[2-(z+z^{-1})\cos\theta\right]
	\begin{bmatrix}
		1		&0\\
		0   	&-1
	\end{bmatrix},\\
	&P(\epsilon,z)= z^2\det\left(H(z)-\epsilon \mathds{1}_2\right).
\end{aligned}
\right.
\end{eqnarray}
For asymptotic ZMs, we have four distinct roots $z_\ell$ ($\ell=1,\ldots,4$) of the polynomial equation $P(\epsilon,z)=0$, associated with four independent solutions of the bulk equation in Eq.\,(\ref{bulkboundaryeq}). Thus, a linear combination of these four solutions can parameterize the asymptotic ZMs as
\begin{eqnarray}
\ket{\epsilon_\pm}=
\sum_{\ell=1}^2\alpha_\ell\ket{z_\ell}\otimes
\begin{bmatrix}
1\\
0
\end{bmatrix}+
\sum_{\ell=3}^4\alpha_\ell\ket{z_\ell}\otimes
\begin{bmatrix}
0\\
1
\end{bmatrix},
\end{eqnarray}
with $\ket{z_\ell}=\sum_{j=1}^N z_\ell^j \ket{j}$. This finally leads to a system of linear equations, 
$B(\epsilon)[\alpha_1,\alpha_2,\alpha_3,\alpha_4]^\text{T}=0$, with $B(\epsilon)$ being the \emph{boundary matrix} given explicitly by 
\begin{widetext}
\begin{eqnarray}
	B(\epsilon)=\begin{bmatrix}
\cos\theta-z_1\qquad\quad	        &\cos\theta-z_2\qquad\quad         &-z_3\sin\theta\qquad\quad         &-z_4\sin\theta\\
z_1\sin\theta\qquad\quad	        &z_2\sin\theta\qquad\quad          &z_3-\cos\theta\qquad\quad         &z_4-\cos\theta\\
z_1^N(z_1\cos\theta-1)\qquad\quad	&z_2^N(z_2\cos\theta-1)\qquad\quad &z_3^N\sin\theta\qquad\quad        &z_4^N\sin\theta\\
-z_1^N\sin\theta\qquad\quad	        &-z_2^N\sin\theta\qquad\quad       &z_3^N(1-z_3\cos\theta)\qquad\quad &z_4^N(1-z_4\cos\theta)
	\end{bmatrix}.
\end{eqnarray}
The eigenvalues of $H_\tau$ are precisely those that ensure $\det B(\epsilon)=0$. After some calculations, we obtain
\begin{align}
\label{detB}
\text{det}B=\frac{\sin^2(2\theta)}{8}(z_1-z_2)(z_3-z_4)\big[(z_1^N+z_2^N)(z_3^N+z_4^N)-4\big]-(z_1^N-z_2^N)(z_3^N-z_4^N)
\Big(2\sin^4\theta-\frac{1+\cos^2\theta}{2}\epsilon^2 \Big),
\end{align}
which is an even function of $\epsilon$, with the zero-th order term absent.
\end{widetext}

In the large-$N$ limit, we expand Eq.\,(\ref{detB}) to order $\mathcal{O}(\epsilon^6)$
\begin{align}
\text{det}B(\epsilon)=
4N^2\sin^2\theta\Big[\epsilon^2-\frac{(\sec\theta+\tan\theta)^{2N}}{16N^2\sin^4\theta\tan^2\theta}\epsilon^4\Big].
\end{align}
Indeed, we see two asymptotic ZMs with eigenvalues
\begin{align}
\epsilon_\pm=\pm\frac{4N\sin^2\theta\tan\theta}{(\sec\theta+\tan\theta)^N},
\end{align}
which can be plugged into $P(\epsilon,z)=0$ to solve for $z_\ell$, $\ell=1,\ldots,4$. 
In the large-$N$ limit, a nontrivial kernel vector of $B(\epsilon_\pm)$ is 
$\big[\alpha_1,\alpha_2,-\alpha_1,\alpha_2\big]^\text{T}$, which leads to 
$\ket{\epsilon_\pm}$ being written as the linear combination of $\ket{\hat{x}_1}$ and $\ket{\hat{p}_1}$ of Eq.\,(\ref{quadraturesv2}), indicating the loss of diagonalizability of $H_\tau$.

\begin{table}[t]
    \centering 
    \caption{ZMs counting for the Kitaev chain with odd number of lattice sites and the squared Kitaev chain with odd number of lattice sites. 
    The number of ZMs refers to the number of linearly 
    independent quasiparticle creation operators that commute with the many-body Hamiltonian. The number of mid-gap modes refers to the sum of the number of exact ZMs and the number of asymptotic ZMs.}
\begin{tabular}{ccccc}
		\hline		\hline
		\multirow{2}{*}{System} & \multirow{2}{*}{Size} & \multirow{2}{*}{$\quad$Diagonalizable$\quad$} & \multirow{2}{0.4in}{\,\,ZMs} & \multirow{2}{0.5in}{Mid-gap $\text{\,\,\,modes}$} \\
		&       &       &       &  \\
		\hline
		\multirow{2}{*}{$H_f$}		& \multirow{1}{*}{Finite} & \multirow{1}{*}{Yes} & \multirow{1}{*}{2} & \multirow{1}{*}{2} \\          
		& \multirow{1}{*}{Semi-infinite} & \multirow{1}{*}{Yes} & \multirow{1}{*}{1} & \multirow{1}{*}{1} \\
		&       &       &       &  \\
		\multirow{3}{*}{$\tau_3 H_f^2$} & \multirow{2}{0.35in}{Finite} & Yes, if $|\Delta/t|\neq 1$ & \multirow{2}{*}{2} & 4 \\
		&       & No, if $|\Delta/t|= 1$ &       & 2 \\
		& \multirow{1}{*}{Semi-infinite} & \multirow{1}{*}{No} & \multirow{1}{*}{1} & \multirow{1}{*}{1} \\
		\hline		\hline
	\end{tabular}
    \label{tab:counting}
\end{table} 

The effective BdG Hamiltonian $H_\tau$ also fails to be diagonalizable when $\theta=\pi/2$, the so-called ``sweet-spot" ($t=\Delta$) of the Kitaev chain, leading to
\begin{eqnarray}
    \widehat{H}_b +\widehat{W}= \sum_{j=2}^{N-1} 2a_j^\dagger a^{\;}_j + (p_1^2+x_N^2),
\end{eqnarray}
where $\hat p_1=i(a_1^\dagger-a^{\;}_1)/\sqrt{2}$ and $\hat x_N=(a_N^\dagger+a^{\;}_N)/\sqrt{2}$ are two independent self-adjoint ZMs, each associated with a Jordan block of size $2$. 
According to Theorem (ZMs, [\onlinecite{ColpaZMs}]) of Sec.\,\ref{statistic}, it is clear that no canonical bosonic ZMs can be built from $\hat p_1$ and $\hat x_N$. 

The above discussion on exact and asymptotic ZMs, together with the diagonalizability of the systems, is summarized in Table\,\ref{tab:counting}.

\subsubsection{Sensitivity of ZMs to perturbations}

In this section we investigate the sensitivity of ZMs to perturbations that preserve the stability of the free-boson system and contrast to perturbations that do not.

\paragraph{Stability-preserving perturbations.}

We first consider a boundary perturbation of the form 
\begin{align}
\label{kits}
\widehat{W}_s = (s-1) \widehat{W},\quad s\in[0,1] .
\end{align}
By adding \(\widehat{W}_s\) to Eq.\,\eqref{two-impurity}, since $H_b+sW \geq 0$ for all $s$, one obtains a stable family of free-boson systems that interpolates between open BCs ($s=0$), which forbid ZMs, and impurity BCs ($s=1$), which elicit exact bosonic Majorana ZMs. It is instructive to investigate how bosonic Majorana ZMs split as a function of the system size $N$ for different $s$'s. Analytically, at the sweet spot $\Delta/t=1$, ZMs split into $\pm 2\sqrt{1-s}$, 
with no $N$-dependence because of the decoupling between the boundary and the bulk. We plot the numerically determined minimal-modulus eigenvalue of $\tau_3 (H_b+sW)$ for $\Delta/t=0.5$ 
in Fig.\,\ref{bdrysplit}(a) and find opposite behaviors for small and large $s$'s. 
For small $s$'s, we see that the splitting of bosonic Majorana ZMs away from zero energy anomalously increases (rather than decreasing) as $N$ grows, which is not the case for protected fermionic Majorana ZMs.

\begin{figure}[t]
\includegraphics[width=\columnwidth]{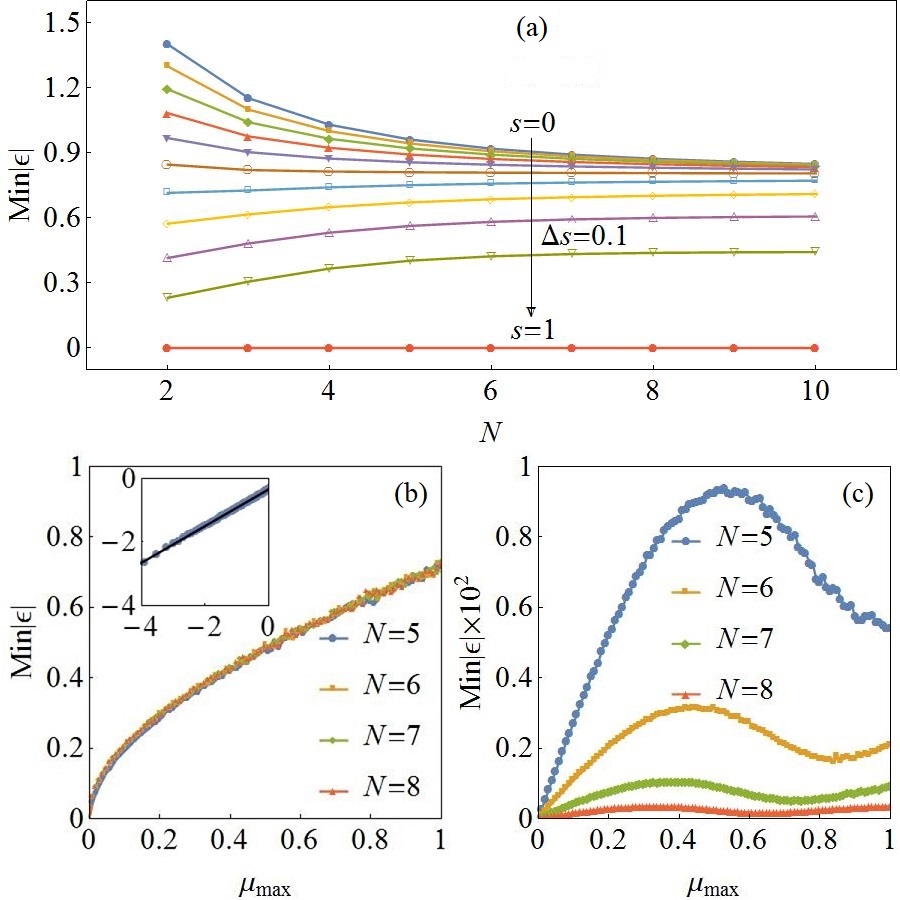}
\caption{The minimal-modulus eigenvalue of (a) $\tau_3(H_b+sW)$ as a function of the system size $N$ for various impurity strength $s$'s; (b) $\tau_3 (H_b+W+D_b)$ as a function of $\mu_\text{max}$ for various $N$'s (the inset is a log-log plot for $N=5$, with a linear fitting $\text{ln}(\text{Min}|\epsilon|)=0.58\, \text{ln}\,\mu_\text{max}-0.34$); (c) $H_f+D_f$ as a function of $\mu_\text{max}$ for various $N$'s. Both (b) and (c) employ randomly distributed disorder $\mu_j\in[0,\mu_\text{max}]$, averaged over 1,000 samples. $\Delta/t=0.5$ for (a), (b) and (c).}
\label{bdrysplit}
\end{figure}

As a second example, we keep the 
impurity BCs intact and perturb instead the bulk with the on-site disorder 
\begin{align}
\label{bosedis}
\widehat{D}_b = \sum_{j=1}^N \mu_j a_j^\dag a^{\;}_j,
\end{align}
where $\mu_j$'s are uniformly sampled in the interval $[0,\mu_\text{max}]$. We 
plot the numerically obtained minimal-modulus eigenvalue of $\tau_3 (H_b+W+D_b)$ as a function of $\mu_\text{max}$ and $N$ in Fig.\,\ref{bdrysplit}(b). 
We see that the splitting of ZMs increases monotonically as a function of $\mu_\text{max}$ and is independent of the system size. 
For comparison, we also plot in Fig.\,\ref{bdrysplit}(c) 
the minimal-modulus eigenvalue of $H_f+D_f$, i.e., the fermionic Kitaev chain subject to the on-site disorder 
\begin{align}
\label{fermdis}
\widehat{D}_f = \sum_{j=1}^{2N-1}\mu_j c_j^\dag c^{\;}_j, 
\end{align}
with $\mu_j$'s also uniformly sampled in the interval $[0,\mu_\text{max}]$. Unlike the bosonic case, the splitting of the Majorana ZMs is not monotonic as a function of $\mu_\text{max}$ and is sensitive to the system size. Moreover, the splitting is a couple of orders of magnitude smaller than the bosonic one. 

\begin{figure}[t]
\includegraphics[width=\columnwidth]{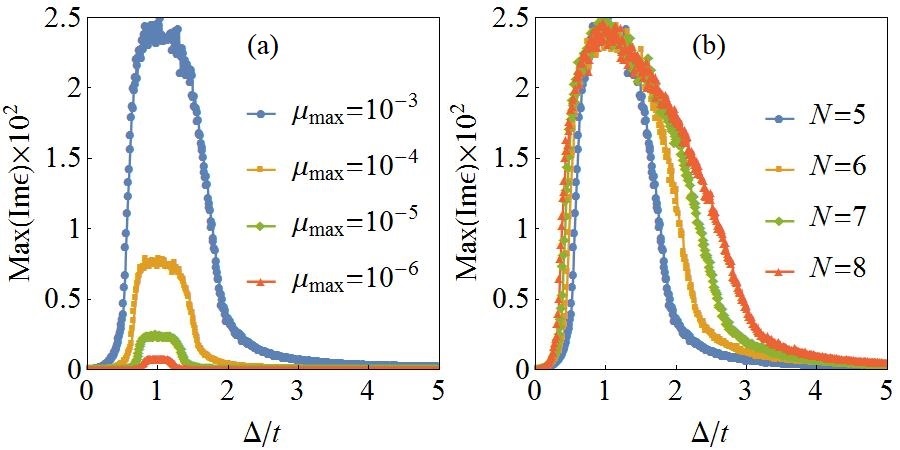}
\caption{(a) The maximal imaginary part of eigenvalues of $\tau_3 (H_b+W+D_b)$ as a function of $\Delta/t$ for various disorder strength $\mu_\text{max}$'s but with fixed system size $N=5$. (b) The maximal imaginary part of eigenvalues of $\tau_3 (H_b+W+D_b)$ as a function of $\Delta/t$ for various $N$'s, but with fixed $\mu_\text{max}=10^{-3}$. Both (a) and (b) are with randomly distributed disorder $\mu_j\in[-\mu_\text{max},\mu_\text{max}]$ averaged over 1,000 samples. When $\Delta/t=1$, we find analytically that Max(Im$\,\epsilon$)=${(8\sqrt{2}}/{15})\,\mu_{\text{max}}^{1/2}+\mathcal{O}(\mu_{\text{max}}^{3/2})$.}
\label{DisKitsqFig1ab}
\end{figure}

\paragraph{Stability-non-preserving perturbations.}
We again consider the bulk disorder in Eq.\,\eqref{bosedis} with the ideal impurity BCs intact, but now with $\mu_j$'s uniformly sampled in the interval $[-\mu_\text{max},\mu_\text{max}]$. As a consequence, the disorder may render the system unstable because of the violation of positive semi-definiteness of $H_b+W+D_b$. Furthermore, because ZMs of the unperturbed effective BdG Hamiltonian have different Krein signatures, as explained in Sec.\,\ref{statistic} there exists arbitrarily small perturbations that split ZMs into the complex plane. This is numerically confirmed in Fig.\,\ref{DisKitsqFig1ab}(a). 
We further plot the maximal imaginary part of eigenvalues of 
the perturbed effective BdG Hamiltonian as a function of \(\Delta/t\) for various system sizes $N$'s in Fig.\,\ref{DisKitsqFig1ab}(b). 
Both Fig.\,\ref{DisKitsqFig1ab}(a) and (b) suggest that ZMs are especially fragile around the sweet spot $\Delta/t=1$.

\subsection{Localized ZM in a bosonic field theory}

We consider next a field-theoretic example of the squaring procedure closely
related to the Kitaev chain: we square the celebrated Jackiw-Rebbi model of 
charge fractionalization\cite{JR} on the infinite real line. The model can be succinctly described in terms of the Dirac equation in one spatial dimension $x$,
\begin{align}
\label{jrham}
i \gamma^\nu \partial_\nu \psi - g V(\phi_c) \psi = 0 ,
\end{align}
where $V(\phi_c) = \phi_c$, $\nu=0,1$, $\partial_0\equiv \partial_t$, $g,\lambda>0$, and 
\begin{eqnarray}
\label{soliton}
\phi_c(x)=\text{tanh}(\lambda x).
\end{eqnarray} 
The stationary 
solutions are of the form \(\psi_\epsilon(x,t) = e^{-i \epsilon t} \psi_\epsilon (x)\),
with   
$-i \gamma^0 \gamma^1 \partial_x \psi_\epsilon(x) + \gamma^0 g V(\phi_c) \psi_\epsilon(x)=\epsilon \psi_\epsilon(x).$
Hence, the Dirac Hamiltonian is 
\begin{align}
\label{HD}
H_D \equiv \gamma^0 \gamma^1 p + \gamma^0 g V(\phi_c) 
\end{align}
with $p = -i\partial_x$. The choice of gamma matrices 
$\gamma^0 = \sigma_3$ and $\gamma^1= i\sigma_2$ puts this Dirac Hamiltonian
in Nambu form, 
\begin{align}
\label{HDsc}
H_D = \begin{bmatrix}
g V(\phi_c) & p \\ p & -gV(\phi_c)
\end{bmatrix},
\end{align}
satisfying $\mathcal{C} H_D \mathcal{C} = -H_D$ in terms of $\mathcal{C}\equiv \sigma_1\mathcal{K}$.

There are two ZMs, $\psi_0^+(x)$ and $\psi_0^-(x)$. They are related by charge conjugation, $\psi_0^-(x)= \mathcal{C} \psi_0^+(x)$, with 
\begin{align}\label{FZMs}
\psi_0^+(x) &= \cosh\left[\frac{g}{\lambda}\ln\big(\cosh(\lambda x)\big)\right]\begin{bmatrix}
1 \\ 0
\end{bmatrix}\nonumber\\ 
&- i\sinh\left[\frac{g}{\lambda}\ln\big(\cosh(\lambda x)\big)\right]\begin{bmatrix}
0 \\ 1
\end{bmatrix}.
\end{align}
They can be combined into a spatially-localized ZM 
\begin{align}
\psi_0^+(x) + i \psi_0^-(x) = \big[\cosh(\lambda x)\big]^{-g/\lambda}\begin{bmatrix}
1 \\ i
\end{bmatrix},
\end{align}
whereas the linearly independent combination $\psi_0^+ - i \psi_0^-$ diverges as $x\to\pm\infty$.

At this point our work separates from Ref.\,[\onlinecite{JR}]. Let us 
second-quantize the model as 
\begin{align}
\widehat{H}_f \equiv \frac{1}{2}\int \hat{\Psi}(x)^\dagger H_D(x) \hat{\Psi}(x),
\end{align}
in terms of the Nambu array $\hat{\Psi}(x)=[c(x)\,\,c^\dag(x)]^\text{T}$, with 
$\{c^{\;}(x),c^\dag(y)\}=\delta(x-y)$ and \(\{c(x),c(y)\}=0\). This Hamiltonian
describes spinless, electrically neutral fermions in a static background potential, 
the soliton of Eq.\,\eqref{soliton}. The charge fractionalization of the original model \cite{JR} is replaced by fermion-number fractionalization in our model. One can think of this model as a field-theory version of the Kitaev chain. 

A normal mode of \(\widehat{H}_f\) is an operator of the form  
\begin{align}
\hat{\psi_\epsilon}(t) \equiv \int \big[u^*(x,t) c(x) + v^*(x,t) c^\dag(x)\big]\,dx,
\end{align} 
with  $\psi_\epsilon(x,t)\equiv[u(x,t)\,\,v(x,t)]^\text{T}$ a 
stationary solution of Eq.\,\eqref{jrham}. It is immediate to check that 
\begin{align}
\{\hat{\psi^{\;}_\epsilon}(t),\hat{\psi}_\epsilon^\dag(t)\} = \int \big[|u(x,t)|^2 +|v(x,t)|^2\big]\,dx. 
\end{align}
Therefore, if one takes 
\begin{align}
\psi_0(x,t) 
& = 
e^{-i \pi/4} \,
\frac{\psi_0^+(x)+i \psi_0^-(x)}{2 {\cal N}^{1/2} } \nonumber\\
& =  \frac{[\cosh(\lambda x)]^{-g/\lambda}}{2{\cal N}^{1/2}}
e^{-i \pi/4}
\begin{bmatrix}
1 \\ i
\end{bmatrix}, 
\end{align}
with \({\cal N}^{-1}= \frac{2 \pi^{1/2}\Gamma[g/2\lambda]}{\lambda\Gamma[(g+\lambda)/2\lambda]}\) and $\Gamma$ being the Euler's Gamma function, one finds that $\hat{\psi}_0(t)=\hat{\psi}_0$ is independent
of time because \(\epsilon=0\) and  
$\{\hat{\psi}^{\;}_0,\hat{\psi}_0^\dagger\}=1,$ 
$\hat{\psi_0}^\dagger=\hat{\psi}^{\;}_0 .$
Since this is the only ZM and \(\hat{\psi}^{\;}_0\) and \(\hat{\psi}_0^\dagger\) are
linearly dependent, there is \emph{only half of a fermionic degree of freedom trapped at $x=0$}: a single Majorana fermion in the sense of Kitaev\cite{Kitaev2001}. The other
half is carried by the unnormalizable ZM $\psi_0^+ - i \psi_0^-$ and
so it has been pushed to infinity. This point of view is nicely
bolstered by solving the model subject to open 
BCs\cite{Roy84}.

We now proceed to investigate the associated free-boson theory. The square of $H_D$ is 
\begin{align}
 H_D^2 = \mathds{1}_2\big[ p^2 + \big(g \tanh(\lambda x)\big)^2\big] -\sigma_2 g\lambda \, \text{sech}^2(\lambda x),
\end{align}
with $H_D^2$ satisfying $\mathcal{C} H_D^2 \mathcal{C}=H_D^2$. 
Unlike the Dirac Hamiltonian of Sec.\,\ref{diquare}, \(\widehat{H}_b\) describes an explicitly particle non-conserving free-boson system. The second-quantized Hamiltonian is \(\widehat{H}_b=\frac{1}{2}\int \hat{\Phi}^\dag (x)H_D^2\hat{\Phi}(x)\,dx\), in terms of the bosonic Nambu array $\hat{\Phi}(x) = [a(x)\,\,a^\dag(x)]^\text{T}$, with $[a(x),a^\dag(y)]=\delta(x-y)$ and \([a(x),a(y)]=0\). 
The Hamiltonian density \( \frac{1}{2}\hat{\Phi}^\dag (x)H_D^2\hat{\Phi}(x)\)
is the sum of three contributions, namely, 
\begin{align*}
\widehat{T}(x) &= \frac{1}{2}\big[a^\dag(x) p^2 a(x) + a(x)p^2 a^\dag(x)\big] ,
\\
\widehat{U}(x) &= \frac{1}{2}\big[g\, \text{tanh}(\lambda x)\big]^2 \big[a^\dag(x)a(x) + a(x)a^\dag(x)\big],
\\
\widehat{\Delta}(x) &= \frac{ig\lambda}{2}\text{sech}^2(\lambda x)\big[ a^\dag(x)a^\dag(x)-a(x)a(x)\big].
\end{align*}
Notice that both the potential energy and the pairing potential are exponentially localized around \(x=0\). 

A normal mode of \(\widehat{H}_b\) is an operator 
\begin{align*}
\hat{\phi_\epsilon}(t) \equiv \int \big[u^*(x,t) a(x) - v^*(x,t) a^\dag (x)\big]\,dx ,
\end{align*}
that satisfies certain conditions. In terms of $\phi(x,t) = [u(x,t)\,\,v(x,t)]^\text{T}$, one can then check that
\begin{eqnarray}
\label{bocomm}
[\hat{\phi}^{\;}_\epsilon(t),\hat{\phi}_\epsilon^\dag(t)]=\int \phi^\dag(x,t)\sigma_3\phi(x,t)\,dx.
\end{eqnarray}
Moreover, the following properties hold: 
\begin{itemize}
\item[(i)] If $\partial_t \phi(x,t)  = i \sigma_3 H_D^2(x) \phi(x,t)$,  
then $\hat{\phi}_\epsilon(t)$ satisfies Heisenberg's equation of motion. 
\item[(ii)] If $\phi(x)$ is an eigenfunction of $\sigma_3 H_D^2(x)$ with 
eigenvalue $\epsilon$, then $\hat{\phi}_\epsilon(t) = e^{-i\epsilon t}\hat{\phi}_\epsilon(0)$. 
\end{itemize}
By construction, both $\psi_0^\pm(x)$ are formally eigenfunctions of $\sigma_3 H_D^2(x)$
with eigenvalue \(\epsilon=0\). However, the combination $\psi_0^+ - i \psi_0^-$ that diverges as $x\to\pm\infty$ is badly behaved to be considered an eigenvector even in a generalized sense. Hence, one can take the view that there is only one 
eigenvector with zero eigenvalue, the localized one. The associated self-adjoint bosonic ZM is 
\begin{align}
\hat{\phi}_0=
\frac{e^{-i\pi/4} }{2 {\cal N}^{1/2} }
\int \big[\cosh(\lambda x)\big]^{-g/\lambda}
\big[a(x)+ia^\dagger(x)\big]\,dx.
\end{align}
Since \(\hat{\phi}_0^\dagger=\hat{\phi}^{\;}_0\), it follows that \([\hat{\phi}^{\;}_0,\hat{\phi}_0^\dagger]=0\)
and one can again take the view that the bosonic theory traps half of a bosonic
degree of freedom (a ``quadrature'') at the origin, while the other quadrature is pushed to infinity. Our analysis of the squared Kitaev chain suggests that a soliton-antisoliton pair will split a single bosonic degree of freedom into two quadratures, one localized at the center of the soliton and the other localized at the center of the antisoliton. We are not aware of a previous description in the literature of this phenomenon for bosons. 

\subsection{The squared Harper-Hofstadter-pairing model}
\label{2DHofstadter}

As a final example, we now apply our squaring map to the spinless $d=2$ Harper-Hofstadter Hamiltonian with additional pairing terms. Time-reversal symmetry is broken in the Harper-Hofstadter model\cite{Harper, Hofstadter}. After introducing pairing terms, the model belongs to class D of the Hermitian classification. Upon squaring, the effective BdG Hamiltonian belongs to class D of pseudo-Hermitian symmetry classes (see Table\,\ref{tab:Particle_non-conserving_squared_irreducible}) 
and, according to Table\,\ref{tab:Bosonic_periodic_table}, is classified by an integer $\mathbb{Z}$, the bosonic Chern number [see Eq.\,(\ref{bosonicChern}) in Sec.\,\ref{invariants}].

We start with the following fermionic tight-binding Hamiltonian subject to open BCs
\begin{widetext}
\begin{align}
\widehat{H}_f=
-\sum_{m=1}^{L_x}\sum_{n=1}^{L_y}\mu c_{m,n}^\dagger c^{\;}_{m,n} +
\sum_{m=1}^{L_x-1}\sum_{n=1}^{L_y}&\big(t_x c_{m+1,n}^\dagger c^{\;}_{m,n} + 
\Delta_x c_{m+1,n}^\dagger c_{m,n}^\dagger + \text{h.c.}\big)\nonumber\\ +
\sum_{m=1}^{L_x}\sum_{n=1}^{L_y-1}\big(t_ye^{i2\pi\phi m} c_{m,n+1}^\dagger &c^{\;}_{m,n} + 
\Delta_y c_{m,n+1}^\dagger c_{m,n}^\dagger + \text{h.c.}\big),
\end{align}
where $L_x\,(L_y)$ is the number of lattice sites along the $x\,(y)$ direction, $\mu$ is the on-site energy, $c_{m,n}^\dagger\,(c^{\;}_{m,n})$ is the fermionic creation (annihilation) operator at site $(m,n)$, $t_x\,(t_y)$ is the nearest-neighbor hopping amplitude along the $x\,(y)$ direction, $\Delta_x\,(\Delta_y)$ is the pairing 
potential along the $x\,(y)$ direction, and $\phi$ is the magnetic flux per plaquette in units of flux quanta $h/e$. Imposing periodic BCs in the $y$ direction and assuming $t_x=t_y=-t$, $\Delta_x=\Delta_y=\Delta$, $\phi=1/3$, 
we have $\widehat{H}_f=\sum_{k_y}\widehat{H}_f(k_y)$, with $\widehat{H}_f(k_y)$ in units of $t$ as
\begin{align}
\widehat{H}_f(k_y)=- \sum_{m=1}^{L_x}\bigg[\Big[2\cos{\Big(k_y-\frac{2\pi m}{3}\Big)}+\frac{\mu}{t}\Big]
c_{m,k_y}^\dagger& c^{\,}_{m,k_y} +\Big[i\frac{\Delta}{t}\sin{k_y} c_{m,k_y}^\dagger c_{m,-k_y}^\dagger 
+ \text{h.c.}\Big]\bigg]\nonumber\\ - \sum_{m=1}^{L_x-1}\bigg[c_{m+1,k_y}^\dagger c^{\;}_{m,k_y} -&\frac{\Delta}{t} 
c_{m+1,k_y}^\dagger c_{m,-k_y}^\dagger +\text{h.c.}\bigg].
\end{align}
\end{widetext}
Thus, we have an effective one-dimensional fermionic many-body Hamiltonian $\widehat{H}_f(k_y)$. By numerical diagonalization of the BdG Hamiltonian $H_f(k_y)$, we obtain the single-particle energy spectra plotted in Fig.\,\ref{fig:p_q=1_3_E(ky)}(a). Apart from mid-gap edge states at finite energies, there are localized Majorana ZMs as well. Specifically, there are two chiral propagating ZMs at each edge of the cylinder, of opposite directions. This is consistent with numerical evaluations of the fermionic Chern numbers of the lowest three negative energy bands under periodic BCs: that is, we find  $(\mathcal{C}_3^-,\,\mathcal{C}_2^-,\,\mathcal{C}_1^-)=(-1,\,2,\,1)$, where $\mathcal{C}_3^-$ is the Chern number of the lowest negative energy band. Due to the bulk-boundary correspondence, a non-vanishing sum of these three Chern numbers 
$\mathcal{C}_3^- + \mathcal{C}_2^- + \mathcal{C}_1^-=2$ corresponds to topologically nontrivial ZMs, namely, topologically protected Majorana fermions localized at the boundaries. However, for free-boson systems, the sum of the bosonic Chern numbers of negative energy bands must vanish\cite{Shindou} and ZMs cannot exist subject to open BCs according to no-go Theorem\,\ref{nogothm} in Sec.\,\ref{nogo}. 

Let us square the BdG Hamiltonian $H_f(k_y)$, which leads to a quadratic bosonic Hamiltonian of the form $\widehat{H}_b(k_y)+\widehat{W}$, with
\begin{widetext}
\begin{align}
\widehat{H}_b(k_y)= 
&\sum_{m=1}^{L_x}\mu(k_y)a_{m,k_y}^\dagger a^{\,}_{m,k_y}
+ \sum_{m=1}^{L_x-1}\Big[t_1(k_y)a_{m+1,k_y}^\dagger a^{\,}_{m,k_y} +\text{h.c.}\Big]+
\sum_{m=1}^{L_x-2}\Big[t_2 a_{m+2,k_y}^\dagger a^{\,}_{m,k_y}+\text{h.c.}\Big]\nonumber\\
&+\sum_{m=1}^{L_x}\Big[\Delta_0(k_y)a_{m,k_y}^\dagger a_{m,-k_y}^\dagger + \text{h.c.}\Big]
+ \sum_{m=1}^{L_x-1}\Big[\Delta_1(k_y)a_{m+1,k_y}^\dagger a_{m,-k_y}^\dagger +\text{h.c.}\Big],\\
\text{where}\quad&\left\{
\begin{aligned}
\mu(k_y)&\equiv\Big[2\cos{\Big(k_y-\frac{2\pi m}{3}\Big)}+\frac{\mu}{t}\Big]^2+4\frac{\Delta^2}{t^2}\sin^2{k_y}+2(1+\frac{\Delta^2}{t^2}),\nonumber\\
t_1(k_y)&\equiv -2\Big[\cos{\Big(k_y-\frac{2\pi(m+2)}{3}\Big)}-2i\frac{\Delta^2}{t^2}\sin{k_y}-\frac{\mu}{t}\Big],\quad\quad t_2\equiv 1-\frac{\Delta^2}{t^2},\nonumber\\
\Delta_0(k_y)&\equiv 4i\frac{\Delta}{t}\sin{\frac{2\pi m}{3}}\sin^2{k_y},\quad\quad \Delta_1(k_y)\equiv4\frac{\Delta}{t}\sin{\frac{2\pi(m+2)}{3}}\sin{\Big(k_y-\frac{\pi}{3}\Big)},\nonumber
\end{aligned}
\right.
\end{align}
\begin{align}
\widehat{W}= -\Big(1+\frac{\Delta^2}{t^2}\Big)\Big[a_{1,k_y}^\dagger a^{\,}_{1,k_y}+
a_{L_x,k_y}^\dagger a^{\,}_{L_x,k_y}\Big] -
\bigg[\frac{\Delta}{t}\Big[a_{1,k_y}^\dagger a_{1,-k_y}^\dagger-a_{L_x,k_y}^\dagger a_{L_x,-k_y}^\dagger\Big] + \text{h.c.}\bigg].
\end{align}
\end{widetext}
Then, the effective BdG Hamiltonian is $H_\tau(k_y)=\tau_3 (H_b(k_y)+W)$, with the single-particle spectra around zero energy plotted in Fig.\,\ref{fig:p_q=1_3_E(ky)}(b). Unlike in Fig.\,\ref{fig:p_q=1_3_E(ky)}(a) for fermions, the bosonic positive bands are disconnected from the negative bands. This feature is a peculiarity of this model. For example, if one consider the free-boson model associated to the square of the chiral \(p+ip\) superconductor\cite{Read00}, one would find bosonic surface bands that cross zero energy for appropriate BCs. What these two free-boson models have in common is that the total Chern number of the negative bands vanishes, as must be the case in general for the bosonic Chern number\cite{Shindou}. In Table \ref{tab:p_q=1_3_chern_number} we compare the Chern numbers of the fermionic and its associated bosonic models. For the bosonic descendant of the Harper-Hofstadter-pairing model, the bosonic Chern numbers of the lowest three negative bands are $(\mathcal{C}_3^-,\,\mathcal{C}_2^-,\,\mathcal{C}_1^-)=(-1,\,2,\,-1)$, with the original fermionic Chern number $\mathcal{C}_1^-=1$ changed to $\mathcal{C}_1^-=-1$ for the bosonic model. This dramatically alters the topological properties of ZMs as the sum of the three bosonic Chern numbers vanishes, meaning that ZMs in the band gap are \emph{not} topologically mandated. 

\begin{figure*}[tb]
    \centering
    \includegraphics[width=0.8\textwidth]{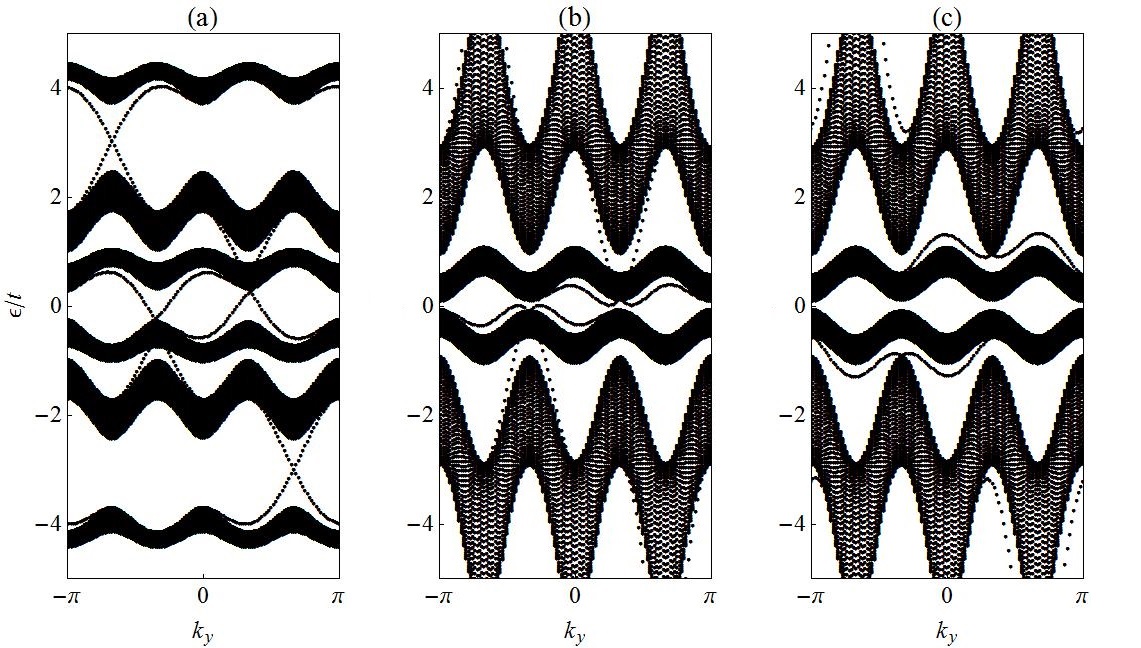}
    \vspace*{-2mm}
    \caption{Single-particle energy spectra of (a) the fermionic Harper-Hofstadter-pairing model subject to open BCs along the $x$ direction; (b) the squared Harper-Hofstadter-pairing model subject to impurity BCs along the $x$ direction; (c) the squared Harper-Hofstadter-pairing model subject to open BCs along the $x$ direction. The parameters are $\mu/t=-1.7$, $\Delta/t=-0.1$, $\phi=1/3$ and $L_x=L_y=120$. Note that for visual clarity only spectra around 
    zero energy are plotted in figures (b) and (c).}
    \label{fig:p_q=1_3_E(ky)}
\end{figure*}

\begin{table}[t]
    \centering 
    \caption{Chern numbers of the Harper-Hofstadter-pairing model $H_f(\mathbf{k})$ and the squared free-boson model $\tau_3H_f^2(\mathbf{k})$. The bosonic Chern number is 
    numerically evaluated based on Eq. \eqref{ChernB}. 
    The parameters are $\mu/t=-1.7$, $\Delta/t=-0.1$, $\phi=1/3$.}
    \smallskip
    \begin{tabular}{x{22mm} x{11mm} x{11mm} x{11mm} x{22mm}}
    \hline\hline \\ [-2ex]
             \      & $\mathcal{C}_3^-$ & $\mathcal{C}_2^-$ & $\mathcal{C}_1^-$ & $\mathcal{C}_3^- + \mathcal{C}_2^- + \mathcal{C}_1^-$ 
    \tabularnewline [1ex] \hline \\[-1ex]
    \(H_f(\mathbf{k})\)         & -1      & 2       & 1       & 2
    \tabularnewline [1ex] 
    \(\tau_3H_f^2(\mathbf{k})\) &   -1    & 2       & -1      & 0
    \tabularnewline [1ex] \hline\hline
    \end{tabular}
    \label{tab:p_q=1_3_chern_number}
\end{table}

The reason why ZMs and edge states around the zero energy appear in our model is due to special BCs caused by the squaring procedure, i.e., the boundary impurity term $\widehat{W}$. To see that this is the case, we remove $\widehat{W}$ and diagonalize $H_\tau(k_y)=\tau_3 H_b(k_y)$ subject to open BCs. The single-particle energy spectra around zero energy is plotted in Fig.\,\ref{fig:p_q=1_3_E(ky)}(c): this shows no ZMs, while surface bands in other gaps survive. Thus, topologically non-trivial stable free-boson systems are systems with topologically non-trivial single-particle states at finite energies but topologically trivial many-body ground states, in agreement with the general no-go theorems we discussed in Sec.\,\ref{zeromodes}.

\section{Conclusions and outlook}
\label{summary}

Many-boson systems display amazing coherent behavior and correlations, such as Bose-Einstein condensation and fragmentation. At the mean-field level description, however, topology seems to provide a limited-scope principle for the organization of low-energy, stable bosonic matter. In this paper, following steps analogous to those in the fermionic tenfold way, with identical classifying symmetry constraints together with the pseudo-Hermiticity, we have presented a topological classification of stable free-boson systems by means of a kernel-preserving squaring map, leading to an elegant threefold way. This topological classification bears great resemblance with standard Dyson symmetry classes because, as shown, out of three classifying symmetries only time-reversal symmetry is of fundamental importance for canonical bosons. 

Moreover, we proved three no-go theorems applicable to arbitrary stable 
gapped free-boson systems, even those that are not derived from our squaring map. Our first theorem establishes the even parity of bosonic ground states, with the immediate consequence of lack of parity switches. Consistently, our second theorem dictates the absence of non-trivial SPT phases of stable free-boson systems. 
By the Gell-Mann-Low theorem and our no-go Theorem 2, one can conclude that there also exist no non-trivial SPT phases of weakly-interacting bosons. Therefore, SPT phases of bosons can only be strongly-correlated phases of matter, beyond the reach of perturbative approaches. Our third theorem puts the last nail in the coffin, by asserting that not only bosonic ZMs, but also midgap states around zero energy are forbidden when the system is subject to open BCs. These results can be traced back to a condition that bosonic Hamiltonians need to satisfy in order to be stable (positive semi-definiteness). There is no counterpart of this condition for fermions.

In spite of these no-go results, we utilized our squaring-the-fermion map for generating a wealth of examples of localized bosonic ZMs and surface bands in the zero gap coexisting with a fully-gapped bulk. The key is to notice that the square of a fermionic BdG Hamiltonian satisfies the particle-hole constraint associated to bosonic effective BdG Hamiltonians. The localized ZMs and surface bands obtained by this method share, with due allowance for the change in statistics, every exotic property of their fermionic counterpart with two exceptions: the resulting BCs \emph{cannot be open}, and, consistently with our no-go results, there generically is \emph{no protection} 
mechanism at play. We investigated this last point numerically in considerable detail for the squared Kitaev chain. Besides two exact bosonic Majorana ZMs inherited from Majorana fermions, we found a pair of unexpected asymptotic ZMs localized at the two ends, which coalesce with those two exact bosonic Majorana ZMs in the thermodynamic limit. 
We have also shown how to generate new bosonic field theories out of our squaring map, in particular, we derived a squared Jackiw-Rebbi field theory with bosonic solitons. Finally, we presented the squared Harper-Hofstadter model with pairing to illustrate the interplay between bosonic topological invariants and midgap states. 

Our no-go Theorem 1, which establishes that gapped stable free bosonic ground states always display even parity, is at odds with the equivalent free fermionic situation. The change in fermion parity between equilibrium phases is an
indicator of the transition between gapped topologically distinct vacua, 
and this fermion parity switch is the invariant that
finds extension in particle-conserving interacting fermionic 
systems\cite{Ortiz14,Ortiz16}. In hindsight, what a topologically 
non-trivial interacting bosonic vacuum represents, constitutes a fundamental 
question. But it is perhaps equally fundamental to establish whether interactions may induce a ground-state topological transition between parity-distinct gapped bosonic phases. In Ref.\,[\onlinecite{Lerma19}], 
an exactly solvable $p$-wave pairing model for two bosonic species was introduced, that shares some commonalities with the $p+ip$ fermionic model.
Contrary to the latter, in the bosonic case the transition separates a 
gapless, fragmented singlet pair Bose-Einstein condensate from a pair Bose gapped superfluid. This raises the concern that boson parity switches may be fundamentally non-existent in interacting bosonic systems, since a \emph{gapless} Bose-Einstein condensate may intervene. One would like to find interacting bosonic models with topologically inequivalent \emph{gapped} phases. This is an open question for future studies.

\acknowledgments

We gratefully acknowledge correspondence with H. Schulz-Baldes and P. T. Nam on the problem of many-body stability. A. A. acknowledges insightful discussions with Barry C. Sanders and David Feder on zero modes in topological bosons. Work at Dartmouth was partially supported by the NSF through Grants No. PHY-1620541 and OIA-1921199, and the Constance and Walter Burke Special Projects Fund in Quantum Information Science. 
A. A. acknowledges support from the Quantum Alberta Initiative by the Government of Alberta.
E. C. gratefully acknowledges support from the Office of Research Advancement, SUNY 
Polytechnic Institute, in the form of a Seed Grant for the year 2019. G. O. acknowledges partial support by the DOE grant DE-SC0020343. The IU Quantum Science and Engineering Center is 
supported by the Office of the Vice Provost for Research through its Emerging Areas of Research program.

\appendix

\section{Proof of the theorem in Sec.\,\ref{statistic}}
\label{proofs}

As in the main text, $H_\tau = \tau_3 H_b$ denotes 
the effective BdG Hamiltonian of a many-body free-boson system with $N$ single-particle states and $H_b\geq 0$. 
We write $N_0\equiv\dim\ker H_\tau$ and let $m\leq N_0$ be the number of size-two Jordan blocks (necessarily at zero frequency as a consequence of positive semi-definiteness\cite{Gohberg}) in the Jordan normal form of $H_\tau$. Similarly, we write $2n\equiv N_0-m$ for the (necessarily even) number of size-one Jordan blocks. Note that there are then $2N'\equiv 2N-2n-2m$ non-zero eigenvalues $\epsilon_\nu$ of $H_\tau$. We first establish the following:

\begin{lemma}
\label{zerolemma}
 The $2n$ kernel vectors of $H_\tau$ corresponding to Jordan blocks of size one denoted by  $\ket{\phi_{0,j}^\pm}$ with $j=1,\ldots, n$, can be taken to satisfy $\braket{\phi_{0,j}^\pm|\tau_3|\phi_{0,\ell}^\pm}
 =\pm \delta_{j\ell}$ and $\braket{\phi_{0,j}^\pm|\tau_3|\phi_{0,\ell}^\mp}=0$. The remaining $m$ kernel 
 vectors, which we will denote by $\ket{\rho_{0,j}}$, can be chosen to be $\tau_3$-orthogonal to the vectors 
 $\ket{\phi^\pm_{0,\ell}}$ and satisfy $\braket{\rho_{0,j}|\tau_3|\rho_{0,\ell}}=0$ for all $j,\ell\in\{1,\ldots,m\}$. 
 Furthermore, we can find $m$ generalized eigenvectors $\ket{\chi_{0,j}}$ satisfying $H_\tau \ket{\chi_{0,j}}=
 \ket{\rho_{0,j}}$ and $\braket{\chi_{0,j}|\tau_3|\chi_{0,\ell}}=0$ and  $\braket{\chi_{0,j}|\tau_3|\rho_{0,\ell}}=
 t_j\delta_{j,\ell}$, with $t_j$ either $1$ or $-1$, for all $j,\ell\in\{1,\ldots,m\}$ that are $\tau_3$-orthogonal to 
 the vectors $\ket{\phi^\pm_{0,j}}$.
\end{lemma}
\begin{proof}
Theorem 5.1.1 in Ref.\,[\onlinecite{Gohberg}] says there exists a matrix $S$ that induces the 
transformation $H_\tau =S^{-1} J S$ with
\begin{align*}
J = \textup{diag}\left(\mathcal{E},-\mathcal{E},0,\ldots,0,J_0,\ldots,J_0\right),
\end{align*}
where $\mathcal{E}\equiv \textup{diag}\left(\epsilon_1,\ldots,\epsilon_{N'}\right)$, $0$ appears $2n$ 
times and $J_0 \equiv \begin{bmatrix}0&1\\0&0\end{bmatrix}$
appears $m$ times, and  $\tau_3=S^\dag P S$ with
\begin{align*}
P&=\textup{diag}\left(\id_{N'},-\id_{N'},P_0\right),
\\
P_0 &= \textup{diag}\left(s_1,\ldots,s_{2n},t_1\sigma_1,\ldots,t_{m} \sigma_1\right) ,
\end{align*}
where the $s_\ell$'s and $t_\ell$'s are either $1$ or $-1$. Furthermore, Theorem A.1.1 in 
Ref.\,[\onlinecite{Gohberg}] tells us that $P$ and $\tau_3$ have the same number of positive and 
negative eigenvalues. Clearly the eigenvalues of $P$ are $\pm 1$ and so they must each have multiplicity $N$. 
We see that $N'$ of the $+1$ ($-1$) eigenvalues are accounted for in the first (second) $N'$ diagonal elements 
of $P$. We also have $m$ of them originating from the $t_\ell \sigma_1$ factors in $P_0$. The remaining 
number of $+1$ ($-1$) eigenvalues is $N-N'-m$. Recalling that $N=N'+m+n$, we conclude that  $n$ of the 
$s_\ell$'s are $+1$ and the remaining are $-1$. Without loss of generality, (by rearranging the Jordan normal 
form as necessary) we can take $s_j=1$ for $j=1,\ldots, n$ and $s_j=-1$ for $j=n+1,\ldots, 2n$.  Now define
\begin{align*}
\ket{\phi_{0,j}^+} \equiv S^{-1}\ket{2N'+j}, \quad \ket{\phi_{0,j}^-} \equiv S^{-1}\ket{2N'+n+j}  ,
\end{align*}
for $j=1,\ldots, n$. Clearly $H_\tau\ket{\phi_{0,j}^\pm}=0$ and by virtue of the relation $\tau_3=S^\dag P S$ 
we have $\braket{\phi_{0,j}^\pm|\tau_3|\phi_{0,\ell}^\pm}=\pm\delta_{j\ell}$ and 
$\braket{\phi_{0,j}^\pm|\tau_3|\phi_{0,\ell}^\mp}=0$.

In the same vein, we can take 
\begin{align*}
\ket{\rho_{0,j}} &\equiv S^{-1}\ket{2N'+2n+2j-1} ,
\\
 \ket{\chi_{0,j}} &\equiv S^{-1}\ket{2N'+2n+2j} ,
\end{align*}
which can be checked to satisfy claimed properties. 
\end{proof}

Now we restate Theorem of Sec.\,\ref{statistic} and give a proof. 

\smallskip
\noindent\textbf{Theorem (ZMs, [\onlinecite{ColpaZMs}]).}
\textit{For the effective BdG Hamiltonian $H_\tau=\tau_3H_b$, let $2n$ and $m$ be the number of linearly independent zero eigenvectors associated to Jordan chains of length one and two, respectively. Then there are $n$ pairs of canonical boson $b^{\;}_{0,j},b_{0,j}^\dag$ that commute with the many-body Hamiltonian \(\widehat{H}_b\) and all other normal modes of the system. 
In addition, there exist $m$ pairs of Hermitian
operators $P_{0,j},Q_{0,j}$ that also commute with all other normal modes of the system and obey $[Q_{0,j},P_{0,\ell}]=i\delta_{j\ell}$, $[\widehat{H}_b,P_{0,j}]=0$, 
and $[\widehat{H}_b,Q_{0,j}]= (i/\mu_j)P_{0,j}$, with \(\mu_j>0\).}

\begin{proof}
Let $\mathcal{D}_0$ be the span of the eigenvectors $\ket{\phi^\pm_{0,j}}$ specified in Lemma \ref{zerolemma}.
In order to construct these bosonic ZMs, we need a basis $\{\ket{\psi_{0,j}^\pm}\}_{j=1}^n$ for $\mathcal{D}_0$ satisfying
\begin{equation}
\begin{aligned}\label{threeconds}
&\braket{\psi_{0,j}^\pm|\tau_3|\psi_{0,\ell}^\pm} = \pm \delta_{j\ell},
\\
&\braket{\psi_{0,j}^\pm|\tau_3|\psi_{0,\ell}^\mp} = 0,
\\
&\ket{\psi_{0,j}^-} =\mathcal{C}\ket{\psi_{0,\ell}^+},
\end{aligned}
\end{equation}
where $\mathcal{C}\equiv \tau_1\mathcal{K}$. With such a basis, we can construct the $n$ pairs 
$(b^{\;}_{0,j},b_{0,j}^\dag)$ in a way identical to Eq.\,(\ref{quasi-particle}).

It is not \textit{a priori} true that the basis $\{\ket{\phi_{0,j}^\pm}\}$ in Lemma \ref{zerolemma} 
satisfies the third condition.  Thus, we define
\begin{align*}
    \ket{\psi_{0,j}^+} = \sum_{\ell=1}^n \alpha_{j\ell} \ket{\phi_{\ell,0}^+}.
\end{align*}
where $\alpha_{j\ell}\in\mathbf{C}$.
In order to ensure that $\braket{\psi_{0,j}^+|\tau_3|\psi_{0,\ell}^+}=\delta_{j\ell}$ the matrix $\alpha$ with 
elements given by $\alpha_{j\ell}$ must be unitary. Thus, out of the original $2n^2$  real parameters 
(two for each $\alpha_{j\ell}$), we are left with $n^2$ free. Now, we wish to impose the condition
\begin{align}\label{cc}
\mathcal{O}_{j\ell}\equiv \braket{\psi_{0,j}^+|\tau_3\mathcal{C}|\psi_{0,\ell}^+} = 0,
\end{align}
for all $j$ and $\ell$. Because of $\mathcal{O}_{j\ell}=-\mathcal{O}_{\ell j}$, Eq.\,(\ref{cc}) 
imposes $n(n-1)/2$ independent conditions. Taking the real and imaginary parts of $\mathcal{O}_{j\ell}=0$ 
yields $n(n-1)$ equations that must be satisfied. Given we have $n^2$ free parameters, we have enough 
freedom to ensure $\mathcal{O}_{j\ell}=0$ for all $j$ and $\ell$. The remaining $n$  free real parameters 
can be associated with the arbitrary phases of each $\ket{\psi_{0,j}^+}$. After choosing $\alpha_{j\ell}$ appropriately, we define $\ket{\psi^-_{0,j}}\equiv\mathcal{C}\ket{\psi^+_{0,\ell}}$. Then, the three conditions in Eq.\,(\ref{threeconds}) are satisfied. 

Moving to the higher order Jordan blocks, we wish to construct a set of $m$ eigenvectors $\ket{P_{0,j}}$ and $m$ generalized eigenvectors $\ket{Q_{0,j}}$ of $H_\tau$ satisfying
\begin{align}
H_\tau\ket{P_{0,j}} &= 0 , \label{pqcondsi}
\\
H_\tau\ket{Q_{0,j}} &= -\frac{i}{\mu_j}\ket{P_{0,j}} , \label{pqcondsii}
\\
\braket{Q_{0,j}|\tau_3|\psi^\pm_{0,\ell}} &= \braket{P_{0,j}|\tau_3|\psi^\pm_{0,\ell}}=0 , \label{pqcondsiii}
\\
\braket{P_{0,j}|\tau_3|P_{0,\ell}} &= 0 , \label{pqcondsiv}
\\
\braket{Q_{0,j}|\tau_3|Q_{0,\ell}} &=0 , \label{pqcondsv}
\\
\braket{Q_{0,j}|\tau_3|P_{0,\ell}} &= i\delta_{j\ell} , \label{pqcondsvi}
\\
\mathcal{C}\ket{P_{0,j}} &= -\ket{P_{0,j}} , \label{pqcondsvii}
\\
\mathcal{C}\ket{Q_{0,j}} &= -\ket{Q_{0,j}} . \label{pqcondsviii}
\end{align}
We claim that $P_{0,j} = \bra{P_{0,j}}\tau_3 \hat{\Phi}$ and $Q_{0,j} = \bra{Q_{0,j}}\tau_3 \hat{\Phi}$ 
provide the desired Hermitian operators. 

We begin by letting $\ket{\rho_{0,j}}$ and $\ket{\chi_{0,j}}$ denote 
the vectors defined in Lemma \ref{zerolemma} and define
\begin{align*}
\ket{P_{0,j}} &= \sum_{\ell=1}^m \beta_{j\ell}\ket{\rho_{0,\ell}},
\\
\ket{Q_{0,j}} &= \frac{-i}{\mu_j}\sum_{\ell=1}^m \beta_{j\ell}\ket{\chi_{0,\ell}}+\sum_{\ell=1}^m \gamma_{j\ell}\ket{\rho_{0,\ell}},
\end{align*}
with $\mu_j^{-1}\equiv \braket{Q_{0,j}|H_\tau|Q_{0,j}}>0$, $\beta_{j\ell},\gamma_{j\ell}\in\mathbf{C}$. By construction, 
conditions \eqref{pqcondsi}, \eqref{pqcondsii}, and \eqref{pqcondsiv} are satisfied. Furthermore, condition \eqref{pqcondsiii} 
is satisfied when the plus sign is chosen. This fact paired with conditions \eqref{pqcondsvii} and \eqref{pqcondsviii} 
will ensure that the minus sign portion of condition \eqref{pqcondsiii} will be satisfied. Thus, we will show that 
conditions \eqref{pqcondsvii} and \eqref{pqcondsviii} can be satisfied first. Moving forward, condition \eqref{pqcondsvi} 
imposes $m(m+1)/2$ constraints on the $2m^2$ free parameters $\beta_{j\ell}$, $\gamma_{j\ell}$. Condition 
\eqref{pqcondsvii} imposes $m$ more constraints. This leaves us with $3m(m-1)/2\geq 0$ free parameters. 
Now, if the vectors $\ket{Q_{0,j}}$ satisfy \eqref{pqcondsviii} then they will satisfy condition \eqref{pqcondsvi} 
for $j=\ell$. Noting that, by virtue of the charge conjugation properties of the vector $\ket{\psi^\pm_{0,j}}$,  the 
span of the vectors $\{\ket{Q_{0,j}},\ket{P_{0,j}}\}_{j=1}^m$ is invariant under the action of $\mathcal{C}$. Thus, we can write
\begin{align*}
\mathcal{C}\ket{Q_{0,j}} = \sum_{\ell=1}^m \left(z_{j\ell}\ket{Q_{0,\ell}} + w_{j\ell}\ket{P_{0,\ell}}\right),
\end{align*}
with $z_{j\ell},w_{j\ell}\in\mathbf{C}$. Projecting with $\bra{P_{0,k}}\tau_3$, and noting that for any vectors 
$\ket{v}$ and $\ket{w}$ we have $\braket{v|\mathcal{C}|w} = \braket{w|\mathcal{C}|v}$ and $\tau_3 
\mathcal{C}=-\mathcal{C} \tau_3$, we obtain $i\delta_{jk} = -iz_{jk}$. Thus
\begin{align*}
\mathcal{C}\ket{Q_{0,j}} = -\ket{Q_{0,j}}+ \sum_{\ell=1}^m w_{j\ell}\ket{P_{0,\ell}}.
\end{align*}
To ensure $\mathcal{C}\ket{Q_{0,j}}=-\ket{Q_{0,j}}$ we can shift $\ket{Q_{0,j}}$ by the appropriate linear 
combination of the vectors $\ket{P_{0,\ell}}$ to make the second term vanish. This imposes no more constraints 
and preserves the already satisfied conditions.  By imposing condition \eqref{pqcondsv} for $j\neq \ell$ we 
obtain $m(m-1)/2$ more constraints. Altogether we still retain $3m(m-1)/2-m(m-1)/2=m(m-1)\geq 0$ free parameters. 
Thus, all 8 conditions can be satisfied by the appropriate choice of constants $\beta_{j\ell}$ and $\gamma_{j\ell}$. 
\end{proof}

\section{Symmetry reduction}
\label{afterandbefore}

Let \(\{H_f\}\) denote an ensemble of BdG Hamiltonians that 
commute with \(Q_f\), where 
\begin{align*}
H_f=|\resizebox{0.88\hsize}{!}{$\uparrow\rangle\langle \uparrow\!|\otimes K +|\!\uparrow\rangle\langle \downarrow \! |\otimes \Delta- |\!\downarrow\rangle\langle \uparrow \!| \otimes \Delta^* -|\!\downarrow\rangle\langle \downarrow \!|\otimes K^*$}
\end{align*}
is written in terms of eigenvectors $|\!\!\uparrow\rangle\equiv \begin{bmatrix}1\\ 0 \end{bmatrix}$, $|\!\!\downarrow\rangle\equiv \begin{bmatrix}0\\ 1 \end{bmatrix}$ of \(\sigma_3\). 
We are interested in the symmetry reduction of the squared ensemble \(\{\tau_3H_f^2\}\), induced by a symmetry \(Q_f\) that commutes
with \(\tau_3\). The eigenvalues of \(Q_f\) are determined by the eigenvalues of the Hermitian operator \(q\) defined in Eq.\,\eqref{diagq}, the spectrum of which 
we denote as \(\sigma(q)\). The blocks of the squared ensemble consists of reduced effective BdG Hamiltonians with respect to a reduced metric. It is important for classification purposes to understand the general structure of both
the reduced metric and the reduced effective BdG Hamiltonian. 
There are three distinct cases to analyze. 

\smallskip

{\textit{Case 1: \(\kappa\in\sigma(q)\) and \(-\kappa \notin \sigma(q)\)}}. Let $|\kappa,\nu\rangle$, $\nu=1,\ldots, m$,  
denote a complete set of orthonormal eigenvectors of \(q\) for
the eigenvalue \(\kappa\), written in the mode basis of the original many-boson Hamiltonian. 
Since, by assumption, \(-\kappa\) is not an eigenvalue of \(q\), the eigenvectors
of \(Q_f\) associated to \(\kappa\) are \(|\! \! \uparrow\rangle|\kappa,\nu\rangle\) and the eigenvectors associated to 
\(-\kappa\) are \(|\! \! \downarrow\rangle|\kappa,\nu\rangle^*\), with \(|\kappa,\nu\rangle^*\equiv{\cal K}|\kappa,\nu\rangle\). The associated canonical fermions, partially labeled by the conserved quantum number \(\kappa\), are
\begin{align}
\label{alfa_fermions}
c_{\kappa,\nu}^\dagger&\equiv \hat{\Psi}^\dagger |\!\uparrow\rangle|\kappa,\nu\rangle= \hat{\psi}^\dagger|\kappa,\nu\rangle,\\
c^{\,}_{\kappa,\nu}&\equiv \hat{\Psi}^\dagger |\!\downarrow\rangle|\kappa,\nu\rangle^*= 
\hat{\psi}^\text{T}|\kappa,\nu\rangle^*=\langle \kappa,\nu|\hat{\psi}.
\end{align}

There is one and only one block of \(\widehat{H}_f\) featuring these degrees of freedom. To compute it, let 
\begin{align*}
P_{\kappa}=|\!\uparrow\rangle\langle \uparrow \! |\otimes \sum_{\nu=1}^m|\kappa,\nu\rangle\langle \kappa,\nu|
+|\!\downarrow\rangle\langle\downarrow\!|\otimes\sum_{\nu=1}^m|\kappa,\nu\rangle^*\langle\kappa,\nu|^*
\end{align*}
denote the projector onto the eigenstates of \(Q_f\) associated to \(\pm\kappa\). Then, the many-body block is  
\begin{align*}
\widehat{H}_{f,\kappa}=\frac{1}{2}\hat{\Psi}^\dagger P_{\kappa}H_fP_{\kappa}\hat{\Psi}=
\hat{\Psi}_\kappa^\dagger H_{f,\kappa}\hat{\Psi}_\kappa ,
\end{align*}
in terms of the Nambu array
$$\hat{\Psi}_\kappa^\dagger=\begin{bmatrix} c^\dagger_{\kappa,1}&\cdots&c^\dagger_{\kappa,m}&c^{\;}_{\kappa,1}&\cdots&c^{\;}_{\kappa,m}\end{bmatrix}.$$
By construction, 
\(
H_{f,\kappa}=\begin{bmatrix}K_\kappa &\Delta_\kappa\\ -\Delta_\kappa^* & -K_\kappa^*\end{bmatrix}.
\)
However,
\begin{align*}
[\Delta_\kappa]_{\nu\nu'}=\langle \kappa,\nu|\Delta|\kappa,\nu'\rangle^*=0
\end{align*}
because  \([H_f,Q_f]=0\) and \(\kappa\neq -\kappa\) by assumption (notice that complex conjugation affects only the nearest vector to the left). Hence, the many-body block is
\begin{align*}
\widehat{H}_{f,\kappa}=\hat{\psi}_\kappa^\dagger K_\kappa \hat{\psi}_\kappa^{\,}, 
\quad \hat{\psi}^\dagger_\kappa=\begin{bmatrix}c^\dagger_{\kappa,1}&\cdots&c_{\kappa,m}^\dagger
\end{bmatrix}, 
\end{align*}
and we see that the number of \(\kappa\) fermions is conserved. Finally, the squaring map ${\mathscr S}$ yields 
the free-boson system 
\begin{align*}
\widehat{H}_{b,\kappa}\equiv \frac{1}{2}\hat{\Phi}_\kappa^\dagger H_{f,\kappa}^2\hat{\Phi}_\kappa=
\hat{\phi}_\kappa^\dagger K_\kappa^2\hat{\phi}_\kappa^{\,} ,
\end{align*}
in terms of canonical bosons $a_{\kappa,\nu}^\dagger\equiv \hat{\phi}^\dagger|\kappa,\nu\rangle$. 

\smallskip

{\em Case 2: {\(0\in\sigma(q)\)}}.  Let \(|0,\nu\rangle\), $\nu=1,\ldots,m$,  denote a complete orthonormal set of eigenvectors of \(q\) belonging to zero eigenvalue. 
The associated fermionic degrees of freedom are just as in Eq.\,\eqref{alfa_fermions}, simply set 
\(\kappa =0 \). Proceeding as before we obtain the many-body block
\(\widehat{H}_{f,0}=\frac{1}{2}\hat{\Psi}^\dagger_0H_{f,0}\hat{\Psi}^{\;}_0\), where the 
single-particle and pairing blocks of \(H_{f,0}\) are  
\begin{align*}
[K_0]_{\nu\nu'}&=\langle 0,\nu|K|0,\nu'\rangle=[K_0]_{\nu'\nu}^*,\\
[\Delta_0]_{\nu\nu'}&=\langle 0,\nu|\Delta|0,\nu'\rangle^*=-[\Delta_0]_{\nu'\nu}.
\end{align*}
Unlike the previous case, the pairing term need not vanish because \(0=-0\). By comparison with Case 3 below, zero eigenvalue  is also special because no further reduction of the single-particle block \(\widehat{H}_{f,0}\) is possible. The squaring map induces the block transformation \(\tau_3H_{f,0}^2\)
in terms of \(\tau_3=\sigma_3\otimes \mathds{1}_m\). 

\smallskip

{\textit{Case 3: \(\kappa, -\kappa \in\sigma(q)\) and \(\kappa\neq0\)}}.   
This is the most elaborate case and, together with Case 2, it comprises 
familiar symmetries like spin rotations and lattice translations.
Let \(|\kappa, \nu\rangle\), $\nu=1,\ldots, m$, and \(|\!-\!\kappa,\bar \nu\rangle\), $\bar \nu=1,\ldots, n$, denote complete orthonormal sets of eigenvectors of \(q\) associated to the indicated eigenvalues. Notice that we do not assume identical degeneracy for \(\kappa\) and \(-\kappa\). They do coincide for spin and crystal momenta but that need not be the case in general. As we will see, the case \(m\neq n\) introduces exotic features into the bosonic problem. Now, \(\kappa, -\kappa\)
are also eigenvalues of \(Q_f\). The corresponding complete sets of orthonormal eigenvectors are
\(|\! \! \uparrow\rangle|\kappa,\nu\rangle\), \(|\!\downarrow\rangle|-\kappa,\bar\nu\rangle^*\),  and 
\(|\!\uparrow\rangle|-\kappa,\bar\nu\rangle\), \( |\!\downarrow\rangle|\kappa,\nu\rangle^*\), respectively. 
The fermionic degrees of freedom are
\begin{align*}
c^\dagger_{\kappa,\nu}\equiv \hat{\psi}^\dagger|\kappa,\nu\rangle,\quad\quad 
c^\dagger_{-\kappa,\bar\nu}\equiv\hat{\psi}^\dagger |-\kappa,\bar\nu\rangle .
\end{align*}

The many-body block can be calculated as before in terms of a projector \(P_{\pm\kappa}\) onto the subspace associated to the eigenvalues \(\pm\kappa\) of \(Q_f\). The resulting
many-body block can be characterized as \(\widehat{H}_{f,\pm\kappa}=\frac{1}{2}\hat{\Psi}_{\pm\kappa}^\dagger
H_{f,\pm\kappa}\hat{\Psi}_{\pm\kappa}\) in terms of the Nambu array
\begin{align*}
\hat{\Psi}_{\pm\kappa}^\dagger\resizebox{0.93\hsize}{!}{$\equiv\!\begin{bmatrix} c^\dagger_{\kappa,1}\cdots c^\dagger_{\kappa,m}\,c^\dagger_{-\kappa,1}\cdots c^\dagger_{-\kappa,n}\,c^{\;}_{\kappa,1}\cdots c^{\;}_{\kappa,m}\,c^{\;}_{-\kappa,1}\cdots c^{\;}_{-\kappa,n}\end{bmatrix}$}
\end{align*}
and the BdG Hamiltonian
\begin{align*}
\label{apparentred}
H_{f,\pm\kappa}=
\begin{bmatrix} 
K_1&0                                   & 0                & \Delta_1\\
0                &K_2                    & \Delta_2      & 0\\
0                &-\Delta_1^*        & -K_1^*       & 0\\
-\Delta_2^* & 0                    & 0                &-K_2^* 
\end{bmatrix},
\end{align*}
with \(K_1\) and $K_2$ Hermitian \(m\times m\) and \(n\times n\) matrices, respectively, and \(\Delta_1=-\Delta_2^\text{T}\) an \(m\times n\) rectangular matrix. Explicit expressions are (the conjugation operation acts only on the bra or ket directly to
the left of it)
\begin{align*}
[K_1]_{\nu\nu'}&=\langle\kappa,\nu|K|\kappa,\nu'\rangle \ , \      
[K_2]_{\bar\nu \bar\nu'}=\langle -\kappa,\bar\nu|^*K|\!-\!\kappa,\bar\nu'\rangle^*,\\
[\Delta_1]_{\nu\bar\nu}&=\langle \kappa,\nu|\Delta|-\kappa,\bar\nu\rangle^*=-[\Delta_2]_{\bar\nu \nu}.
\end{align*}

Note that the single-particle Hamiltonian \(H_{f,\pm \kappa}\) is ``reducible'' but the many-body block is not. So, we have 
\begin{align*}
\widehat{H}_{f,\pm\kappa}=\widetilde\Psi_{\pm\kappa}^\dagger H_{\pm\kappa}\widetilde\Psi^{\;}_{\pm\kappa}
-(\mbox{tr}\,K_1 -\mbox{tr}\,K_2) ,
\end{align*}
in terms of the \textit{generic} 
Hermitian matrix
\begin{eqnarray}
\label{unconstrainedpm}
H_{\pm\kappa}=\begin{bmatrix} K_1 &\Delta_1\\
-\Delta_{2}^* & -K_2^*
\end{bmatrix}
\end{eqnarray}
and the (not Nambu!) array 
\begin{align*}
\widetilde\Psi_{\pm\kappa}^\dagger=
\begin{bmatrix} 
c_{\kappa,1}^\dagger &\cdots & c_{\kappa,m}^\dagger&c^{\;}_{-\kappa,1}&\cdots&c^{\;}_{-\kappa,n}
\end{bmatrix}.
\end{align*}
This form of the many-body block makes it clear that it commutes with the charge
\begin{align*}
\widehat{N}_{\pm\kappa}=\sum_{\nu=1}^mc^\dagger_{\kappa,\nu}c^{\;}_{\kappa,\nu}-
\sum_{\bar\nu=1}^n c^\dagger_{-\kappa,\bar\nu}c^{\;}_{-\kappa,\bar\nu}.
\end{align*}
The free-boson system induced by the squaring map ${\mathscr S}$,
\begin{align*}
\widehat{H}_{b,\pm\kappa}&=
\frac{1}{2}\hat{\Phi}_{\pm\kappa}^\dagger H_{f,\pm\kappa}^2\hat{\Phi}_{\pm\kappa}\\
&=\widetilde{\Phi}_{\pm\kappa}^\dagger H_{\pm\kappa}^2\widetilde{\Phi}^{\;}_{\pm\kappa}+(\mbox{tr}K_1-\mbox{tr}K_2),
\end{align*}
is then written in terms of the ({\em not} Nambu!) array
\begin{align*}
\widetilde{\Phi}_{\pm\kappa}^\dagger=
\begin{bmatrix} 
a_{\kappa,1}^\dagger &\cdots & a_{\kappa,m}^\dagger&a^{\;}_{-\kappa,1}&\cdots&a^{\;}_{-\kappa,n}
\end{bmatrix},
\end{align*}
of canonical bosons \(a_{\kappa,\nu}^\dagger=\hat{\phi}^\dagger |\kappa,\nu\rangle,\,
a^{\;}_{-\kappa,\bar\nu}=\langle -\kappa,\bar\nu|\hat{\phi}\), and the generic Hermitian 
\(H_{\pm\kappa}\) of Eq.\,\eqref{unconstrainedpm}. Notice that 
\begin{align*}
[\widetilde{\Phi}_{\pm\kappa}, \widetilde{\Phi}_{\pm\kappa}^\dagger]=\tau_{m,n},\quad 
\tau_{m,n}\equiv\begin{bmatrix} \mathds{1}_m &0\\ 0 &-\mathds{1}_n\end{bmatrix}.
\end{align*}

Hence, the irreducible effective BdG Hamiltonian is \(\tau_{m,n}H_{f,\pm\kappa}^2\), 
which is \(\tau_{m,n}\) pseudo-Hermitian. The pseudo-unitary transformations associated to Gaussian isometries of the array \(\widetilde{\Phi}_{\pm \kappa}\) satisfy \(U\tau_{m,n}U^\dagger=\tau_{m,n}\). The case \(m=n\) is certainly well understood, as we have seen. The geometric and algebraic features of the case \(m\neq n\) are treated in the mathematical literature under the heading of ``indefinite linear algebra,'' see Ref.\,[\onlinecite{Gohberg}]. 

\medskip
\medskip
\medskip

\section{On general stable bosonic ensembles}
\label{simpleproof}

We begin with the observation (see  Table\,\ref{tab:Particle_non-conserving_squared_reducible}) that the time-reversal symmetry of $\tau_3\widetilde{H}_f^2(\mathbf{k})$ always commutes with $\tau_3$, while both particle-hole  and chiral symmetries of $\tau_3\widetilde{H}_f^2(\mathbf{k})$ always anticommute with $\tau_3$. 
To see that Table\,\ref{tab:Bosonic_periodic_table} 
also holds for ensembles of stable free-boson systems not arising from the squaring procedure, we need to show that the time-reversal symmetry of $\widetilde{H}_\tau(\mathbf{k})\equiv \tau_3\widetilde{H}_b(\mathbf{k})$ (with $\widetilde{H}_b^\dagger(\mathbf{k})=\widetilde{H}_b(\mathbf{k})$, $\widetilde{H}_b(\mathbf{k})>0$) cannot anticommute with $\tau_3$, and particle-hole or 
chiral symmetries cannot commute with $\tau_3$. 
We show that this is the case by contradiction.

(a) Suppose that the time-reversal symmetry $U_T^\dagger\mathcal{K}$, $U_T^\dagger\widetilde{H}_\tau^*(\mathbf{-k})U_T=\widetilde{H}_\tau(\mathbf{k})$, anticommutes with $\tau_3$, i.e., $\{\tau_3, U_T\}=0$. Then we have $U_T^\dagger\widetilde{H}_b^*(\mathbf{-k})U_T=-\widetilde{H}_b(\mathbf{k})$, which translates into a particle-hole symmetry of $\widetilde{H}_b(\mathbf{k})$, thus violating the stability condition $\widetilde{H}_b(\mathbf{k})>0$.

(b) Suppose that the particle-hole symmetry $U_C^\dagger\mathcal{K}$, $U_C^\dagger\widetilde{H}_\tau^*(\mathbf{-k})U_C=-\widetilde{H}_\tau(\mathbf{k})$, commutes with $\tau_3$, i.e., $[\tau_3, U_C]=0$. 
Then we have $U_C^\dagger\widetilde{H}_b^*(\mathbf{-k})U_C=-\widetilde{H}_b(\mathbf{k})$, which also translates into a particle-hole symmetry of $\widetilde{H}_b(\mathbf{k})$, thus violating the stability condition $\widetilde{H}_b(\mathbf{k})>0$.

(c) Suppose that the chiral symmetry $U_S^\dagger$, 
$U_S^\dagger\widetilde{H}_\tau(\mathbf{k})U_S=-\widetilde{H}_\tau(\mathbf{k})$, commutes with $\tau_3$, i.e., $[\tau_3, U_S]=0$. 
Then we have $U_S^\dagger\widetilde{H}_b(\mathbf{k})U_S=-\widetilde{H}_b(\mathbf{k})$, which translates into a chiral symmetry of $\widetilde{H}_b(\mathbf{k})$, thus violating the stability condition $\widetilde{H}_b(\mathbf{k})>0$. 

Thus, we expect that our Table\,\ref{tab:Bosonic_periodic_table} also holds for generic stable free-boson systems with metric $\tau_3$.  For the ``unbalanced" metric $\tau_{m,n}$ (see Table\,\ref{squared_subensembles}), only time-reversal symmetry is possible, and with simple modifications, the reasoning (a) above is still valid.

A simple argument points to the importance of the stability condition in establishing the threefold way classification. 
Let \(\{H_\tau(\mathbf{k})=\tau_3H_b(\mathbf{k})\}\) denote an ensemble of effective Bloch BdG Hamiltonians with  \(H_b(\mathbf{k})>0\), but arbitrary otherwise. By block-diagonalizing the linear Gaussian many-body symmetries of the ensemble, one obtains a family of sub-ensembles \(\{H_\tau^{(i)}(\mathbf{k})\}\) and associated metrics  
\(\tau^{(i)}_{m,n}\). This characterization of the sub-ensembles 
relies heavily on the assumption \(H_b(\mathbf{k})>0\) which guarantees that the quasiparticles of \(H_b(\mathbf{k})\) are canonical bosons. Moreover, because it is true of the parent ensemble, the Hermitian Bloch Hamiltonians \(\tau_{m,n}^{(i)}H_\tau^{(i)}(\mathbf{k})\) satisfy  $\tau_{m,n}^{(i)}H_\tau^{(i)}(\mathbf{k})>0$ necessarily. 
Hence, if irreducible, the Hermitian ensembles \(\{\tau_{m,n}^{(i)}H_\tau^{(i)}(\mathbf{k})\}\) can only  be of class A, AI, or AII according to the usual tenfold way.



\begin{thebibliography}{100}


\bibitem{Gell-Mann51}
M. Gell-Mann and F. Low,
\textit{Bound States in Quantum Field Theory},
Phys. Rev. \textbf{84}, 350 (1951). 

\bibitem{Schnyder}
A. P. Schnyder, S. Ryu, A. Furusaki, and A. W. W. Ludwig, 
\textit{Classification of topological insulators and superconductors in three spatial dimensions},
Phys. Rev. B \textbf{78}, 195125 (2008).

\bibitem{Kitaev2009}
A. Kitaev, 
\textit{Periodic table for topological insulators and superconductors},
AIP Conf. Proc. {\bf 1134}, 22 (2009).

\bibitem{Ryu}
S. Ryu, A. P. Schnyder, A. Furusaki, and A. W. W. Ludwig, 
\textit{Topological insulators and superconductors: tenfold way and dimensional hierarchy}, 
New J. Phys. {\bf 12}, 065010 (2010).

\bibitem{Chiu}
C.-K. Chiu, J. C. Y. Teo, A. P. Schnyder, and S. Ryu, 
\textit{Classification of topological quantum matter with symmetries}, 
Rev. Mod. Phys. {\bf 88}, 035005 (2016).

\bibitem{Alase20}
A. Alase, E. Cobanera, G. Ortiz, and L. Viola, 
\textit{Matrix factorization approach to the bulk-boundary correspondence and the stabilty of zero modes},
in preparation (2020). 

\bibitem{Read00}
N. Read and D. Green,
\textit{Paired states of fermions in two dimensions with breaking of parity and time-reversal symmetries and the fractional quantum Hall effect},
Phys. Rev. B \textbf{61}, 10267 (2000).

\bibitem{Kitaev2001}
A. Y. Kitaev, 
\textit{Unpaired Majorana fermions in quantum wires}, 
Phys.-Usp. {\bf 44}, 131 (2001).

\bibitem{Deng2012}
S. Deng, L. Viola, and G. Ortiz, \textit{Majorana modes in
time-reversal invariant $s$-wave topological superconductors}, 
Phys. Rev. Lett. {\bf 108}, 036803 (2012). 

\bibitem{Deng2014}
S. Deng, G. Ortiz, A. Poudel, and L. Viola, 
\textit{Majorana flat bands in $s$-wave gapless topological superconductors}, 
Phys. Rev. B {\bf 89}, 140507(R) (2014). 

\bibitem{Alicea}
J. Alicea, 
\textit{New directions in the pursuit of Majorana Fermions in solid state systems},
Rep. Prog. Phys. {\bf 75}, 076501 (2012).

\bibitem{Beenakker}
C. W. J. Beenakker, 
\textit{Search for Majorana Fermions in Superconductors},
Annu. Rev. Condens. Matter Phys. {\bf 4}, 113 (2013).

\bibitem{Blaizot}
J.-P. Blaizot and G. Ripka, 
{\it Quantum Theory of Finite Systems} (The MIT Press, Cambridge, 1986).

\bibitem{Ortiz14} 
G. Ortiz, J. Dukelsky, E. Cobanera, C. Esebbag, and C. Beenakker, 
\textit{Many-Body Characterization of Particle-Conserving Topological Superfluids},
Phys. Rev. Lett. {\bf 113}, 267002 (2014).

\bibitem{Ortiz16} 
G. Ortiz and E. Cobanera, 
\textit{What is a particle-conserving Topological Superfluid? The fate of Majorana modes beyond mean-field theory},
Ann. Phys. (NY) {\bf 372}, 357 (2016).

\bibitem{Schulz}
H. Schulz-Baldes, 
\textit{Signature and Spectral Flow of $J$-Unitary $\mathds{S}^1$-Fredholm Operators}, 
Integr. Equ. Oper. Theory {\bf 78}, 323 (2014).

\bibitem{Peano}
V. Peano and H. Schulz-Baldes,
\textit{Topological edge states for disordered bosonic systems}, 
J. Math. Phys. {\bf 59}, 031901 (2018).

\bibitem{Vincent20}
V. P. Flynn, E. Cobanera, and L. Viola, 
\textit{Deconstructing effective non-Hermitian dynamics in quadratic bosonic Hamiltonians}, arXiv:2003.03405, 
New J. Phys. https://doi.org/10.1088/1367-2630/ab9e87.

\bibitem{ColpaZMs}
J. H. P. Colpa, 
\textit{Diagonalization of the quadratic boson Hamiltonian with zero modes}: 
\textit{Part I: Mathematical}, Physica {\bf 134A}, 377 (1986); \textit{Part II: Physical}, {\bf 134A}, 417 (1986).

\bibitem{Gohberg}
I. Gohberg, P. Lancaster, and L. Rodman, 
{\it Indefinite Linear Algebra and Applications} (Birkh\"{a}user, Basel, 2005). 

\bibitem{AZ}
A. Altland and M. R. Zirnbauer, 
\textit{Nonstandard symmetry classes in mesoscopic normal-superconducting hybrid structures},
Phys. Rev. B {\bf 55}, 1142 (1997).

\bibitem{colpa1}
J. H. P. Colpa, 
\textit{Diagonalization of the quadratic boson Hamiltonian}, 
Physica A {\bf 93}, 327 (1978).

\bibitem{AlasePRL}
A. Alase, E. Cobanera, G. Ortiz, and L. Viola,
\textit{Exact Solution of Quadratic Fermionic Hamiltonians for Arbitrary Boundary Conditions}, 
Phys. Rev. Lett. {\bf 117}, 076804 (2016).

\bibitem{PRB1}
A. Alase,  E. Cobanera, G. Ortiz, and L. Viola,
\textit{Generalization of Bloch's theorem for arbitrary boundary conditions: Theory}, Phys. Rev. B {\bf 96}, 195133 (2017).

\bibitem{JPA}
E. Cobanera, A. Alase,  G. Ortiz, and L. Viola,
\textit{Exact solution of corner-modified banded block-Toeplitz eigensystems},
J. Phys. A {\bf 50}, 195204 (2017).

\bibitem{Lee}
D. H. Lee and J. D. Joannopoulos, 
\textit{Simple scheme for surface-band calculations. I}, 
Phys. Rev. B {\bf 23}, 4988 (1981).

\bibitem{Bechstedt}
F. Bechstedt, 
\textit{Principles of Surface Physics} (Springer-Verlag, Berlin, Heidelberg, 2003).

\bibitem{Cobanera18}
E. Cobanera, A. Alase, G. Ortiz, and L. Viola,
\textit{Generalization of Bloch's theorem for arbitrary boundary conditions: Interfaces and topological surface band structure},
Phys. Rev. B \textbf{98}, 245423 (2018).

\bibitem{Derezinski17}
J. Derezi\'nski,
\textit{Bosonic quadratic Hamiltonians},
J. Math. Phys. \textbf{58}, 121101 (2017).

\bibitem{Yaku}
V. Yakubovich and V. Starzhinskii, 
{\it Linear Differential Equations with Periodic Coefficients, Vol. 1} (Wiley, New York, 1975).

\bibitem{Zirnbauer10}
M. R. Zirnbauer,
\textit{Symmetry classes},
in: \textit{The Oxford Handbook of Random Matrix Theory},
edited by G. Akemann, J. Baik, and P. D. Francesco (Oxford University Press, 2011).

\bibitem{BLC}
D. Bernard and A. LeClair, 
\textit{A classification of non-{H}ermitian random matrices},
in: \textit{Statistical Field Theories}, 
edited by A. Cappelli and G. Mussardo (Springer, Dordrecht, 2002).

\bibitem{Kawabata}
K. Kawabata, K. Shiozaki, M. Ueda, and M. Sato,
\textit{Symmetry and topology in non-{H}ermitian physics},
Phys. Rev. X \textbf{9}, 041015 (2019).

\bibitem{Dyson}
F. J. Dyson, 
\textit{The Threefold Way: Algebraic Structure of Symmetry Groups and Ensembles in Quantum Mechanics},
J. Math. Phys. {\bf 3}, 1199 (1962).

\bibitem{Verbaarschot00}
J. J. M. Verbaarschot and T. Wettig,
\textit{Random Matrix Theory and Chiral Symmetry in {QCD}},
Annu. Rev. Nucl. Part. Sci. {\bf 50}, 343 (2000).

\bibitem{Cobanera}
E. Cobanera and G. Ortiz, 
\textit{Equivalence of topological insulators and superconductors},
Phys. Rev. B {\bf 92}, 155125 (2015).

\bibitem{rosenberg_j:qc2004a}
J. Rosenberg, 
\textit{A selective history of the Stone-von Neumann theorem},
in: {\it Operator Algebras, Quantization and Noncommutative Geometry}, 
edited by R. S. Doran and R. V. Kadison (American Mathematical Society, 2004).

\bibitem{Derezinski2006}
J. Derezi\'nski, 
\textit{Introduction to Representations of the Canonical Commutation and Anticommutation Relations}, 
in: \textit{Large Coulomb Systems},
edited by J. Derezi\'nski and H. Siedentop (Springer-Verlag, Berlin, Heidelberg, 2006). 

\bibitem{Shindou}
R. Shindou, R. Matsumoto, S. Murakami, and J.-i. Ohe, 
\textit{Topological chiral magnonic edge mode in a magnonic crystal},
Phys. Rev. B {\bf 87}, 174427 (2013).

\bibitem{gohberg}
I. Gohberg, M. A. Kaashoek, and  I. M. Spitkovsky, 
\textit{An overview of matrix factorization theory and operator applications},  
in: {\it Factorization and Integrable Systems}, 
edited by I. Gohberg, N. Manojlovic, and A. F. dos Santos (Birkh\"{a}user, Basel, 2003).

\bibitem{youla}
D. Youla and N. Kazanjian, 
\textit{Bauer-type factorization of positive matrices and the theory of matrix polynomials orthogonal on the unit circle}, 
IEEE Trans. Circuits Systems {\bf 25}, 57 (1978).

\bibitem{NoteHerbut}
A smooth fermion-to-boson map for Dirac Hamiltonians is also considered by P. S. Kumar, I. F. Herbut, and R. Ganesh, \textit{Dirac Hamiltonians for bosonic spectra}, arXiv:2001.02694. Our squaring procedure has no relation at all to this map, however. In this work, the authors start from an effectively positive-definite real symmetric matrix and the bosonic map involves doubling of the dimension of that matrix, while preserving spectral properties. Our squaring procedure maps an arbitrary fermionic Hermitian matrix to a bosonic one, preserving both its dimensionality and (only) its kernel. By construction, our map generates a positive-semidefinite bosonic Hamiltonian.

\bibitem{Zhou}
H. Zhou and J. Y. Lee,
\textit{Periodic table for topological bands with non-Hermitian symmetries},
Phys. Rev. B \textbf{99}, 235112 (2019).

\bibitem{Lieu}
S. Lieu, \textit{Topological symmetry classes for non-Hermitian models and connections to the bosonic Bogoliubov-de Gennes equation}, Phys. Rev. B {\bf 98}, 115135 (2018).

\bibitem{nomenclature}
Notice that the particle-hole symmetry of Ref.\,[\onlinecite{Kawabata}], defined in their Eq.\,(13), is different from the definition in our Eq.\,(\ref{DAZconditions}), which leads to a different nomenclature for pseudo-Hermitian symmetry classes in their Table\,VIII and IX. However, there is a one-to-one correspondence between these two nomenclatures for pseudo-Hermitian symmetry classes, which will become clear at the end of our Sec.\,\ref{Squaredensembleswithnon-vanishingpairing}. Our class A corresponds to ($\leftrightarrow$) their class A with their metric $\eta$, our class AIII $\leftrightarrow$ their class AIII with $\eta_-$, our class AI $\leftrightarrow$ their class AI with $\eta_+$, our class BDI $\leftrightarrow$ their class CI with $\eta_{+-}$, our class D $\leftrightarrow$ their class C with $\eta_-$, our class DIII $\leftrightarrow$ their class CII with $\eta_{+-}$, our class AII $\leftrightarrow$ their class AII with $\eta_+$, our class CII $\leftrightarrow$ their class DIII with $\eta_{+-}$, our class C $\leftrightarrow$ their class D with $\eta_-$, and finally our class CI $\leftrightarrow$ their class BDI with $\eta_{+-}$. The subscripts of $\eta_{\pm}$ and $\eta_{+-}$ specify commutation $(+)$ or anticommutation $(-)$ relations to internal symmetries.

\bibitem{SSH}
W. P. Su, J. R. Schrieffer, and A. J. Heeger,
\textit{Solitons in Polyacetylene}, 
Phys. Rev. Lett. {\bf 42}, 1698 (1979).

\bibitem{Berry}
M. V. Berry, 
\textit{Quantal phase factors accompanying adiabatic changes}, 
Proc. R. Soc. Lond. A {\bf 392}, 45 (1984).

\bibitem{Deng}
S. Deng, G. Ortiz, and L. Viola, 
\textit{Multiband $s$-wave topological superconductors: Role of dimensionality and magnetic field response}, 
Phys. Rev. B {\bf 87}, 205414 (2013).

\bibitem{JR}
R. Jackiw and C. Rebbi,
\textit{Solitons with fermion number $1/2$},
Phys. Rev. D \textbf{13}, 3398 (1976).

\bibitem{Roy84}
S. M. Roy and V. Singh, 
\textit{Fractional total-charge eigenvalues for a fermion in a finite one-dimensional box},
Phys. Lett. B \textbf{143}, 179 (1984).

\bibitem{Harper}
P. G. Harper, 
\textit{The General Motion of Conduction Electrons in a Uniform Magnetic Field, with Application to the Diamagnetism of Metals}, 
Proc. Phys. Soc. A {\bf 68}, 874 (1955).

\bibitem{Hofstadter}
D. R. Hofstadter, 
\textit{Energy levels and wave functions of Bloch electrons in rational and irrational magnetic fields}, 
Phys. Rev. B {\bf 14}, 2239 (1976).

\bibitem{Lerma19}
S. Lerma-Hern\'andez, J. Dukelsky, and G. Ortiz, 
\textit{Integrable model of a $p$-wave bosonic superfluid},
Phys. Rev. Res. {\bf 1}, 032021(R) (2019). 



\end{thebibliography}
\end{document}